\documentclass[12pt, a4paper]{article}
\pdfoutput=1
\usepackage{jheppub}
\usepackage{subcaption}
\usepackage[T1]{fontenc}
\usepackage{float}
\usepackage{latexsym}
\usepackage{amsmath,amstext,amssymb,amsfonts,
amscd,bm,array,multirow,amsbsy,mathrsfs,amsthm,calc}
\usepackage{t1enc}
\usepackage{indentfirst}
\usepackage{pb-diagram}
\usepackage{graphicx}
\usepackage{enumerate}
\usepackage{color}
\usepackage[all]{xy}
\usepackage{hyperref}
\hypersetup{colorlinks=false,pdfborderstyle={/S/U/W 0}}
\usepackage[usenames,dvipsnames]{xcolor}
\usepackage{url}
\usepackage{upgreek}

\theoremstyle{plain}
\newtheorem{thm}{Theorem}[section]             

\newtheorem{prop}[thm]{Proposition}

\theoremstyle{definition}
\newtheorem{definition}[thm]{Definition}

\theoremstyle{remark}
\newtheorem{remark}[thm]{Remark}

\newcommand{\be}{\begin{equation*}}
\newcommand{\ee}{\end{equation*}}
\newcommand{\ben}{\begin{equation}}
\newcommand{\een}{\end{equation}}
\newcommand{\beqa}{\begin{eqnarray*}}
\newcommand{\eeqa}{\end{eqnarray*}}
\newcommand{\beqan}{\begin{eqnarray}}
\newcommand{\eeqan}{\end{eqnarray}}
\newcommand{\nn}{\nonumber}

\def\i{\mathbf{i}}
\def \const{\mathrm{const}}
\def\bgamma{\boldsymbol{gamma}}
\def\bGamma{\boldsymbol{\Gamma}}

\def\Z{\mathbb{Z}}

\def\R{\mathbb{R}}

\def\End{\mathrm{End}}
\def\Aut{\mathrm{Aut}}
\def\Isom{\mathrm{Isom}}
\def\Diff{\mathrm{Diff}}
\def\Hess{\mathrm{Hess}}
\def\wHess{\widehat{\Hess}}

\def\Crit{\mathrm{Crit}}

\def\hV{\widehat{V}}

\def\const{\mathrm{const}}
\def\Ends{\mathrm{Ends}}

\def\fM{\mathfrak{M}}

\def\Card{\mathrm{Card}}

\def\rT{\mathrm{T}}
\def\rK{\mathrm{K}}

\newcommand{\pd}{\partial}
\def\dd{\mathrm{d}}

\newcommand{\Tr}{\mathrm{Tr}}

\newcommand{\sign}{\mathrm{sign}}

\def\cA{\mathcal{A}}

\def\cC{\mathcal{C}}

\def\cF{\mathcal{F}}
\def\cG{\mathcal{G}}

\def\cL{\mathcal{L}}

\def\cS{\mathcal{S}}

\def\cU{\mathcal{U}}

\def\tc{{\tilde c}}

\def\cA{\mathcal{A}}

\def\rmax{\mathrm{max}}

\def\mD{\mathbb{D}}

\def\rD{\mathrm{D}}
\def\rA{\mathrm{A}}

\def\rS{\mathrm{S}}

\def\fM{\mathfrak{M}}

\def\e{{\bf e}}

\def\c{{\bf c}}

\def\SO{\mathrm{SO}}

\newcommand{\eqdef}{\stackrel{{\rm def.}}{=}}

\DeclareMathOperator{\arctanh}{arctanh}

\def\grad{\mathrm{grad}}
\def\Sym{\mathrm{Sym}}
\def\End{\mathrm{End}}

\def\Re{\mathrm{Re}}
\def\Im{\mathrm{Im}}
\def\im{\mathrm{im}}
\def\O{\mathrm{O}}

\def\bgamma{\boldsymbol{\gamma}}
\def\alim{\mathrm{lim}_\alpha}
\def\olim{\mathrm{lim}_\omega}

\def\hSigma{\widehat{\Sigma}}
\def\hV{\widehat{V}}

\def\hPhi{{\hat \Phi}}

\def\bvert{\big{\vert}}

\def\ind{\mathrm{ind}}
\def\rM{\mathrm{M}}

\newcommand{\fourpartdef}[8]
{
	\left\{
	\begin{array}{ll}
		#1 & \mbox{if} #2 \\
		#3 & \mbox{if} #4 \\
		#5 & \mbox{if} #6\\
                #7 & \mbox{if} #8
	\end{array}
	\right.
}

\newcommand{\twopartdef}[4]
{
	\left\{
	\begin{array}{ll}
		#1 & \mbox{if } #2 \\
		#3 & \mbox{if } #4
	\end{array}
	\right.
}

\newcommand{\threepartdef}[6]
{
	\left\{
	\begin{array}{ll}
		#1 & \mbox{if} #2 \\
		#3 & \mbox{if} #4 \\
		#5 & \mbox{if} #6
	\end{array}
	\right.
}

\def\vol{\mathrm{vol}}
\def\P{\mathbb{P}}
\def\cR{\mathcal{R}}

\def \Grad{\mathrm{Grad}}

\title{The infrared behavior of tame two-field cosmological models}

\author{Elena Mirela Babalic, Calin Iuliu Lazaroiu}

\affiliation{Horia Hulubei National
  Institute of Physics and Nuclear Engineering,\\
  Reactorului 30, Bucharest-Magurele, 077125, Romania\\ 
  }

\emailAdd{mbabalic@theory.nipne.ro, lcalin@theory.nipne.ro}

\abstract{We study the first order infared behavior of tame
hyperbolizable two-field cosmological models, defined as those
classical two-field models whose scalar manifold is a connected,
oriented and topologically finite hyperbolizable Riemann surface
$(\Sigma,\cG)$ and whose scalar potential $\Phi$ admits a positive and
Morse extension to the end compactification of $\Sigma$. We achieve
this by determining the universal forms of the asymptotic gradient
flow of the classical effective potential $V$ with respect to the
uniformizing metric $G$ near all interior critical points and ends of
$\Sigma$, finding that some of the latter act like fictitious but
exotic stationary points of the gradient flow. We also compare these
results with numerical studies of cosmological orbits.  For critical
cusp ends, we find that cosmological curves have transient
quasiperiodic behavior but are eventually attracted or repelled by the
cusp along principal geodesic orbits determined by the extended
effective potential. This behavior is approximated in the infrared by
that of gradient flow curves near the cusp.}

\begin{document}

\maketitle

\pagebreak

\section*{Introduction}

Two-field cosmological models provide the simplest testing ground for
multifield cosmological dynamics. Such models are important for
connecting cosmology with fundamental theories of gravity and matter,
since the effective description of the generic string or M-theory
compactification contains many moduli fields. In particular,
multifield models are crucial in cosmological applications of the
swampland program \cite{V, OV, BCV, BCMV}, as pointed out for example
in \cite{AP, OOSV,GK}. They may also afford a unified description
of inflation, dark matter and dark energy \cite{AL}.

A two-field cosmological model is parameterized by the rescaled Planck
mass $M_0\eqdef M\sqrt{\frac{2}{3}}$ (where $M$ is the reduced Planck
mass) and by its {\em scalar triple} $(\Sigma,\cG,\Phi)$, where the
generally non-compact borderless connected surface $\Sigma$ is the target
manifold for the scalar fields, $\cG$ is the scalar field
metric and $\Phi$ is the scalar potential. To ensure conservation of
energy, one requires that $\cG$ is complete; we
also assume that $\Phi$ is strictly positive. In
\cite{ren}, we used a dynamical RG flow analysis and the
uniformization theorem of Poincar\'e to show that two-field models
whose scalar field metric has constant Gaussian curvature $K$ equal to
$-1$, $0$ or $+1$ give distinguished representatives for the IR
universality classes of all two-field cosmological models. More
precisely, the first order IR approximants of cosmological orbits for
the model parameterized by $(M_0,\Sigma,\cG,\Phi)$ coincide with those
of the model parameterized by $(M_0,\Sigma,G,\Phi)$, where $G$ is the
uniformizing metric of $\cG$. Moreover, these approximants coincide
with the gradient flow orbits of $(\Sigma,G,V)$, where $V\eqdef
M_0\sqrt{2\Phi}$ is the {\em classical effective potential} of the
model. In particular, IR universality classes depend only on the
scalar triple $(\Sigma,G,V)$. This result allows for
systematic studies of two-field cosmological models belonging to a
fixed IR universality class by using the infrared expansion of
cosmological curves outlined in \cite{ren}. The reduction to {\em
uniformized models}, defined as those whose scalar field metric has
Gaussian curvature $K$ equal to $-1$, $0$ or $+1$ serves as an organizing
principle for the infrared expansion, the first order of which is
captured by the gradient flow of $(\Sigma,G,V)$.

The case $K=-1$ is generic and obtains when the topology of $\Sigma$
is of {\em general type}; for such models, the uniformizing metric is
hyperbolic. The few exceptions to this situation arise when $\Sigma$
is of {\em special type}, namely diffeomorphic to $\R^2$, the
two-sphere $\rS^2$, the real projective plane $\R\P^2$, the two torus
$\rT^2$, the open Klein bottle $\rK^2=\R\P^2\times \R\P^2\simeq
\rT^2/\Z_2$, the open annulus $\rA^2$ or the open M\"{o}bius strip
$\rM^2\simeq \rA^2/\Z_2$. When $\Sigma$ is diffeomorphic to $\rS^2$ or
$\R\P^2$, the uniformizing metric has Gaussian curvature $+1$,
while when it is diffeomorphic to a torus or Klein bottle the
uniformizing metric is flat and complete. When $\Sigma$ is of {\em
exceptional type}, i.e. diffeomorphic to $\R^2$, $\rA^2$ or $\rM^2$,
the metric $\cG$ uniformizes to a complete flat metric or to a
hyperbolic metric depending on its conformal class\footnote{A
hyperbolic metric on an exceptional surface is conformally flat but
{\em not} conformally equivalent to a {\em complete} flat
metric.}. The cosmological model, its scalar field metric $\cG$ and
the conformal class of the latter are called {\em hyperbolizable} when
$\cG$ uniformizes to a hyperbolic metric. Thus hyperbolizable models
comprise all two-field models whose target is of general type as well
as those models whose target is exceptional (i.e diffeomorphic with
$\R^2$, $\rA^2$ or $\rM^2$) and for which $\cG$
belongs to a hyperbolizable conformal class. The uniformized form
$(M_0,\Sigma,G,\Phi)$ of a hyperbolizable model is a {\em two field
generalized $\alpha$-attractor model} in the sense of
\cite{genalpha}. Some aspects of such models were investigated
previously in \cite{elem,modular, Noether1, Noether2, Hesse,
Lilia1,Lilia2} (see \cite{unif,Nis,Tim19, LiliaRev} for brief
reviews).

In this paper, we study the infrared behavior of hyperbolizable
two-field models with certain technical assumptions on the topology of
$\Sigma$ and on the scalar potential $\Phi$. Namely, we assume that
$\Sigma$ is oriented and {\em topologically finite} in the sense
that it has finitely-generated fundamental group. When $\Sigma$ is
non-compact, this condition insures that it has a finite number of
Freudenthal ends \cite{Freudenthal1,Freudenthal2,Freudenthal3} and
that its end (a.k.a. Kerekjarto-Stoilow
\cite{Kerekjarto,Stoilow,Richards}) compactification $\hSigma$ is a
smooth and oriented compact surface. Thus $\Sigma$ is recovered from
$\hSigma$ by removing a finite number of points.  We also assume that
$\Phi$ admits a smooth extension $\hPhi$ to $\hSigma$ which is a
strictly-positive Morse function defined on $\hSigma$. A two-field
cosmological model is called {\em tame} when these conditions are
satisfied.

To first order in the scale expansion of \cite{ren}, the IR limit of a
tame two-field model is given by the gradient flow of the classical
effective potential $V=M_0\sqrt{2\Phi}$ on the geometrically finite
hyperbolic surface $(\Sigma,G)$. Since the future limit points of
cosmological curves and of the gradient flow curves of $(\Sigma,G,V)$ are
critical points of $\Phi$ or Freudenthal ends of $\Sigma$, the
asymptotic behavior of such curves for late cosmological times is
determined by the form of $G$ and $V$ near such points. The form of
$G$ near critical points follows from the fact that any hyperbolic
surface is locally isometric with a domain of the Poincar\'e disk,
while that near each end follows from the uniformization of
geometrically finite hyperbolic surfaces. Since the Morse assumption
on the extended potential determines its asymptotic form near the
points of interest, this allows us to derive closed form expressions
for the asymptotic gradient flow and hence to describe the infrared
phases of such models in the sense of \cite{ren}. In particular, we
find that the asymptotic gradient flow of $(\Sigma,G,V)$ near each end
which is a critical point of the extended potential can be expressed
using the incomplete gamma function of order two and certain constants
which depend on the type of end under consideration and on the
(appropriately-defined) principal values of the extended effective
potential $\hV$ at that end. We also find that flaring ends which are
not critical points of $\hV$ act like fictitious but non-standard
stationary points of the effective gradient flow. While the local
form near the critical points of $V$ is standard (since they are
hyperbolic stationary points \cite{Palis,Katok} of the cosmological
and gradient flow), the asymptotic behavior near Freudenthal ends is
exotic in that some of the ends act like fictitious stationary points
with unusual characteristics. For example, the stable and unstable
manifolds of an end under the gradient flow of $(\Sigma,G,V)$ can have
dimensions which differ from those of hyperbolic stationary points of
dynamical systems.

We compare these results with numerical computations of cosmological
curves near the points of interest. We find particularly interesting
behavior near cusp ends, around which generic cosmological
trajectories tend to spiral a large number of times before either
``falling into the cusp'' or being ``repelled'' back toward the
compact core of $\Sigma$ along principal geodesic orbits determined by
$V$.  In particular, cusp ends lead naturally to ``fast turn''
behavior of cosmological curves, a phenomenon which we already
illustrated in our previous analysis of the hyperbolic triply
punctured sphere (see \cite{modular}).

The paper is organized as follows. In Section \ref{sec:models}, we
briefly recall the global description of multifield cosmological
models through a second order geometric ODE and their first order
infrared approximation introduced in \cite{ren}. Section
\ref{sec:tame2field} defines tame two-field cosmological models,
describes the critical points of their extended potential and
discusses principal coordinates centered at ends. In the same section,
we recall the form of the hyperbolic metric $G$ in a canonical
vicinity of an end and extract its asymptotic behavior near each type
of end. Section \ref{sec:int} discusses the behavior of cosmological
curves and their first order IR approximants near interior critical
points. Section \ref{sec:noncritends} performs the asymptotic analysis
of gradient flow curves and compares it with numerical results for
cosmological curves near those ends of $\Sigma$ which are noncritical
for the extended scalar potential, while Section \ref{sec:critends}
performs the same analysis for critical ends. Section
\ref{sec:Conclusions} presents our conclusions and some directions for
further research. The appendix gives some details of the computation
of cosmological curves near interior critical points and near
Freudenthal ends.

\paragraph{Notations and conventions.}

All surfaces $\Sigma$ considered in this paper are connected,
smooth, Hausdorff and paracompact. If $V$ is a smooth real-valued
function defined on $\Sigma$, we denote by:
\be
\Crit V\eqdef \{c\in \Sigma| (\dd V)(c)=0\}
\ee
the set of its critical points. For any $c\in \Crit V$, we denote
by $\Hess(V)(c)\in \Sym^2(T^\ast_c\Sigma)$ the Hessian of $V$ at $c$,
which is a well-defined and coordinate independent symmetric bilinear
form on the tangent space $T_c\Sigma$. Given a metric $\cG$ on
$\Sigma$, we denote by:
\be
\Hess_\cG(V)\eqdef \nabla\dd V\in \Gamma(\Sigma,\Sym^2(T^\ast \Sigma))
\ee
the covariant Hessian tensor of $V$ relative to $\cG$, where
$\nabla$ is the Levi-Civita connection of $\cG$. This symmetric tensor
has the following local expression in coordinates
$(x^1,x^2)$ on $\Sigma$:
\be
\Hess_\cG(V)=(\pd_i\pd_j-\Gamma^k_{ij}(x)\pd_k)V \dd x^i\otimes \dd
x^j~~,
\ee
where $\Gamma^k_{ij}(x)$ are the Christoffel symbols of $\cG$. For any
critical point $c\in \Crit V$, we have $\Hess_\cG(V)(c)=\Hess(V)(c)$.
Recall that a critical point $c$ of $V$ is called {\em nondegenerate}
if $\Hess(V)(c)$ is a non-degenerate bilinear form. When $V$ is a
Morse function (i.e. has only non-degenerate critical points),
the set $\Crit V$ is discrete. 

We denote by $\hSigma$ the Freudenthal (a.k.a. end) compactification
of $\Sigma$, which is a compact Hausdorff topological space containing
$\Sigma$ (see \cite{Freudenthal1,Freudenthal2,Freudenthal3}). We say
that $\Sigma$ is {\em topologically finite} if its fundamental group
is finitely generated.  In this case, $\Sigma$ has a finite number of
Freudenthal ends and $\hSigma$ is a smooth compact surface. In this
situation, we say that $V$ is {\em globally well-behaved} on $\Sigma$
if it admits a smooth extension $\hV$ to $\hSigma$. A metric $\cG$ on
$\Sigma$ is called {\em hyperbolic} if it is complete and of constant
Gaussian curvature equal to $-1$.

\section{Two-field cosmological models and their IR approximants}
\label{sec:models}

Recall that a {\em two-field cosmological model} is a classical
cosmological model with two scalar fields derived from the following action 
on a spacetime with topology $\R^4$:
\ben
\label{S}
S[g,\varphi]=\int \vol_g \cL[g,\varphi]~~,
\een
where:
\ben
\label{cL}
\cL[g,\varphi]=\frac{M^2}{2} \mathrm{R}(g)-\frac{1}{2}\Tr_g \varphi^\ast(\cG)-\Phi\circ \varphi~~.
\een
Here $M$ is the reduced Planck mass, $g$ is the spacetime metric on
$\R^4$ (taken to be of ``mostly plus'') signature, while $\vol_g$ and
$\mathrm{R}(g)$ are the volume form and Ricci scalar of $g$. The
scalar fields are described by a smooth map $\varphi:\R^4\rightarrow
\Sigma$, where $\Sigma$ is a (generally non-compact) smooth and
connected paracompact surface without boundary which is endowed with a
smooth Riemannian metric $\cG$, while $\Phi:\Sigma\rightarrow \R$ is a
smooth function which plays the role of potential for the scalar
fields. We require that $\cG$ is complete to ensure conservation of
energy. For simplicity, we also assume that $\Phi$ is strictly
positive on $\Sigma$. Notice that the model is parameterized by the
quadruplet $\fM\eqdef(M_0,\Sigma,\cG,\Phi)$, where:
\be
M_0\eqdef M\sqrt{\frac{2}{3}}
\ee
is the {\em rescaled Planck mass}.

\subsection{The cosmological equation}

\noindent The two-field model parameterized by
$\fM$ is obtained by assuming that $g$ is an FLRW
metric with flat spatial section:
\ben
\label{FLRW}
\dd s^2_g=-\dd t^2+a(t)^2\sum_{i=1}^3 \dd x_i^2
\een
(where $a(t)>0$) and that $\varphi$ depends only on the {\em
cosmological time} $t\eqdef x^0$. One sets:
\be
H(t)\eqdef \frac{\dot{a}(t)}{a(t)}~~,
\ee
where the dot indicates derivation with respect to $t$. When $H>0$
(which we assume throughout), the variational equations of \eqref{S}
reduce to the {\em cosmological equation}:
\ben
\label{eomsingle}
\nabla_t \dot{\varphi}(t)+\frac{1}{M_0} \left[||\dot{\varphi}(t)||_\cG^2+2\Phi(\varphi(t))\right]^{1/2}\dot{\varphi}(t)+ (\grad_{\cG} \Phi)(\varphi(t))=0
\een
together with the condition:
\ben
\label{Hcond}
H(t)=H_\varphi(t)~~,
\een
where the {\em Hubble parameter} of $\varphi$ is defined through:
\ben
\label{Hvarphi}
H_\varphi(t)\eqdef \frac{1}{3 M_0}\left[||\dot{\varphi}(t)||_\cG^2+2\Phi(\varphi(t))\right]^{1/2}~~.
\een
Here $\nabla_t\eqdef \nabla_{\dot{\varphi}(t)}$ is the covariant
derivative with respect to the tangent vector $\dot{\varphi}(t)\in
T_{\varphi(t)}\Sigma$, which takes the following form in local
coordinates on $\Sigma$:
\be
\nabla_t  \dot{\varphi}^i(t) =\ddot{\varphi}^i(t)+\Gamma^i_{jk}(\varphi(t))\dot\varphi^j(t)\dot\varphi^k(t)~~,
\ee
where $\Gamma^i_{jk}$ are the Christoffel symbols of $\cG$.  The
solutions $\varphi:I\rightarrow \Sigma$ of \eqref{eomsingle} (where
$I$ is a non-degenerate interval) are called {\em cosmological
curves}, while their images in $\Sigma$ are called {\em cosmological
orbits}. Given a cosmological curve $\varphi$, relation \eqref{Hcond}
determines $a$ up to a multiplicative constant. The cosmological
equation can be reduced to first order by passing to the tangent
bundle of $T\Sigma$ (see \cite{SLK}). More precisely,
\eqref{eomsingle} is equivalent with the integral curve equation of a
semispray (a.k.a. second order vector field) $S$ defined on $T\Sigma$ which
is called the {\em cosmological semispray} of the model (see
\cite{ren}). The flow of this vector field on the total space of
$T\Sigma$ is called the {\em cosmological flow}.

\begin{remark}
The cosmological equation can be written as:
\be
\nabla_t \dot{\varphi}(t)+\left[||\dot{\varphi}(t)||_{\cG_0}^2+2\Phi_0(\varphi(t))\right]^{1/2}\dot{\varphi}(t)+ (\grad_{\cG_0} \Phi_0)(\varphi(t))=0~~,
\ee
where we defined the {\em rescaled scalar field metric} and {\em
rescaled scalar potential} by:
\be
\cG_0\eqdef \frac{1}{M_0^2}\cG~~\mathrm{and}~~\Phi_0\eqdef\ \frac{1}{M_0^2}\Phi~~.
\ee
Moreover, \eqref{Hvarphi} reads:
\be
H_\varphi(t)=\frac{1}{3}\left[||\dot{\varphi}(t)||_{\cG_0}^2+2\Phi_0(\varphi(t))\right]^{1/2}~~.
\ee
Hence the cosmological curves and their Hubble parameters depend only
on the {\em rescaled scalar triple} $(\Sigma,\cG_0,\Phi_0)$.
\end{remark}

\subsection{Uniformized models and first order IR approximants}

\noindent Consider a cosmological curve $\varphi:I\rightarrow \Sigma$ of the model
parameterized by $(M_0,\Sigma,\cG,\Phi)$, where we can assume that $0\in
I$ by shifting the cosmological time since the cosmological equation
is autonomous. Define the {\em classical effective potential} $V$ of the
model by:
\be
V\eqdef M_0\sqrt{2\Phi}~~.
\ee
The dynamical RG flow analysis of \cite{ren} shows that the first
order IR approximant of $\varphi$ is the gradient flow curve $\eta$ of
the scalar triple $(\Sigma,\cG,V)$ which satisfies the initial condition:
\ben
\label{incond}
\eta(0)=\varphi(0)~~.
\een

Since the gradient flow of $V$ is invariant under Weyl transformations
of $\cG$ up to increasing reparameterization of the gradient flow
curves, the uniformization theorem of Poincar\'e allows us to replace
$\cG$ with its uniformizing metric $G$ without changing the oriented
gradient flow orbits. Hence the first order IR approximation of
cosmological flow {\em orbits} is given by the gradient flow orbits of
the hyperbolic scalar triple $(\Sigma,G,V)$. Thus the
original model parameterized by $(M_0,\Sigma,\cG,\Phi)$ and the
{\em uniformized model} parameterized by $(M_0,\Sigma,G,\Phi)$
have the same first order IR orbits. The uniformized
model provides a distinguished representative of the IR universality
class of the original model as defined in \cite{ren}. Moreover, this
universality class depends only on the scalar triple
$(\Sigma,G,V)$. Notice that the initial condition \eqref{incond} for
first order IR approximants is invariant under reparameterizations
since both the cosmological and gradient flow equations are autonomous
and we can shift parameters to ensure that \eqref{incond} does not
change. From now on, we work exclusively with the uniformized model
$(M_0,\Sigma,G,\Phi)$ and its scalar triple $(\Sigma, G,V)$, whose gradient
flow we call the {\em effective gradient flow}. 

Since the gradient flow equation of $(\Sigma,G,V)$:
\be
\dot{\eta}(t)=-(\grad_G V)(\eta(t))
\ee
is a first order ODE, the degree of the cosmological equation drops by
one in the first order IR approximation. As a result, the tangent
vector to $\eta$ is constrained to lie within the {\em gradient
flow shell} $\Grad_G V$ of $(\Sigma,G,V)$. The latter is the closed
submanifold of $T\Sigma$ defined as the graph of the vector field
$-\grad_G V$:
\be
\Grad_G V\eqdef \{u\in T\Sigma~\vert~u=-(\grad_G V)(\pi(u))\}~~,
\ee
where $\pi:T\Sigma\rightarrow \Sigma$ is the tangent bundle
projection. In particular, the cosmological flow of the uniformized
model (which is defined on $T\Sigma$) becomes confined to $\Grad_G V$
in this approximation. In this order of the IR expansion, the tangent
vector $\dot{\eta}(0)$ is constrained to equal $-(\grad_G V)(\eta(0))$
and cannot be specified independently; one has to consider higher
orders of the expansion to obtain an approximant of the cosmological
flow which is defined on the entirety of $T\Sigma$. In particular, the
IR approximation is rather coarse.

A cosmological curve $\varphi$ is called {\em infrared optimal}
if its speed at $t=0$ lies in gradient flow shell of $(\Sigma,G,V)$,
i.e. if $\varphi$ satisfies the condition:
\be
\dot{\varphi}(0)=-(\grad_G V)(\varphi(0))~~.
\ee
The first order IR approximant $\eta$ of an infrared optimal
cosmological curve osculates in first order to $\varphi$ at $t=0$.
Thus $\eta(t)$ is a {\em first order} asymptotic approximant of
$\varphi(t)$ for $|t|\ll 1$. The covariant acceleration
$\nabla_t\dot{\varphi}(0)$ of $\varphi$ at $t=0$ need not agree with
that of $\eta$ (which is determined by the gradient flow equation). As
a consequence, $\varphi(t)$ and $\eta(t)$ can differ already to the
second order in $t$. The two curves osculate in second order at $t=0$
only if $G$ and $\Phi$ satisfy a certain condition at the initial
point $\varphi(0)$ (see \cite{ren}). In particular, the IR
approximation of an infrared optimal cosmological curve can be
expected to be accurate only for sufficiently small cosmological
times. For cosmological curves which are not IR optimal, the approximation can be
accurate only when the speed of the curve at $t=0$ is sufficiently
close to the gradient flow shell of $(\Sigma,\cG,V)$. Despite these
limitations, the first order IR approximation provides an important
conceptual tool for classifying multifield cosmological models into IR
universality classes and gives a useful picture of the low frequency
behavior of cosmological curves (see \cite{ren}).

\begin{remark}
\label{rem:rescaledgradshell}
The transformation $G\rightarrow G_0=\frac{1}{M_0^2}G$,
$\Phi\rightarrow \Phi_0=\frac{1}{M_0^2}\Phi$ replaces
$(\Sigma,G,V)$ with  $(\Sigma,G_0,V_0)$, where
$V_0=M_0\sqrt{2\Phi_0}=\sqrt{2\Phi}=\frac{1}{M_0}V$. Since
$\grad_{G_0} V_0=M_0 \grad_G V$, the gradient flow shell of
$(\Sigma_0,G_0,V_0)$ differs from that of $(\Sigma,G,V)$ by a constant
rescaling in the fiber directions. This can be absorbed by a constant
reparameterization of the gradient flow curves and hence does not
affect the gradient flow orbits.
\end{remark}

\section{Hyperbolizable tame two-field models}
\label{sec:tame2field}

The results of \cite{ren} allow us to describe the infrared behavior
of two-field models under certain assumptions on the scalar manifold
and potential. Throughout this section, we consider a hyperbolizable
model parameterized by $(M_0,\Sigma,\cG,\Phi)$ and let $G$ be the
hyperbolization of $\cG$ and $V\eqdef M_0\sqrt{2\Phi}$.

\subsection{The tameness conditions}

\noindent Recall that adding the Freudenthal ends to $\Sigma$
produces its {\em end compactification} $\hSigma$, where
each point of the set:
\be
\Ends(\Sigma)\eqdef \hSigma\setminus \Sigma
\ee
corresponds to an end. When endowed with its natural topology,
$\hSigma$ is the classical compactification of surfaces considered by
Kerekjarto and Stoilow \cite{Kerekjarto,Stoilow}, which was clarified
further and extended to the unoriented case by Richards
\cite{Richards}; it coincides with Freudenthal's end compactification
of manifolds for the case of dimension two. In general, the set of
ends $\Ends(\Sigma)$ can be infinite and rather complicated (it is a
totally disconnected space which can be a Cantor space). Moreover, the
scalar potential $\Phi$ (and the effective potential $V$) can have
complicated asymptotic behavior near each end; in particular, they may
fail to extend to smooth functions on $\hSigma$. Furthermore, $\Phi$
(and thus $V$) may have non-isolated critical points on $\Sigma$. To
obtain a tractable set of models, we make the following

\paragraph{\bf Assumptions.}
\begin{enumerate}
\item $\Sigma$ is oriented and {\em topologically finite} in the sense
that its fundamental group $\pi_1(\Sigma)$ is finitely-generated. This
implies that $\Sigma$ has finite genus and a finite number of ends and
that its end compactification $\hSigma$ is a compact smooth
surface. Notice that $(\Sigma,G)$ need not have finite area.
\item The scalar potential $\Phi$ is {\em globally well-behaved},
i.e. $\Phi$ admits a smooth extension $\hPhi$ to $\hSigma$. We require
that $\hPhi$ is strictly positive on $\hSigma$, which means that the
limit of $\Phi$ at each end of $\Sigma$ is a {\em strictly} positive
number.
\item The extended potential $\hPhi$ is a Morse function on $\hSigma$
(in particular, $\Phi$ is a Morse function on $\Sigma$).
\end{enumerate}

\begin{definition}
A hyperbolic two-dimensional scalar triple $(\Sigma,\cG,\Phi)$ is called {\em
tame} if it satisfies conditions $1$, $2$ and $3$ above. A two-field
cosmological model with tame scalar triple is called tame.
\end{definition}

\begin{remark} It may seem at first sight that our tameness
assumptions could reduce the study of the IR behavior of
cosmological curves to an application of known results from Morse
theory \cite{MB,Bott} and from the theory of gradient flows. However
this is {\em not} the case because the metrics $\cG$ and $G$ do not
extend to the end compactification of $\Sigma$ and because the vector
field $\grad_G V$ is singular at the ends. On the other hand, the flow
of $\grad_G V$ on the non-compact surface $\Sigma$ is not amenable to
ordinary Morse theory, which assumes a compact manifold. One might
hope that some version of Morse theory on manifolds with boundary (see
\cite{KM,Laudenbach,Akaho}) could apply to the conformal
compactification of $(\Sigma,G)$. However, the assumptions made in
common versions of that theory are not satisfied in our case. As
already shown in \cite{genalpha}, Morse theoretic results are
nevertheless useful for relating the indices of the critical points of
$\hPhi$ to the topology of $\Sigma$ -- a relation which can be used in
principle to constrain the topology of $\Sigma$ using cosmological
observations.
\end{remark}

\subsection{Interior critical points. Critical and  noncritical ends}

\noindent The assumption that $\Sigma$ is topologically finite implies
that the set of ends $\Ends(\Sigma)$ is finite, while the assumption
that $\hPhi$ is Morse constrains the asymptotic behavior of $\Phi$ at
the ends of $\Sigma$. Notice that the extended potential is uniquely
determined by $\Phi$, since continuity of $\hPhi$ implies:
\be
\hPhi(\e)=\lim_{\Sigma\ni m\rightarrow \e}\Phi(m)~~\forall \e\in \Ends(\Sigma)~~.
\ee
Also notice that $\hPhi$ (hence also $\Phi$) is bounded since it is
continuous while $\hSigma$ is compact. The condition that $\hPhi$ is
Morse implies that its critical points are isolated. Since $\hSigma$
is compact, it follows that the set:
\be
\Crit\hPhi\eqdef\{c\in \hSigma\vert (\dd\hPhi)(c)=0\}
\ee
is finite. Since we assume that $\hPhi$ is strictly
positive on $\hSigma$, the classical effective potential
$V=M_0\sqrt{2\Phi}$ is also globally well-behaved, i.e.  admits a
smooth extension $\hV$ to $\hSigma$, which is given by:
\be
\hV=M_0\sqrt{2\hPhi}~~.
\ee
Moreover, $\hV$ has the same critical points\footnote{Indeed, we have
$\dd \hV=M_0\frac{\dd \hPhi}{\sqrt{2\hPhi}}$.}  as $\hPhi$:
\be
\Crit\hV=\Crit \hPhi~~.
\ee
Since $V$ is the restriction of $\hV$ to $\Sigma$, the critical points
of $V$ coincide with the {\em interior critical points} of $\hV$ (and
$\hPhi$), i.e. those critical points which lie on $\Sigma$:
\be
\Crit V=\Crit \Phi =\Sigma\cap \Crit \hV=\Sigma\cap\Crit \hPhi~~.
\ee
Let:
\be
\Crit_\infty V =\Crit_\infty \Phi\eqdef \Ends(\Sigma)\cap\Crit \hV=\Ends(\Sigma)\cap\Crit \hPhi
\ee
be the set of {\em critical ends} (or ``critical points at
infinity''), i.e.  those critical points of the extended potential
which are also ends of $\Sigma$. We have the disjoint union
decomposition:
\be
\Crit \hV =\Crit V \sqcup \Crit_\infty V ~~.
\ee
Finally, an end of $\Sigma$ which is not a critical point of $\hPhi$
(and hence of $\hV$) will be called a {\em noncritical end}. Such
ends form the set $\Ends(\Sigma)\setminus \Crit_\infty(V)$. We denote
interior critical points by $\c$ and arbitrary critical points of
$\hV$ by $c$; the latter can be interior critical points or critical
ends. Finally, we denote by $\e$ the ends of $\Sigma$.  To describe
the early and late time behavior of the gradient flow of
$(\Sigma,G,V)$, we must study the asymptotic form of this flow near
the interior critical points as well as near all ends of $\Sigma$.

\subsection{Stable and unstable manifolds under the effective gradient flow}

\noindent For any maximal gradient flow curve
$\eta:(a_-,a_+)\rightarrow\Sigma$ of $(\Sigma,G,V)$, we denote by:
\be
\alim\eta\eqdef \lim_{t\rightarrow a_-}\eta(t)\in \hSigma~~\mathrm{and}~~
\olim\eta\eqdef \lim_{t\rightarrow a_+}\eta(t)\in \hSigma
\ee
its $\alpha$- and $\omega$- limits as a curve in $\hSigma$. Each
of these points is either an interior critical point or an end of
$\Sigma$.

For any $m\in \Sigma$, let $\eta_m$ be the maximal gradient flow curve
of $(\Sigma,G,V)$ which satisfies $\eta_m(0)=m$. Recall that the stable and
unstable manifolds of an interior critical point $\c\in\Crit V$ under
the gradient flow of $(\Sigma,G,V)$ are defined
through:
\ben
\label{SUdefc} \cS(\c)\eqdef \{m\in
\Sigma|\olim\eta_m(q)=\c\}~~,~~\cU(\c)\eqdef\{m\in \Sigma|\alim
\eta_m(q)=\c\}~~.
\een
By analogy, we define the {\em stable and unstable
  manifolds of an end} $\e$ by:
\ben
\label{SUdefe} \cS(\e)\eqdef \{m\in
\Sigma|\olim\eta_m(q)=\e\}~~,~~\cU(\e)\eqdef\{m\in
\Sigma|\alim\eta_m(q)=\e\}~~.
\een
Notice that the stable and unstable manifolds of an end are subsets of $\Sigma$.
 
\subsection{The form of $G$ in a vicinity of an end}

\noindent Recall (see \cite{Borthwick} or \cite[Appendix
D.6]{genalpha}) that any end $\e$ of a geometrically-finite and
oriented hyperbolic surface $(\Sigma,G)$ admits an open neighborhood
$U_\e\subset \hSigma$ (which is diffeomorphic with a disk) such that
the hyperbolic metric $G$ takes a canonical form when restricted to
$\dot{U}_\e\eqdef U_\e\setminus \{\e\}\subset \Sigma$. More precisely,
there exist {\em semigeodesic polar coordinates} $(r,\theta)\in
\R_{>0}\times \rS^1$ defined on $\dot{U}_\e$ in which the hyperbolic
metric has the form:
\ben
\label{emetric}
\dd s_G^2|_{\dot{U}_\e}=\dd r^2+f_\e(r)\dd \theta^2~~,
\een
where: 
\ben
\label{fe}
f_\e(r)= \fourpartdef{\sinh^2(r)}{~\e=\mathrm{plane~end}}
{\frac{1}{(2\pi)^2}e^{2r}}{~\e=\mathrm{horn~end}}
{\frac{\ell^2}{(2\pi)^2}\cosh^2(r)}{~\e=\mathrm{funnel~end~of~circumference}~\ell>0}{\frac{1}{(2\pi)^2}e^{-2r}}{~\e=\mathrm{cusp~end}}~~.
\een
In such coordinates, the end $\e$ corresponds to $r\rightarrow
\infty$. Setting $\zeta\eqdef r e^{\i \theta}$, the corresponding {\em  semigeodesic 
Cartesian coordinates} are defined through:
\be
\zeta_1\eqdef \Re(\zeta)=r \cos \theta~~\mathrm{and}~~\zeta_2\eqdef \Im (\zeta)=r \sin\theta~~.
\ee
Plane, horn and funnel ends are called {\em flaring ends}. For such
ends, the length of horocycles in $(\Sigma,G)$ grows exponentially
when one approaches $\e$. The only non-flaring ends are cusp
ends, for which the length of horocycles tends to zero as one
approaches the cusp. An oriented geometrically-finite hyperbolic surface is
called {\em elementary} if it is isometric with the Poincar\'e disk,
the hyperbolic punctured disk or a hyperbolic annulus. A non-elementary
oriented geometrically finite hyperbolic surface admits only cusp and funnel
ends. The Poincar\'e disk has a single end which is a plane end. The
hyperbolic punctured disk has two ends, namely a cusp end and a horn
end. Finally, a hyperbolic annulus has two funnel ends.

\subsection{Canonical coordinates centered at an end}

\noindent It is convenient for what follows to consider {\em canonical
coordinates} centered at $\e$. These are defined on $U_\e$ using the
relation:
\ben
\label{inv}
z=\frac{1}{\bar{\zeta}}=\frac{1}{r}e^{\i\theta}~~.
\een
The {\em canonical polar coordinates} $(\omega,\theta)$ centered at
$\e$ are defined through:
\be
\omega\eqdef |z|=\frac{1}{r}~~,
\ee
while the {\em canonical Cartesian coordinates} $(x,y)$ centered at $\e$ are given by:
\be
x\eqdef \Re(z)=\omega\cos \theta=\frac{1}{r}\cos \theta~~~ \mathrm{and}~~~ y\eqdef \Im(z)=\omega\sin\theta=\frac{1}{r}\sin\theta~~.
\ee
In such coordinates, the end $\e$ corresponds to $\omega=0$,
i.e. $(x,y)=(0,0)$.

\subsection{Local isometries near an end}

\noindent Semigeodesic coordinates near $\e$ are not unique. In fact,
expression \eqref{emetric} is invariant under the following action of
the orthogonal group $\O(2)$:
\ben
\label{O2action1}
\left[\begin{array}{c} \zeta_1 \\ \zeta_2 \end{array}\right]\rightarrow A \left[\begin{array}{c} \zeta_1 \\ \zeta_2 \end{array}\right]~~\forall A\in \O(2)~~
\een
which thus acts by isometries of $(\dot{U}_\e,G)$. Explicitly, we have
an action $\rho_\e:\O(2)\rightarrow \Isom(\dot{U_\e},G)$ which is
defined through the conditions:
\be
\left[\begin{array}{c} \zeta_1(\rho_\e(A)(m)) \\ \zeta_2(\rho_\e(A)(m)) \end{array}\right]= A \left[\begin{array}{c} \zeta_1(m) \\ \zeta_2(m) \end{array}\right]~~\forall A\in \O(2)~~\forall m\in \dot{U}_\e~~.
\ee
The $\SO(2)$ subgroup of orientation-preserving isometries acts by
shifting $\theta$ while the axis reflections act as $\theta\rightarrow
\pi -\theta$ (i.e. $\zeta \rightarrow -\bar{\zeta}$) and
$\theta\rightarrow -\theta$ (i.e.  $\zeta\rightarrow \bar{\zeta}$). In
particular, invariance under the $\SO(2)$ subgroup implies that we can
choose the origin of $\theta$ arbitrarily. Accordingly, canonical
coordinates at $\e$ are also determined only up to this $\O(2)$
action. Since the canonical coordinates are well-defined at $\e$, the
action $\rho_\e$ extends to an action $\bar{\rho}_\e:\O(2)\rightarrow
\Diff(U_\e)$ by diffeomorphisms of $U_\e$, which is given by:
\ben
\label{O2action2}
\left[\begin{array}{c} x(\bar{\rho}_\e({\hat m})) \\ y(\bar{\rho}_\e({\hat m})) \end{array}\right]= A \left[\begin{array}{c} x({\hat m}) \\ y({\hat m}) \end{array}\right]~~
\forall A\in \O(2)~~\forall {\hat m}\in U_\e~~.
\een

\paragraph{The end representation of the local isometry action.}

Differentiating the local isometry action $\bar{\rho}_\e$ at $\e$
gives a linear representation ${\hat \rho}_\e:\O(2)\rightarrow
\Aut(T_\e\hSigma)$ of $\O(2)$ on the tangent space $T_\e\hSigma$ which
transforms the basis vectors $v_x\eqdef \pd_x\vert_\e$ and $v_y\eqdef
\pd_y\vert_\e$ as:
\ben
\label{O2action3}
\left[\begin{array}{c} {\hat \rho}_\e(A)(v_x) \\ {\hat \rho}_\e(A)(v_y) \end{array}\right]=A^t \left[\begin{array}{c} v_x \\ v_y \end{array}\right]~~\forall A\in \O(2)~~.
\een
This representation is well-defined even though the metric
\eqref{emetric} is singular at the point $\e\in \hSigma$. Let:
\ben
\label{spe}
(~,~)_\e:T_\e\hSigma\times T_\e\hSigma\rightarrow \R
\een
be any scalar product on $T_\e\hSigma$ which is invariant with respect
to this representation. Since ${\hat \rho}_\e$ is equivalent with the
fundamental representation of $\O(2)$, such a scalar product is
determined up to homothety transformations of the form:
\ben
\label{kappatf}
(~,~)_\e\rightarrow \kappa (~,~)_\e~~\forall \kappa\in \R_{>0}~~.
\een
Moreover, ${\hat \rho}_\e$ is an injective map and we have:
\ben
\label{imlambda}
\im({\hat \rho}_\e)\eqdef {\hat \rho}_\e(\O(2))=\O(T_\e\Sigma,(~,~)_\e)~~,
\een
where the right hand side is the group of linear isometries of the
Euclidean vector space $(T_\e\hSigma,(~,~)_\e)$.

\subsection{Principal values and characteristic signs at a critical end}

\noindent Suppose that $\e$ is a critical end of $(\Sigma,V)$. A
choice of $\O(2)$-invariant scalar product \eqref{spe} on $T_\e\Sigma$
allows us to define a linear operator $H^V_\e\in \End(T_\e\Sigma)$
through the relation:
\be
\Hess(\hV)(\e)(u,v)=(H^V_\e(u),v)_\e~~\forall u,v\in T_\e\hSigma~~.
\ee
Since $\Hess(\hV)(\e)$ is a symmetric bilinear form, $H_\e^V$ is a
symmetric operator in the Euclidean vector space
$(T_\e\Sigma,(~,~)_\e)$.

\begin{definition}
An orthogonal basis $(\epsilon_1,\epsilon_2)$ of
$(T_\e\hSigma,(~,~)_\e)$ is called {\em principal} for $V$ if
$\epsilon_1$ and $\epsilon_2$ are eigenvectors of $H_\e^V$ ordered
such that their eigenvalues $\mu_1(\e)$ and $\mu_2(\e)$ satisfy:
\ben
\label{ineq}
|\mu_1(\e)|\,||\epsilon_1||^2_\e\leq |\mu_2(\e)|\,||\epsilon_2||^2_\e ~~,
\een
where $||~||_\e$ is the norm defined by the scalar product $(~,~)_\e$.
\end{definition}

\noindent Given a principal orthogonal basis $(\epsilon_1,\epsilon_2)$
of $T_\e\hSigma$, we have:
\be
\Hess(\hV)(\e)(u,v)=\lambda_1(\e) u^1 v^1 +\lambda_2(\e) u^2 v^2~~, 
\ee
where $u=\sum_{i=1}^2 u^i \epsilon_i$ and $v=\sum_{i=1}^2 v^i
\epsilon_i$ are arbitrary vectors of $T_\e\hSigma$ and we defined:
\ben
\label{lambdadef}
\lambda_i(\e)\eqdef \mu_i(\e) ||\epsilon_i||_\e^2~~(i=1,2)~~.
\een
Under a rescaling \eqref{kappatf} of $(~,~)_\e$, we have:
\ben
\label{normtf}
||~||_\e\rightarrow \kappa^{1/2} ||~||_\e~~.
\een
Moreover, the operator $H_\e^V$ transforms as $H_\e^V\rightarrow
\frac{1}{\kappa} H_\e^V$. Accordingly, its eigenvalues change as:
\ben
\label{mutf}
\mu_i(\e)\rightarrow \frac{1}{\kappa}\mu_i(\e)~~(i=1,2)~~.
\een
Using \eqref{mutf} and \eqref{normtf} in \eqref{lambdadef} shows that
$\lambda_i(\e)$ are invariant under such transformations and hence
depend only on $(\Sigma,G,V)$ and $\e$.

\begin{definition} The quantities $\lambda_1(\e)$ and $\lambda_2(\e)$
are called the {\em principal values} of $(\Sigma,G,V)$ at the
critical end $\e$.
\end{definition}

\begin{definition}
The globally well-behaved potential $V$ is called {\em circular} at
the critical end $\e$ if the Hessian $\Hess(\hV)(\e)\in
\Sym^2(T^\ast_\e\Sigma)$ satisfies:
\be
(\Hess(\hV)(\e))({\hat \rho}_\e(R)(u), {\hat \rho}_\e(R)(v))=\Hess(\hV)(\e)(u,v)~~\forall R\in \SO(2)~~\forall u,v\in T_\e\hSigma~.
\ee
\end{definition}

\noindent Notice that $V$ is circular at $\e$ iff
$\lambda_1(\e)=\lambda_2(\e)$.

\begin{definition}
The {\em critical modulus} $\beta_\e$ of $(\Sigma,\!G,\!V)$ at the
critical end $\e$ is the ratio:
\ben
\label{beta}
\beta_\e\eqdef \frac{\lambda_1(\e)}{\lambda_2(\e)}\in [-1,1]\setminus \{0\}~~,
\een
where $\lambda_1(\e)$ and $\lambda_2(\e)$ are the principal values of
$(\Sigma,G,V)$ at $\e$. The sign factors:
\be
\epsilon_i(\e)\eqdef \sign(\lambda_i(\e))\in \{-1,1\}~~(i=1,2)
\ee
are called the {\em characteristic signs} of $(\Sigma,G,V)$ at $\e$.
\end{definition}

\noindent Notice that $\sign(\beta_\e)=\epsilon_1(\e)\epsilon_2(\e)$.

\subsection{Principal canonical coordinates centered at a critical
  end}
\label{subsec:PrincCritEnds}

\begin{definition}
A canonical Cartesian coordinate system $(x,y)$ for $(\Sigma,G)$
centered at the critical end $\e$ is called {\em principal} for $V$ if
the tangent vectors $\epsilon_x=\frac{\pd}{\pd x}\bvert_\e$ and
$\epsilon_y=\frac{\pd}{\pd y}\bvert_\e$ form a principal basis for $V$
at $\e$.
\end{definition}

\noindent In a principal coordinate system $(x,y)$ centered at $\e$,
the Taylor expansion of $\hV$ at $\e$ has the form:
\beqan
\label{Vase}
\!\!\hV(x,y)\!&=&\!\hV(\e)+
\frac{1}{2}\!\left[\lambda_1(\e) x^2\!+\!\lambda_2(\e) y^2\right]\!+\!\O((x^2+y^2)^{3/2})\!=
\nn\\
&=&\!\hV(\e)+\frac{1}{2}\omega^2\!\left[\lambda_1(\e) \cos^2 \theta\!+\!\lambda_2(\e)\sin^2 \theta\right]\!+\!\O(\omega^3)~,
\eeqan
where $\lambda_1(\e)$ and $\lambda_2(\e)$ are the principal values of
$V$ at $\e$ and $\omega=\sqrt{x^2+y^2}$, $\theta= \arg(x+\i y)$.

\begin{remark} A system of principal canonical coordinates at $\e$
determines a system of semigeodesic coordinates near $\e$ through
relation \eqref{inv}, which will be called a system of {\em principal
semigeodesic coordinates} near $\e$.
\end{remark}

Let $\Delta\simeq \Z_2\times \Z_2$ be the subgroup of $\O(2)$
generated by the axis reflections $(x,y)\rightarrow (-x,y)$ and
$(x,y)\rightarrow (x,-y)$. This subgroup also contains the point
reflection $(x,y)\rightarrow (-y,-x)$.

\begin{prop}
\label{prop:pn} There exists a principal Cartesian canonical
coordinate system $(x,y)$ for $(\Sigma,G,V)$ at every critical end
$\e$. When $V$ is circular at $\e$, these coordinates are determined
by $V$ and $G$ up to an $\O(2)$ transformation. When $V$ is not
circular at $\e$, these coordinates are determined by $V$ and $G$ up
to the action of the subgroup $\Delta$ of $\O(2)$.
\end{prop}

\begin{proof}
Starting from any Cartesian canonical coordinate system of
$(\Sigma,G)$ centered at $\e$, we can use the local isometry action
\eqref{O2action2} to rotate it into a principal canonical coordinate
system centered at $\e$.  The remaining statements are obvious.
\end{proof}

\noindent If $V$ is circular at $\e$, then any Cartesian canonical
coordinate system centered at $\e$ is principal. When $V$ is not
circular at $\e$ (i.e. when $\lambda_1(\e)\neq \lambda_2(\e)$), the
geodesic orbits of $(\dot{U}_\e,G)$ given by
$(\theta-\theta_\e){\mathrm \mod}\,2\pi \in
\{0,\frac{\pi}{2},\pi,\frac{3\pi}{2}\}$ will be called the {\em
principal geodesic orbits} at $\e$ determined by $V$. These geodesic
orbits correspond to the four semi-axes defined by the principal
Cartesian coordinate system $(x,y)$ centered at $\e$; they have the
end $\e$ as a common limit point.

\subsection{The asymptotic form of $G$ near the ends}

\noindent In canonical Cartesian coordinates $(x,y)$ centered at $\e$, we have
$\omega=\sqrt{x^2+y^2}=\frac{1}{r}$, $\theta=\arg(x+\i y)$ and (see
\eqref{emetric}):
\ben
\label{eomegametric}
\dd s_G^2|_{\dot{U}_\e}=\frac{\dd \omega^2}{\omega^4}+f_\e(1/\omega)\dd \theta^2~~,
\een
with:
\ben
\label{feas}
f_\e(1/\omega)= {\tilde c}_\e e^{\frac{2\epsilon_\e}{\omega}}\left[1+\O\left(e^{-\frac{2}{\omega}}\right)\right]~~\mathrm{for}~~\omega\ll 1~~,
\een
where\footnote{The quantity ${\tilde c}_\e$ is related to the quantity $c_\e$ used in \cite{genalpha} through the formula
  ${\tilde c}_\e=\left(\frac{c_\e}{4\pi}\right)^2$.}:
\ben
\label{tce}
{\tilde c}_\e=\fourpartdef{\frac{1}{4}}{~\e=\mathrm{plane~end}}
{\frac{1}{(2\pi)^2}}{~\e=\mathrm{horn~end}}
{\frac{\ell^2}{(4\pi)^2}}{~\e=\mathrm{funnel~end~of~circumference}~\ell>0}{\frac{1}{(2\pi)^2}}{~\e=\mathrm{cusp~end}}
\een
and:
\ben
\label{epsilone}
\epsilon_\e=\twopartdef{+1}{~\e=\mathrm{flaring~(i.e.~plane,~horn~or~funnel)~end}}{-1}{~\e=\mathrm{cusp~end}}~~.
\een
The term $\O\left(e^{-\frac{2}{\omega}}\right)$ in \eqref{feas}
vanishes identically when $\e$ is a cusp or horn end. In particular,
the constants ${\tilde c}_\e$ and $\epsilon_\e$ determine the leading
asymptotic behavior of the hyperbolic metric $G$ near $\e$.

The gradient flow equations of $(\dot{U}_\e,G|_{\dot{U}_\e},V|_{\dot{U}_\e})$ read:
\beqan
\label{gradeqe}
&& \frac{\dd \omega}{\dd q}=-(\grad V)^\omega\simeq -\omega^4 \pd_\omega V~~\nn\\
&& \frac{\dd \theta}{\dd q}=-(\grad V)^\theta\simeq -\frac{1}{{\tilde c}_\e}e^{-\frac{2\epsilon_\e}{\omega}}\pd_\theta V~~.
\eeqan
We now proceed to study these equations for each end of $\Sigma$.

\begin{remark}
Recall that $V$ is globally well-behaved and $\hV$ is Morse on
$\hSigma$. Together with the formulas above, this implies that
$(\grad_G V)^\omega$ tends to zero at all ends while $(\grad_G
V)^\theta$ tends to zero exponentially at flaring ends and to infinity
at cusp ends. On the other hand, we have:
\be
\label{nue}
||\grad_G V||^2=||\dd V||^2=\frac{1}{\omega^4} (\partial_\omega V)^2+f_e(1/\omega) (\partial_\theta V)^2\approx \frac{1}{\omega^4} (\partial_\omega V)^2+ {\tilde c}_\e e^{\frac{2\epsilon_\e}{\omega}} (\partial_\theta V)^2~~.
\ee
Thus $||\grad_G V||$ tends to infinity at all ends.
\end{remark}

\section{The IR phases of interior critical points}
\label{sec:int}

\noindent To describe the IR phases of interior critical points, we
first use the hyperbolic geometry of $(\Sigma,G)$ to introduce
convenient local coordinate systems centered at such points. We first
describe the canonical systems of local coordinates afforded by the
hyperbolic metric around each point of $\Sigma$, then we specialize
these to local coordinates centered at an interior critical point in
which the second order of the Taylor expansion of $V$ has no
off-diagonal terms. Using such coordinates allows us to determine
explicitly the asymptotic form of the gradient flow of $(\Sigma,G,V)$
near each interior critical point.

\subsection{Canonical local coordinates centered at a point of $\Sigma$}
\label{subsec:CanIntCrit}

\noindent Denote by $\exp_m^G:T_m\Sigma\rightarrow \Sigma$ the exponential
map of $(\Sigma,G)$ at a point $m\in \Sigma$. Since the metric $G$ is
complete, this map is surjective by the Hopf-Rinow theorem. For any point $m\in
\Sigma$, let $r(m)$ be the injectivity radius of $G$ at $m$ (see
\cite{Petersen}) and set
$\omega_{\rmax}(m)=\tanh\big(\frac{r(m)}{2}\big)$. Let:
\be
\rD_{\omega_\rmax(m)}\eqdef \{(x,y)\in \R^2|\sqrt{x^2+y^2}< \omega_{\rmax}(m)\}~~
\ee
be the open disk of radius $\omega_{\rmax}(m)$ centered at the origin of the plane.

\begin{definition}
A system of {\em canonical Cartesian coordinates} for $(\Sigma,G)$ at
a point $m\in \Sigma$ is a system of local coordinates $(U_m,x,y)$
centered at $m$ (where $U_m$ is an open neighborhood of $m$) such that
the image of the coordinate map $(x,y):U_m\rightarrow \R^2$ coincides
with the disk $\rD_{\omega_\rmax(m)}$ and the restriction of the
hyperbolic metric $G$ to $U_m$ takes the Poincar\'e form:
\ben
\label{princanint}
\dd s^2_G\vert_{U_m}=\frac{4}{[1-(x^2+y^2)]^2} (\dd x^2+\dd y^2)~~.
\een  
\end{definition}

\begin{remark}
The map $f=(x,y):U_m\rightarrow \rD_{\omega_\rmax(m)}$ is an isometry
between $(U_m,G_m)$ and the region $|u|\leq
\tanh\big(\frac{r(m)}{2}\big)$ of the Poincar\'e disk. The proof
below shows that $U_m=\exp_m(B(r(m)))$, where $B(r(m))\subset
T_m\Sigma$ is the Euclidean disk of radius $r(m)$ centered at the
origin of the Euclidean space $(T_m\Sigma,G_m)$.
\end{remark}

\begin{prop}
\label{prop:CanCoords}
Canonical Cartesian coordinates for $(\Sigma,G)$ exist at any point $m\in \Sigma$. 
\end{prop}

\begin{proof}
By the Poincar\'e uniformization theorem, the hyperbolic surface
$(\Sigma,G)$ is locally isometric with the Poincar\'e disk
$\mD$. Since the latter is a homogeneous space, the local isometry $f$
in a vicinity of $m$ can be taken to send $m$ to the origin of $\mD$
and to be defined on the vicinity $U_m=\exp_m(B(r(m)))$ of $m$ in
$\Sigma$, where $B(r(m))\subset T_m\Sigma$ is the Euclidean disk of radius
$r(m)$ centered at the origin of $T_m\Sigma$. The image of $U_m$
through $f$ coincides with the image of a similar ball inside $T_0\mD$
through the exponential map $\exp_0^\mD$ at the origin of the
Poincar\'e disk. Recall that $\exp_0(v)=\psi_v(1)$, where $\psi_v(t)$
is a geodesic of $\mD$ which starts at the origin and satisfies
$\frac{\dd \psi_v(t)}{\dd t}\bvert_{t=0}=v\in T_0\mD$.  For such a
geodesic, the set $\psi_v([0,1])$ is a segment of hyperbolic length
$||v||$ which connects the origin of $\mD$ with the point $u\eqdef
\psi_v(1)\in \mD$. The hyperbolic length formula gives
$||v||=2\arctanh|u|$, hence $\exp_0(v)\!=\!\psi_v(1)$ lies on the circle
of Euclidean radius $\tanh\left(\frac{||v||}{2}\right)$. Thus
$\exp_0^\mD(B(r(m)))\!=\!D_{\omega_\rmax(m)}$ and hence $f$ is a
diffeomorphism from $U_m$ to $D_{\omega_\rmax(m)}$. The conclusion
follows by setting $(x, y)=f$.
\end{proof}

\subsection{Principal values, critical modulus and characteristic signs at an interior critical point}

\begin{definition} The {\em principal values} of
$(\Sigma,G,V)$ at an interior critical point $\c\in \Crit V$ are the
eigenvalues $\lambda_1(\c)$ and $\lambda_2(\c)$ of the Hessian
operator $\wHess_G(V)(\c)\in \End^s(T_\c\Sigma)$, ordered such that
$|\lambda_1(\c)|\leq |\lambda_2(\c)|$. We say that $V$ is {\em
circular} at $\c$ if $\lambda_1(\c)=\lambda_2(\c)$.
\end{definition}

\begin{remark}
Notice that $\wHess_G(V)(\c)$ is the symmetric linear operator defined by
$\Hess(V)(\c)$ in the Euclidean vector space $(T_\c\Sigma,G_\c)$:
\be
\Hess(V)(\c)(w_1,w_2)=G_\c(\wHess_G(V)(\c)(w_1),w_2)~~\forall w_1,w_2\in T_\c\Sigma~~.
\ee
\end{remark}

\noindent When $V$ is not circular at $\c$, the one-dimensional
eigenspaces $L_1(\c),L_2(\c)\subset T_\c\Sigma$ of $\wHess_G(V)(\c)$
are called the {\em principal lines} of $(\Sigma,G,V)$ at $\c$. The
geodesic orbits $\exp_\c^G(L_1(\c))$ and $\exp_\c^G(L_2(\c))$
determined by the principal lines are called the {\em principal
geodesic orbits} of $(\Sigma,G,V)$ at $\c$.

\begin{definition}
With the notations of the previous definition, the {\em critical
modulus} $\beta_\c$ and {\em characteristic signs} $\epsilon_1(\c)$
and $\epsilon_2(\c)$ of $(\Sigma,G,V)$ at $\c$ are defined through:
\ben
\beta_\c\eqdef\frac{\lambda_1(\c)}{\lambda_2(\c)}\in [-1,1]\setminus \{0\}~~,~~\epsilon_i(\c)\eqdef \sign(\lambda_i(\c))~~(i=1,2)~~.
\een
\end{definition}

\noindent Notice the relation:
\be
\sign(\beta_\c)=\epsilon_1(\c)\epsilon_2(\c)~~.
\ee
The critical point $\c$ is a local extremum of $V$ (a sink or a
source) when $\beta_\c>0$ and a saddle point of $V$ when
$\beta_\c<0$. In the first case, $\c$ is a sink
($\epsilon_1(\c)=\epsilon_2(\c)=1$) or source
($\epsilon_1(\c)=\epsilon_2(\c)=-1$) depending on whether the gradient
curves of $(\Sigma,G,V)$ near $\c$ flow toward $\c$ or away from
$\c$. When $\c$ is a saddle point, the gradient curves flow toward the
$x$ axis when $\epsilon_1(\c)=-1$ and toward the $y$ axis when
$\epsilon_2(\c)=-1$.

\subsection{Principal canonical coordinates centered at an interior critical point}

\noindent \noindent Let $\c\in \Crit V$ be an interior critical point.

\begin{definition}
A system of {\em principal Cartesian canonical coordinates} for
$(\Sigma,G,V)$ at $\c$ is a system of canonical Cartesian coordinates for
$(\Sigma,G)$ centered at $\c$ such that the tangent vectors
$\epsilon_x=\frac{\pd}{\pd x}\bvert_\c$ and $\epsilon_y=\frac{\pd}{\pd
y}\bvert_\c$ are eigenvectors of the Hessian operator
$\wHess_G(V)(\c)$ corresponding to the principal values
$\lambda_1(\c)$ and $\lambda_2(\c)$ of $(\Sigma,G,V)$ at $\c$.
\end{definition}

\noindent In principal Cartesian canonical coordinates $(x,y)$
centered at $\c$, the Taylor expansion of $\hV$ at $\c$ has the form:
\beqan
\label{Vasc}
\!\!\hV(x,y)\!&=&\hV(\c)+
\!\frac{1}{2}\!\left[\lambda_1(\c) x^2\!+\!\lambda_2(\c) y^2\right]\!+\!\O((x^2+y^2)^{3/2})\!=\nn\\
&=&\!\hV(\c)+\frac{1}{2}\omega^2\!\left[\lambda_1(\c) \cos^2 \theta\!+\!\lambda_2(\c)\sin^2 \theta\right]\!+\!\O(\omega^3)~,
\eeqan
where $\lambda_1(\c)$ and $\lambda_2(\c)$ are the principal values of
$V$ at $\c$ and we defined $\omega\eqdef \sqrt{x^2+y^2}$, $\theta\eqdef
\arg(x+\i y)$.

\begin{prop}
Principal Cartesian canonical coordinates for $(\Sigma,G,V)$ exist
at any interior critical point $\c\in \Crit V$.
\end{prop}

\begin{proof}
The proof of Proposition \ref{prop:CanCoords} implies that the
canonical Cartesian coordinates of $(\Sigma,G)$ centered at $\c$ are
determined up to transformations of the form:
\be
\left[\begin{array}{c} x\\ y\end{array}\right]\rightarrow A
\left[\begin{array}{c} x\\ y\end{array}\right]~~\forall A\in \O(2)~~.
\ee
These corresponds to isometries of the Poincar\'e disk metric which
fix the origin and induce orthogonal transformations of the Euclidean
space $(T_\c\Sigma,G_\c)$. Performing such a transformation we can ensure
that the orthogonal vectors $\pd_x\vert_\c$ and $\pd_y\vert_\c$ are
eigenvectors of the symmetric operator ${\wHess}_G(V)(\c)\in \End(T_\c
\Sigma)$ whose eigenvalues $\lambda_1(\c)$ and $\lambda_2(\c)$ satisfy
$|\lambda_1(\c)|\leq |\lambda_2(\c)|$.
\end{proof}

\subsection{The infrared behavior near an interior critical point}
\label{IRint}

\noindent Let $\c$ be an interior critical point and $(x,y)$ be
principal Cartesian canonical coordinates centered at $\c$. Setting
$\omega\eqdef\sqrt{x^2+y^2}$ and $\theta\eqdef \arg(x+\i y)$, we have:
\be
\dd s^2_G=\frac{4}{(1-\omega^2)^2}[\dd \omega^2+\omega^2\dd \theta^2]~~
\ee
and:
\be
V(\omega,\theta)= V(\c)+\frac{1}{2}\omega^2\left[\lambda_1(\c) \cos^2\theta +\lambda_2(\c) \sin^2\theta\right]+\O(\omega^3)~~.
\ee
Thus:
\beqan
\label{gradc}
&& (\grad V)^\omega\!\approx\!\frac{(1-\omega^2)^2}{4} \pd_\omega V\!=\!\frac{(1-\omega^2)^2\omega}{4}[\lambda_1(\c)\cos^2\theta+\lambda_2(\c)\sin^2\theta]~,\nn\\
&& (\grad V)^\theta\!\approx\!\frac{(1-\omega^2)^2}{4\omega^2} \pd_\theta V\!=\!\frac{(1-\omega^2)^2}{4} [\lambda_2(\c)\!-\!\lambda_1(\c)]\sin(\theta)\cos(\theta)~.
\eeqan
Distinguish the cases:
\begin{enumerate}
\item $\lambda_1(\c)=\lambda_2(\c):=\lambda(\c)$,
  i.e. $\beta_\c=1$. Then $\epsilon_1(\c)=\epsilon_2(\c):=\epsilon(\c)$
and $\c$ is a local minimum of $V$ when $\lambda(\c)$ is positive
(i.e. when $\epsilon(\c)=1$) and a local maximum of $V$ when
$\lambda(\c)$ is negative (i.e. when $\epsilon(\c)=-1$). Relations
\eqref{gradc} become:
\be
(\grad V)^\omega\approx \frac{(1-\omega^2)^2\omega}{4}\lambda(\c)~~,~~ (\grad V)^\theta\approx 0
\ee
and the gradient flow equation of $(\Sigma,G,V)$ takes the following approximate form near $\c$:
\be
\frac{\dd \omega}{\dd q}=-\frac{(1-\omega^2)^2\omega}{4}\lambda(\c)~~,~~\frac{\dd \theta}{\dd q}=0~~.
\ee
This gives $\theta=\const$, i.e. the gradient flow curves near $\c$
are approximated by straight lines through the origin when drawn in
principal Cartesian canonical coordinates $(x,y)$ at $\c$; their orbits are
geodesic orbits of $(\Sigma,G)$ passing though $\c$ since $G$ identifies
near $\c$ with the Poincar\'e disk metric. The gradient lines flow
toward/from the origin when $\c$ is a local minimum/maximum of $V$.
\item $\lambda_1(\c)\neq \lambda_2(\c)$, i.e. $\beta_\c\neq 1$.  When
$\theta\in \{0,\frac{\pi}{2}, \pi, \frac{3\pi}{2}\}$, the gradient
flow equation reduces to:
\beqa
&&\frac{\dd \omega}{\dd q}=\frac{(1-\omega^2)^2\omega}{4}\times \twopartdef{\lambda_1(\c)}{\theta\in \{0,\pi\}}{\lambda_2(\c)}{\theta\in \{\frac{\pi}{2},\frac{3\pi}{2}\}}\\
&& \frac{\dd \theta}{\dd q}=0~~.
\eeqa
This gives four gradient flow orbits which are approximated near $\c$
by the principal geodesic orbits. When $\theta\not \in
\{0,\frac{\pi}{2}, \pi, \frac{3\pi}{2}\}$, the gradient flow equation
takes the form:
\ben
\label{gradcn}
(1-\beta_\c)\frac{\dd \omega}{\dd \theta}=\omega (\beta_\c\cot\theta+\tan\theta)~~,
\een
with general solution:
\ben
\label{gradflowc}
\omega=C \frac{|\sin(\theta)|^{\frac{\beta_\c}{1-\beta_\c}}}{|\cos(\theta)|^\frac{1}{1-\beta_\c}}~~,
\een
where $C$ is a positive integration constant and we used the primitives:
\ben
\label{primitives}
\int \dd \theta \cot\theta =\log|\sin \theta|~~\mathrm{and}~~\int \dd \theta \tan\theta =-\log|\cos \theta|~~.
\een
This gives four gradient flow orbits, each of which lies within one of
the four quadrants; the orbits are related to each other by
reflections in the coordinate axes. In this case, $\beta_\c$ can be
positive or negative and the orientation of the gradient flow orbits
is determined by the characteristic signs.
\end{enumerate}

\noindent Figure \ref{fig:IntCritical} shows the unoriented gradient
flow orbits of $V$ near an interior critical point $\c$ for
$\beta_\c=-1/2$ and $\beta_\c=+1/2$.

\begin{figure}[H]
\centering
\begin{minipage}{.49\textwidth}
\centering  \includegraphics[width=.95\linewidth]{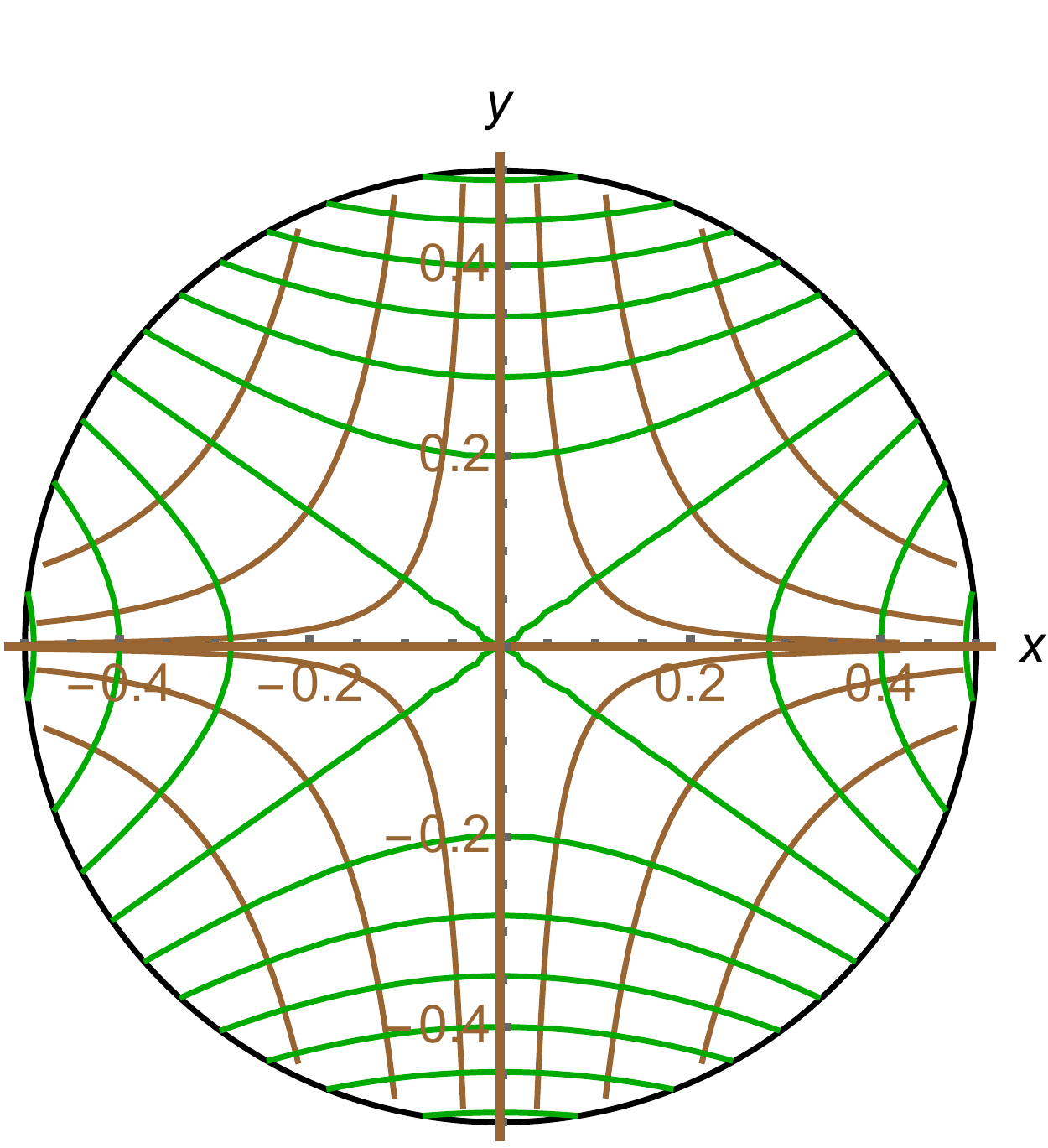}
\subcaption{For $\beta_\c\!=\!\!-0.5$ (interior saddle point of $V$)}
\end{minipage}\hfill 
\begin{minipage}{.5\textwidth}
\centering \includegraphics[width=.93\linewidth]{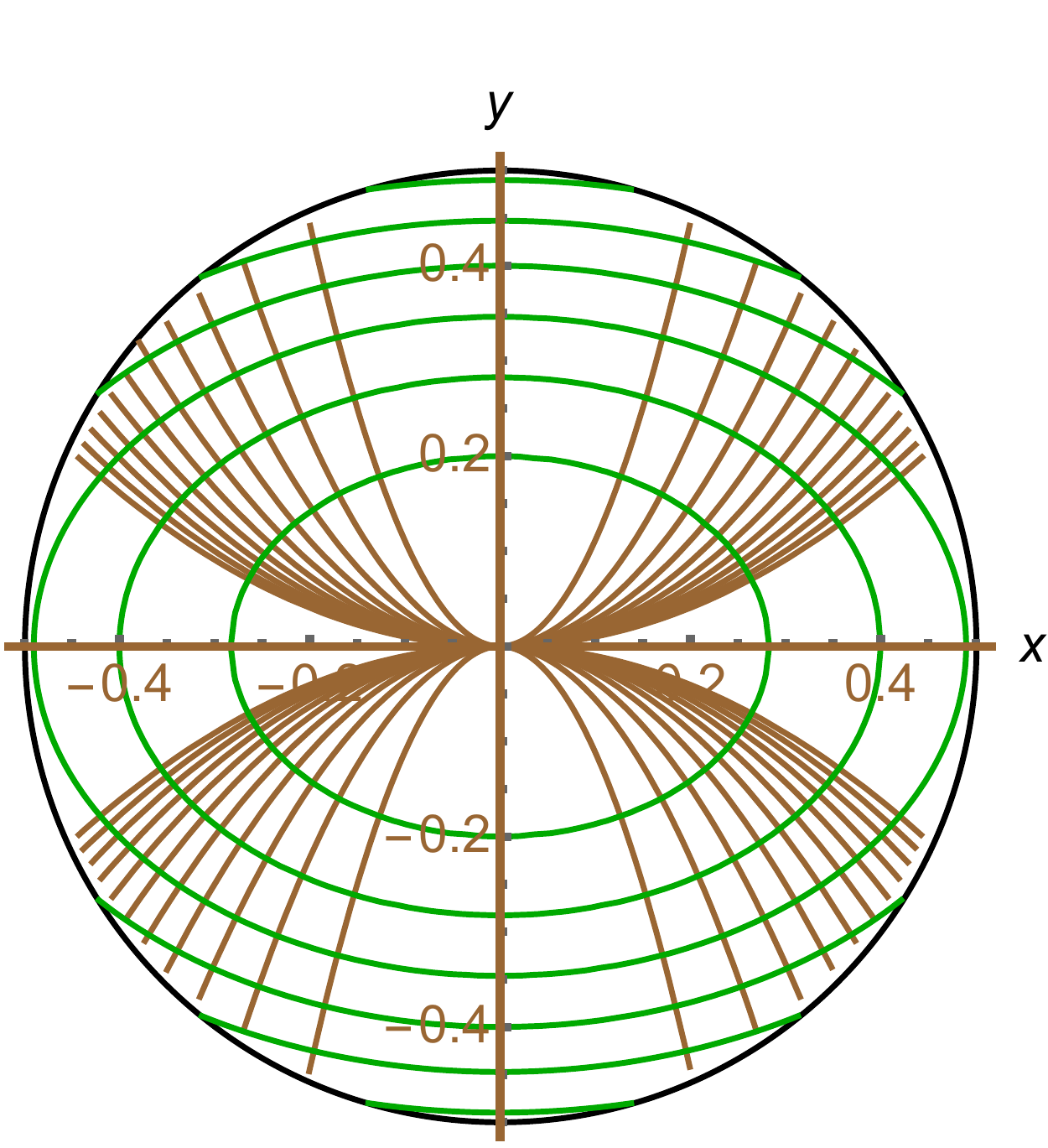}
\subcaption{For $\!\beta_\c\!\!=\!\!0.5$ (interior local extremum of $V$)}
\end{minipage}
\caption{Unoriented gradient flow orbits of $V$ (shown in brown) near
an interior critical point superposed over the level lines of $V$
(shown in green) for two values of $\beta_\c$, plotted in principal
Cartesian canonical coordinates centered at the critical point. The figure
assumes $\omega_\rmax(\c)=1/2$, i.e. that the injectivity radius at
$\c$ equals $r(\c)=2\arctanh(1/2)=1.098$. The principal coordinate
axes correspond to the principal geodesic orbits at $\c$, which
coincide with four special gradient flow orbits.}
\label{fig:IntCritical}
\end{figure}

\

The scalar potential $\Phi$ of the model can be recovered
from the classical effective potential $V=M_0\sqrt{2\Phi}$ as:
\ben
\label{Phic}
\Phi=\frac{1}{2 M_0^2} V^2\approx \frac{{\bar \lambda}_2(\c)^2}{2}\left[\bar{V}(\c)+\frac{1}{2}\omega^2 (\beta_\c \cos^2\theta+\sin^2\theta)\right]^2~~,
\een
where we defined:
\be
{\bar \lambda}_2(\c)\eqdef \frac{\lambda_2(\c)}{M_0}~~,~~{\bar V}(\c)\eqdef \frac{V(\c)}{\lambda_2(\c)}~~.
\ee
Figure \ref{fig:CosmIntCritical} shows some numerically computed
infrared optimal cosmological curves of the uniformized model
parameterized by $(M_0,\Sigma,G,\Phi)$ near an interior critical point
$\c$. In this figure, we took ${\bar \lambda}_2(\c)=1$, ${\bar
V}(\c)=1$ and $M_0=1$, so the rescaled scalar triple
$(\Sigma,G_0,\Phi_0)$ coincides with $(\Sigma,G,\Phi)$ (recall that
the choice of $M_0$ does not affect the cosmological or gradient flow
orbits). Notice that the accuracy of the first order IR approximation
depends on the value of ${\bar V}(\c)$, since the first IR parameter
of \cite{ren} depends on this value.

\begin{figure}[H] \centering
\begin{minipage}{.49\textwidth} \centering
\includegraphics[width=.95\linewidth]{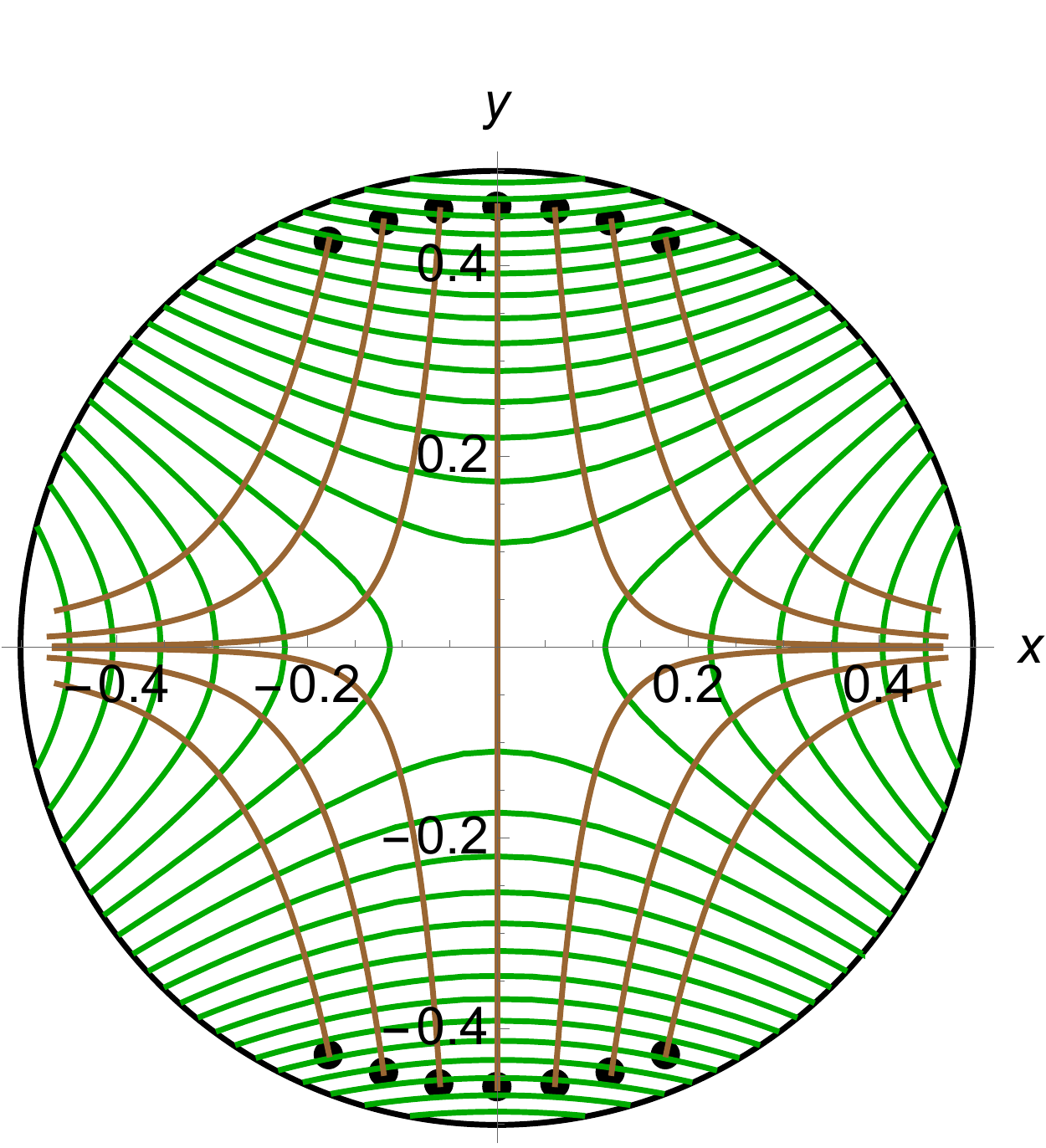} \subcaption{For
$\beta_\c\!=\!\!-0.5$ (interior saddle point of $\Phi$)}
\end{minipage}\hfill
\begin{minipage}{.5\textwidth} \centering
\includegraphics[width=.93\linewidth]{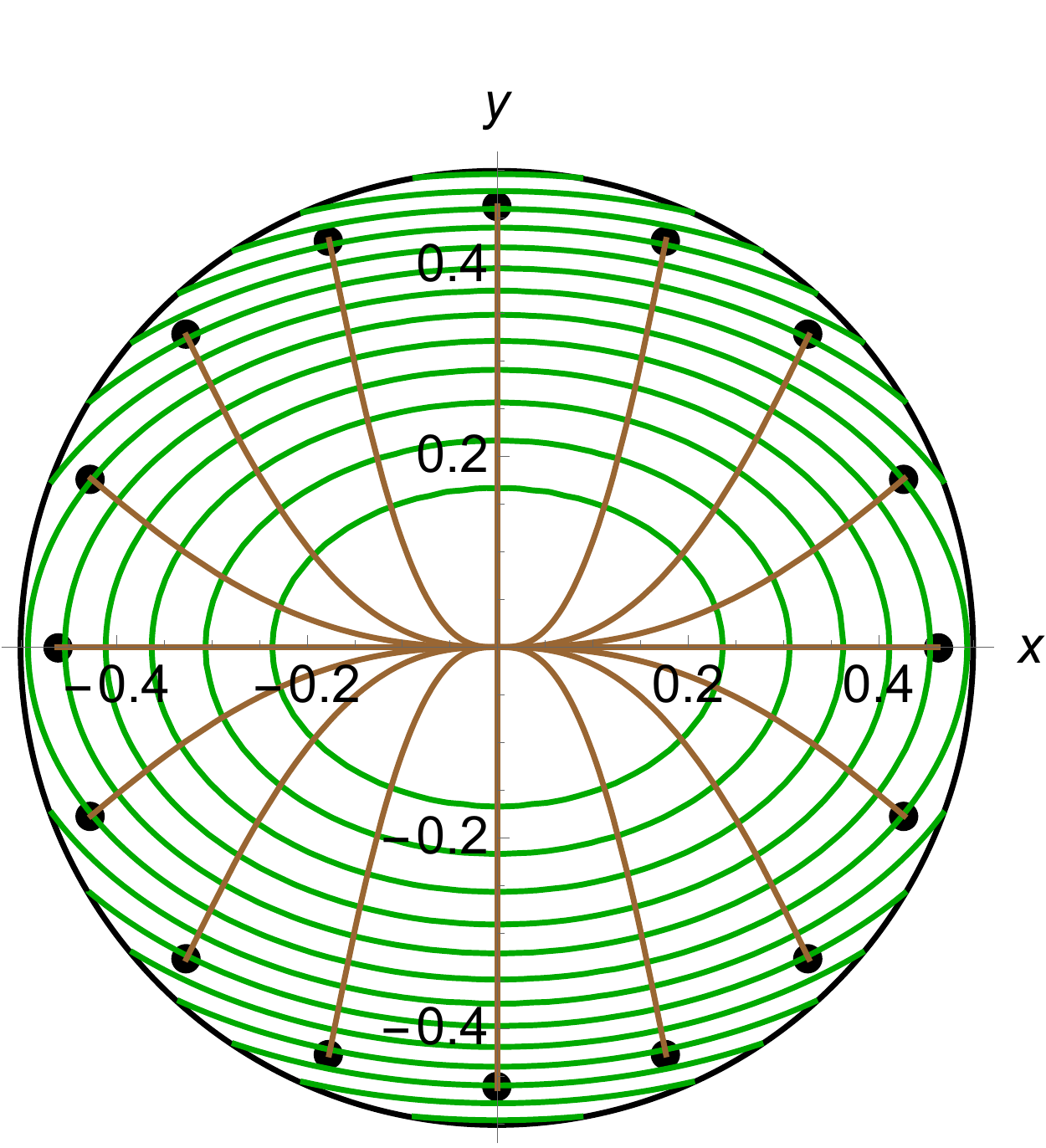} \subcaption{For
$\!\beta_\c\!\!=\!\!0.5$ (interior local extremum of $\Phi$)}
\end{minipage}
\caption{Numerically computed orbits of infrared optimal cosmological
curves of the uniformized model (shown in brown) near an interior
critical point $\c$, superposed over the level lines of $\Phi$ (shown
in green) for two values of $\beta_\c$. Here $x,y$ are principal
Cartesian canonical coordinates centered at the critical point. We assume
$\omega_\rmax(\c)=1/2$, i.e. that the injectivity radius is
$r(\c)=2\arctanh(1/2)=1.098$. The figure shows cosmological orbits
for IR optimal cosmological curves $\varphi$, whose initial speed
belongs to the gradient shell of $(\Sigma,G,V)$. The initial points
$\varphi(0)$ of these curves are shown as black dots.}
\label{fig:CosmIntCritical}
\end{figure}

\noindent The results above imply the following:

\begin{prop}
\label{prop:IRIntCrit}
The asymptotic form of the gradient flow orbits of $(\Sigma, G,V)$
near $\c$ is determined by the critical modulus $\beta_\c$, while the
orientation of these orbits is determined by the characteristic signs
at $\c$. In particular, the first order IR approximation of those
cosmological orbits which have $\c$ as an $\alpha$- or $\omega$- limit
point depends only on these quantities.
\end{prop}

\noindent Notice that the unoriented orbits of the asymptotic gradient
flow near $\c$ are determined by the positive homothety class of the
pair $(\lambda_1(\c),\lambda_2(\c))$, i.e. by the image of this pair
in the quotient $(\R^\times)^2/\R_{>0}$, where the multiplicative
group $\R_{>0}$ acts diagonally:
\be
(\lambda_1,\lambda_2)\rightarrow (\alpha\lambda_1,\alpha\lambda_2)~~\forall \alpha>0~~.
\ee
Equivalently, these unoriented orbits depend only on the positive
homothety class $[\wHess(V)(\c)]\in \End^s(T_\c\Sigma)/\R_{>0}$ of the
Hessian operator of $V$ at $\c$. On the other hand, the topological
and smooth topological equivalence class of the gradient flow of
$(\Sigma,G,V)$ near $\c$ is given by the following well-known result
(see, for example, \cite[Chap. 2, Thm.  5.1]{Palis}), where
$\varepsilon_1(\c)\eqdef \min(\epsilon_1(\c),\epsilon_2(\c))$ and
$\varepsilon_2(\c)\eqdef \max(\epsilon_1(\c),\epsilon_2(\c))$:
 
\begin{prop}
\label{prop:topint}
The gradient flow of $(\Sigma,G,V)$ is locally topologically
equivalent near $\c$ with the gradient flow of the function:
\be
V_0(x,y)=\frac{1}{2}(\varepsilon_1(\c)x^2+\varepsilon_2(\c) y^2)=
\threepartdef{\frac{1}{2}(x^2+y^2)}{~\ind(\c)=0}{\frac{1}{2}(-x^2+y^2)}{~\ind(\c)=1}{-\frac{1}{2}(x^2+y^2)}{~\ind(\c)=2}
\ee
computed with respect to the Euclidean metric:
\ben
\label{s0}
\dd s^2_0=\dd x^2+\dd y^2~~.
\een
This topological equivalence can be chosen to be smooth when
$\ind(\c)=1$ and when $\ind(\c)\in \{0,2\}$ with
$\lambda_1(\c)=\lambda_2(\c)$. When $\ind(\c)\in \{0,2\}$ and
$\lambda_1(\c)\neq \lambda_2(\c)$, the gradient flow of $(\Sigma,G,V)$
is locally smoothly topologically equivalent near $\c$ with the
gradient flow of the function:
\be
V_0(x,y)=\frac{1}{2}(\varepsilon_1(\c)x^2+\varepsilon_2(\c) (2y)^2)=
\twopartdef{\frac{1}{2}(x^2+4y^2)}{~\ind(\c)=0}{-\frac{1}{2}(x^2+4y^2)}{~\ind(\c)=2}~~.
\ee
\end{prop}

\noindent Hence the IR phase determined by $\c$ belongs to one of five
IR universality classes (in the sense of \cite{ren}), which are
characterized respectively by the conditions $\ind(\c)=1$,
$\ind(\c)=0$ with $\beta_\c=1$, $\ind(\c)=2$ with $\beta_\c= 1$,
$\ind(\c)=0$ with $\beta_\c\neq 1$ and $\ind(\c)=2$ with $\beta_\c\neq
1$. Notice that Proposition \ref{prop:IRIntCrit} and equation
\eqref{gradflowc} give detailed information about the asymptotic form
of the first order IR orbits of the uniformized model near $\c$, while
Proposition \ref{prop:topint} characterizes its local IR universality
class in the phase determined by $\c$.

\section{The IR phases of noncritical ends}
\label{sec:noncritends}

\noindent Recall that an end $\e\in \Ends(\Sigma)$ is noncritical if
$(\dd \hV)(\e)\neq 0$. In this case, the kernel of the linear map
$(\dd \hV)(\e):T_\e\hSigma\rightarrow \R$ is one-dimensional and hence
we can rotate the Cartesian canonical coordinates $(x,y)$ centered at
$\e$ to ensure that $(\dd \hV)(\e)$ vanishes on the tangent vector
$v_y\eqdef\frac{\partial}{\partial y}\bvert_\e\in T_\e \hSigma$ and
that the quantity $\mu_\e\eqdef (\dd \hV)(\e)(v_x)=(\pd_x \hV)(\e)$ is
positive, where $v_x\eqdef \frac{\partial}{\partial x}\bvert_\e$.

\begin{definition}
A system of {\em special Cartesian canonical coordinates} for
$(\Sigma,G,V)$ centered at the noncritical end $\e$ is a system of
canonical Cartesian coordinates $(x,y)$ centered at $\e$ which
satisfies the conditions:
\be
(\dd \hV)(\e)(v_y)=0~~\mathrm{i.e.}~~(\pd_y \hV)(\e)=0~~
\ee
and:
\be
(\dd \hV)(\e)(v_x)>0~~\mathrm{i.e.}~~(\pd_x \hV)(\e)>0~~.
\ee
Given such coordinates, we set $\mu_\e\eqdef (\dd
\hV)(\e)(v_x)=(\pd_x\hV)(\e)$.
\end{definition}

\noindent In special Cartesian canonical coordinates centered at a
noncritical end $\e$, the Taylor expansion of the extended potential
$\hV$ at $\e$ has the form:
\ben
\label{Vase1}
\hV(x,y)=\hV(\e)+\mu_\e x +\O(x^2+y^2)=\hV(\e)+\mu_\e\omega\cos \theta+\O(\omega^2)~~.
\een
In leading order, we have:
\beqan
\label{egradas}
&& (\grad_G V)^\omega=\omega^4 \partial_\omega V= \mu_\e \omega^4 \cos \theta\nn\\
&& (\grad_G V)^\theta=\frac{1}{\tc_\e} e^{-\frac{2\epsilon_\e}{\omega}}\partial_\theta V=
-\frac{\mu_\e}{{\tilde c}_\e} e^{-\frac{2\epsilon_\e}{\omega}} \omega \sin\theta ~~.
\eeqan

\subsection{Special gradient flow orbits}

\noindent For $\theta\in \{0,\pi\}$, the gradient flow equation reduces in
leading order to $\theta=\const$ and:
\be
\frac{\dd\omega}{\dd q}=\twopartdef{-\mu(\e)\omega^4}{\theta=0}{\mu(\e)\omega^4}{\theta=\pi}~~,
\ee
with general solution:
\be
\omega=\twopartdef{\frac{1}{(3\mu(\e)(q_0+q))^{1/3}}}{\theta=0~(\mathrm{with}~q>-q_0)}{\frac{1}{(3\mu(\e)(q_0-q))^{1/3}}}{\theta=\pi~(\mathrm{with}~q<q_0)}~~,
\ee
where $q_0$ is an arbitrary constant. Shifting $q$ by $q_0$ when
$\theta=0$ and by $-q_0$ when $\theta=\pi$, we can bring this to the
form:
\be
\omega=\twopartdef{\frac{1}{(3\mu(\e)q)^{1/3}}}{\theta=0~(\mathrm{with}~q>0)}{\frac{1}{(3\mu(\e)|q|)^{1/3}}}{\theta=\pi~(\mathrm{with}~q<0)}~~.
\ee
This gives two gradient flow orbits which have $\e$ as a limit point
and asymptote close to $\e$ to the principal geodesic orbits
corresponding to the $x$ axis. The first of these lies along the
positive $x$ semi-axis ($\theta=0$) and tends to $\e$ for
$q\rightarrow +\infty$ (thus it approximates the late time behavior of
a cosmological orbit near the end) while the second lies on the
negative $x$ semi-axis ($\theta=\pi$) and tends to $\e$ for
$q\rightarrow -\infty$ (thus it approximates the early time behavior
of a cosmological orbit near the end). The cosmological orbit
approximated by the gradient orbit with $\theta=0$ flows towards $\e$
in the distant future while that approximated by the gradient orbit
with $\theta=\pi$ originates from $\e$ in the distant past. Of course,
the end $\e$ is never reached by these orbits since it lies at
infinity on $\Sigma$.

\subsection{Non-special gradient flow orbits}

\noindent For $\theta\not \in \{0,\pi\}$, the gradient flow equation reduces to:
\ben
\frac{\dd \omega}{\dd \theta}=-\tilde{c}_\e \omega^3 \cot(\theta) e^{\frac{2\epsilon_\e}{\omega}}
\een
with Pfaffian form:
\ben
\label{gradnc}
\frac{e^{-\frac{2\epsilon_\e}{\omega}}}{\omega^3}\dd \omega =-\tilde{c}_\e \cot(\theta) \dd \theta~~.
\een
Setting $v \eqdef \frac{2\epsilon_\e}{\omega}$, we have:
\be
\frac{e^{-\frac{2\epsilon_\e}{\omega}}}{\omega^3}\dd \omega=-\frac{1}{4} v e^{-v} \dd v=-\frac{1}{4}\dd \bgamma_2(v)~~,
\ee
where:
\ben
\label{gamma}
\bgamma_n(v):=\bGamma(n,0,v)\eqdef \int_0^v w^{n-1} e^{-w} \dd w ~~(\mathrm{with}~ v\in \R)
\een
is the lower incomplete Gamma function of order $n$, which is plotted
in Figure \ref{fig:gamma2} for $n=2$. We have:
\ben
\label{gamma2}
\bgamma_2(v)=1-e^{-v}-ve^{-v}~~.
\een

\begin{figure}[H]
\centering
\includegraphics[width=.62\linewidth]{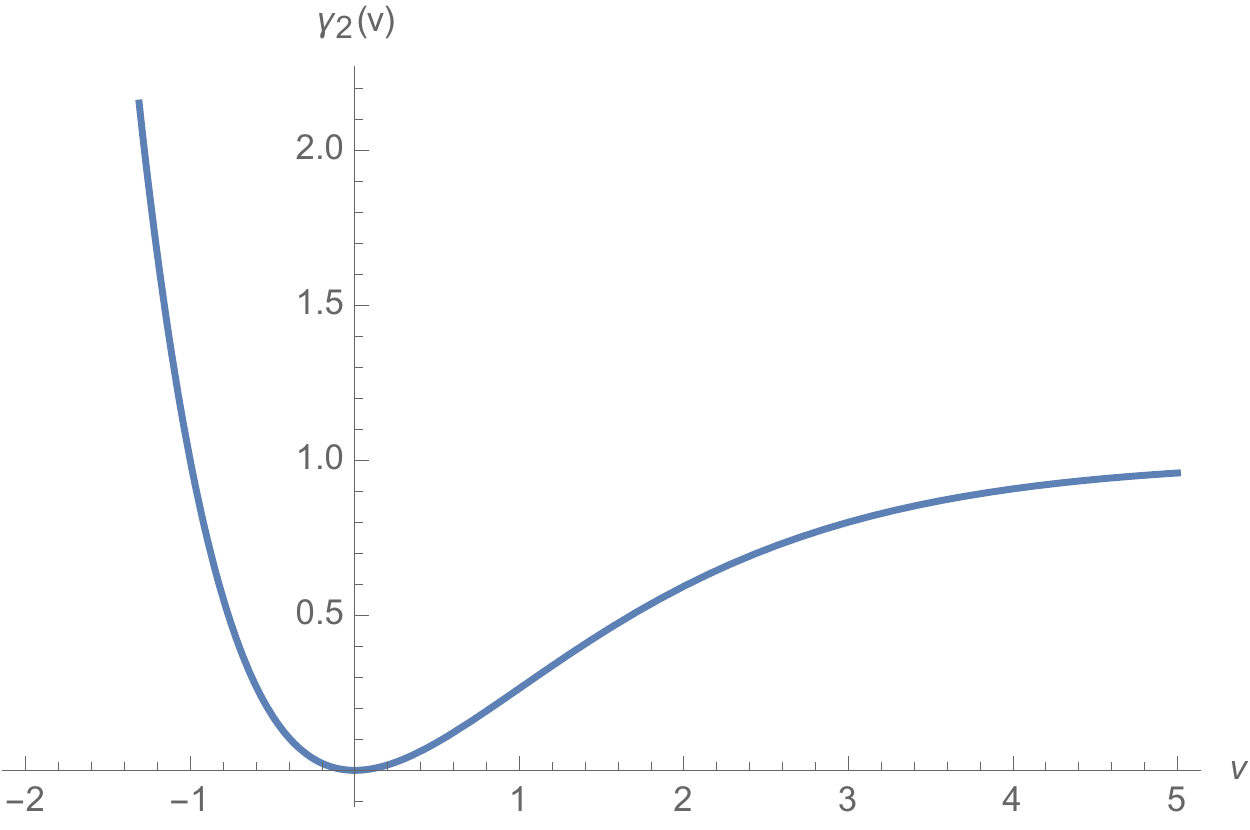}
\caption{Graph of the incomplete gamma function $\bgamma_2$ of order $2$.}
\label{fig:gamma2}
\end{figure}

Since $\cot(\theta) \dd \theta=\dd \log|\sin\theta|$, we find that the
non-special gradient flow orbits have the following implicit equation
near $\e$:
\ben
\label{gradflownc}
\frac{1}{4}\bgamma_2\left(\frac{2\epsilon_\e}{\omega}\right)=A+\tilde{c}_\e \log|\sin\theta|~~,
\een
where $A$ is an integration constant. In particular, the asymptotic
form of the unoriented gradient flow orbits as one approaches the
noncritical end $\e$ is determined in leading order by the constants
$\epsilon_\e$ and ${\tilde c}_\e$. Notice that the asymptotic orbits
are independent of the constant $\mu(\e)$. Since $\mu(\e)>0$, the effective
potential increases with $x$ and all curves flow in the direction
opposite to that of the $x$ axis.

\paragraph{Asymptotic sectors for non-special gradient flow orbits near the end.}

Since the left hand side of \eqref{gradflownc} is positive while the
right hand side is bounded from above by $A$, this equation requires
$A> 0$. Moreover, a solution exists only for ${\tilde
c}_\e\log|\sin\theta|\geq -A$ i.e.:
\be
|\sin\theta|> \sigma_A:=\sigma_A(\e)\eqdef e^{-\frac{A}{{\tilde c}_\e}}~~,
\ee
which becomes an equality in the limit $\omega\rightarrow
\infty$. Hence each gradient flow curve is contained in the angular
region:
\ben
\label{condnc1}
\theta\in \cC_A\eqdef \left(\theta_A,\pi-\theta_A\right)\cup \left(\pi+\theta_A, 2\pi -\theta_A\right)~~,
\een
where:
\be
\theta_A\eqdef \arcsin(\sigma_A)~~.
\ee
When $\e$ is a flaring end, the left hand side of \eqref{gradflownc}
is smaller than $1/4$, so in this case we have the further condition:
\ben
\label{condnc2}
|\sin\theta|<\sigma'_A\eqdef e^{\frac{1}{4\tc_\e}} e^{-\frac{A}{\tc_\e}}~~\mathrm{for~flaring~ends.}
\een
This is automatically satisfied when $A<\frac{1}{4}$, while for $A\geq
\frac{1}{4}$ it gives a further constraint which excludes a region of
the form $\cR_A$ but with $\theta_A$ replaced by:
\be
\theta'_A\eqdef \arcsin(\sigma'_A)~~.
\ee
Since $\sigma'_A>\sigma_A$, we have $\theta'_A>\theta_A$, so
\eqref{condnc1} and \eqref{condnc2} give:
\be
\theta\in \cF_A \eqdef \left(\theta_A,\theta'_A\right)\cup \left(\pi-\theta'_A,\pi-\theta_A\right)\cup
\left(\pi+\theta_A, \pi+\theta'_A\right) \cup \left(2\pi-\theta'_A, 2\pi -\theta_A\right)~~.
\ee
Thus when $\e$ is a cusp end we have $\theta\in \cC_A$, while
when $\e$ is a flaring end we have:
\be
\theta\in \twopartdef{\cR_A}{A\in (0,1/4)}{\cF_A}{A\in [1/4,\infty)}~~.
\ee
Since the sign $\epsilon_\e$ is fixed for each end, equation
\eqref{gradflownc} (with fixed $A>0$) produces a single solution
$\omega_A(\theta)$ for each value of $\theta$ in the allowed region
and $\omega_A$ is a decreasing and continuous function of
$|\sin\theta|$. The right hand side of \eqref{gradflownc} is invariant
under the transformations $\theta\rightarrow \pi-\theta$ and
$\theta\rightarrow -\theta$, which correspond to reflections in the
coordinate axes. Hence for each $A>0$ we have two gradient flow orbits
for a cusp end, while for a flaring end we have two gradient flow
orbits when $A\in (0,1/4)$ and four gradient flow orbit when $A\in
[1/4,\infty)$. Each orbit is contained in a connected component of the
corresponding allowed region. This collection of orbits is invariant
under axis reflections.

\paragraph{Non-special gradient flow orbits which have $\e$ as a limit point.}

Distinguish the cases:
\begin{enumerate}
\item $\epsilon_\e=+1$, i.e. $\e$ is a flaring end. Then \eqref{gradflownc} becomes:
\be
\frac{1}{4}\bgamma_2\left(\frac{2}{\omega}\right)=A+\tilde{c}_\e \log|\sin\theta|~~.
\ee
Since the left hand side tends to $\frac{1}{4}$ for $\omega\rightarrow
0$ while the right hand side is bounded from above by $A$, it follows
that the gradient flow orbits which reach the end have $A\in
[1/4,\infty)$. For each such $A$, we have four gradient flow orbits
which reach the end at angles equal to $\pm \theta'_A\mod \pi$. For
$A\rightarrow +\infty$, these orbits asymptote near the origin to the
$x$ axis, while for $A\rightarrow 1/4$ they asymptote to the $y$
axis. The orbits obtained for $A\in (0,1/4)$ do not have $\e$ as a
limit point and should be discarded in our approximation.
\item $\epsilon_\e=-1$, i.e. $\e$ is a cusp end. Then
\eqref{gradflownc} becomes:
\be
\frac{1}{4}\bgamma_2\left(-\frac{2}{\omega}\right)=A+\frac{1}{(2\pi)^2} \log|\sin\theta|\leq A~~.
\ee
In this case, the left hand side tends to infinity when $\omega$ tends
to zero so a gradient flow orbit cannot have $\e$ as a limit point
because $A$ is finite. When $A$ tends to infinity, the gradient flow orbits
approach the noncritical cusp closer and closer but they never pass
through it. Along the two gradient flow orbits with integration constant $A$,
the smallest value $\omega_0:=\omega_0(A)$ of $\omega$ is realized for
$\theta=\frac{\pi}{2}$ or $\theta=\frac{3\pi}{2}$ (at the points where
these orbits intersect the $y$ axis) and is the solution of the
equation:
\be
\frac{1}{4}\bgamma_2\left(-\frac{2}{\omega_0}\right)=A~~.
\ee
\end{enumerate} 

\noindent Figure \ref{fig:NonCritQP} shows some unoriented asymptotic
gradient flow orbits close to each type of noncritical end.

\begin{figure}[H]
\centering
\begin{minipage}{.45\textwidth}
\centering ~~\includegraphics[width=.97\linewidth]{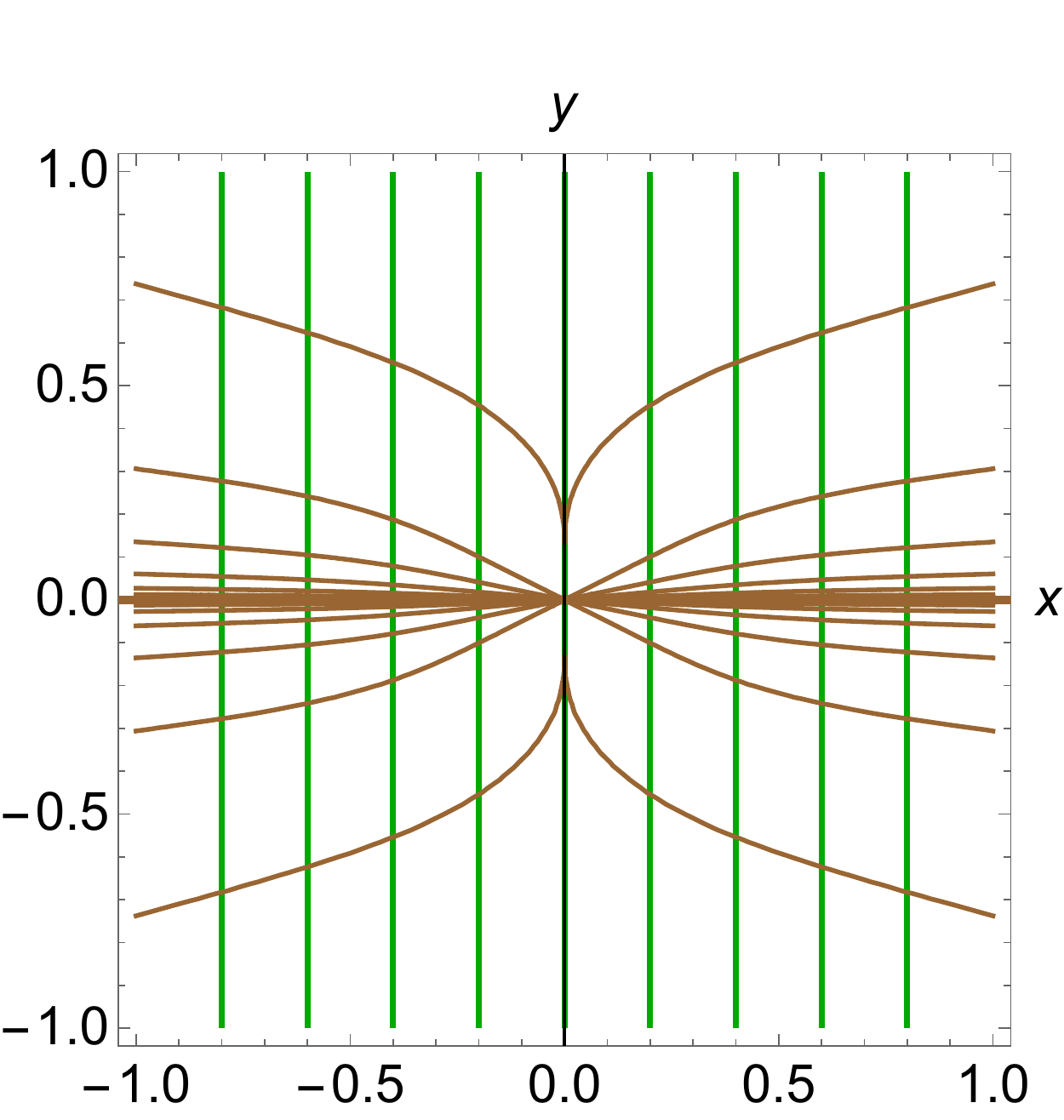}
\subcaption{Plane end.}
\end{minipage}\hfill 
\begin{minipage}{.45\textwidth}
\centering ~~\includegraphics[width=.95\linewidth]{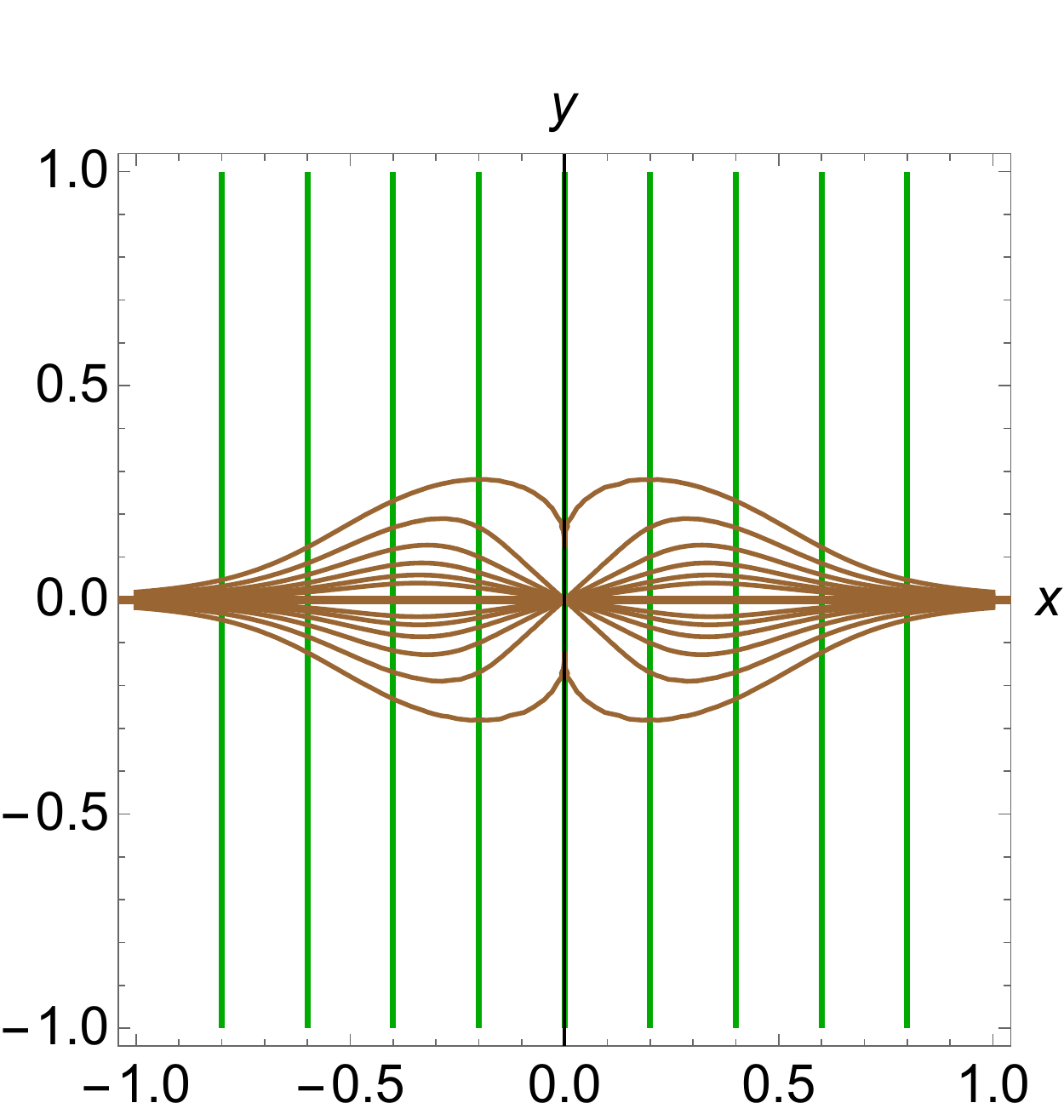}
\subcaption{Horn end.}
\end{minipage}\\
\begin{minipage}{.45\textwidth}
\centering ~~\includegraphics[width=.95\linewidth]{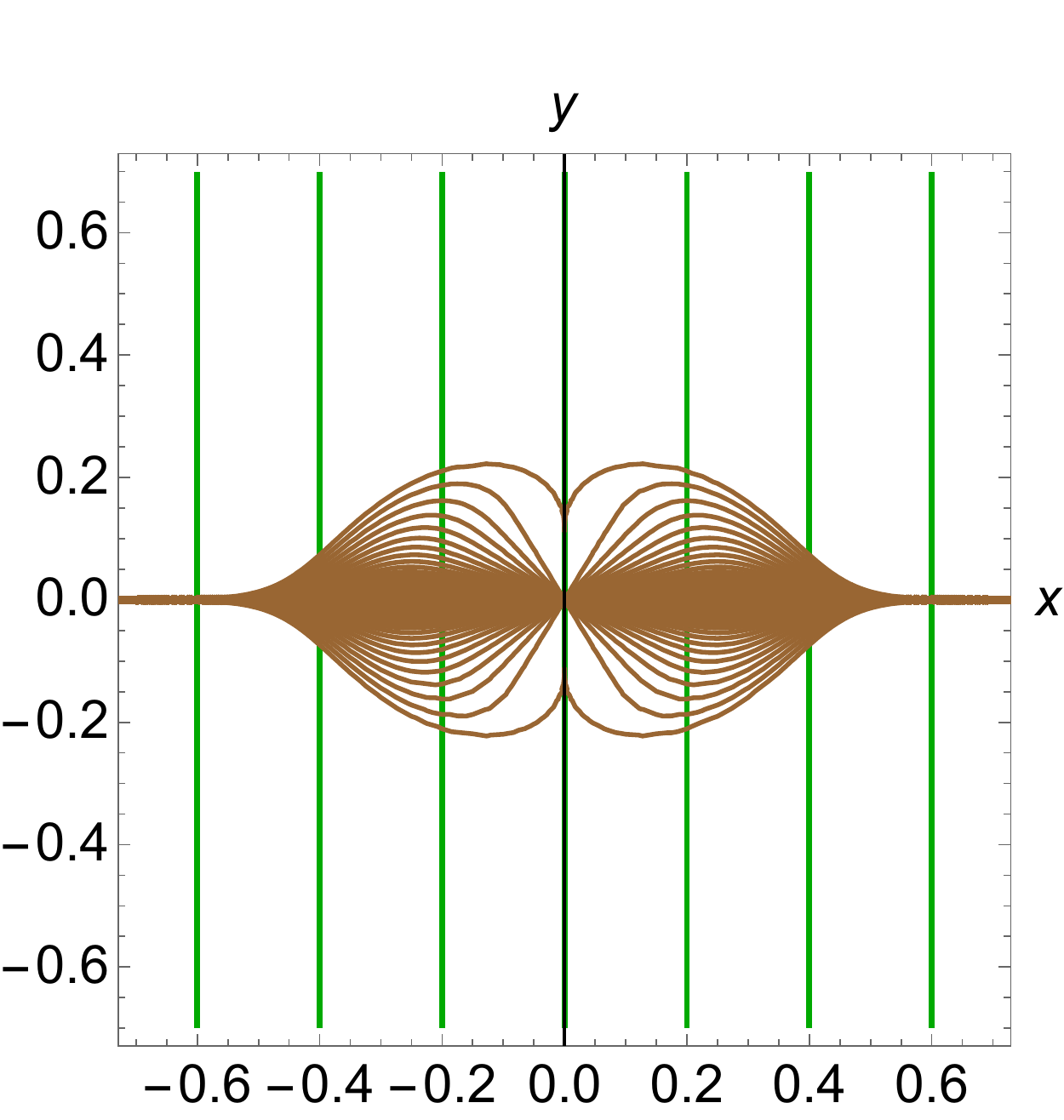}
\subcaption{Funnel end with $\ell=1$.}
\end{minipage}\hfill
\begin{minipage}{.45\textwidth}
\centering ~~ \includegraphics[width=.97\linewidth]{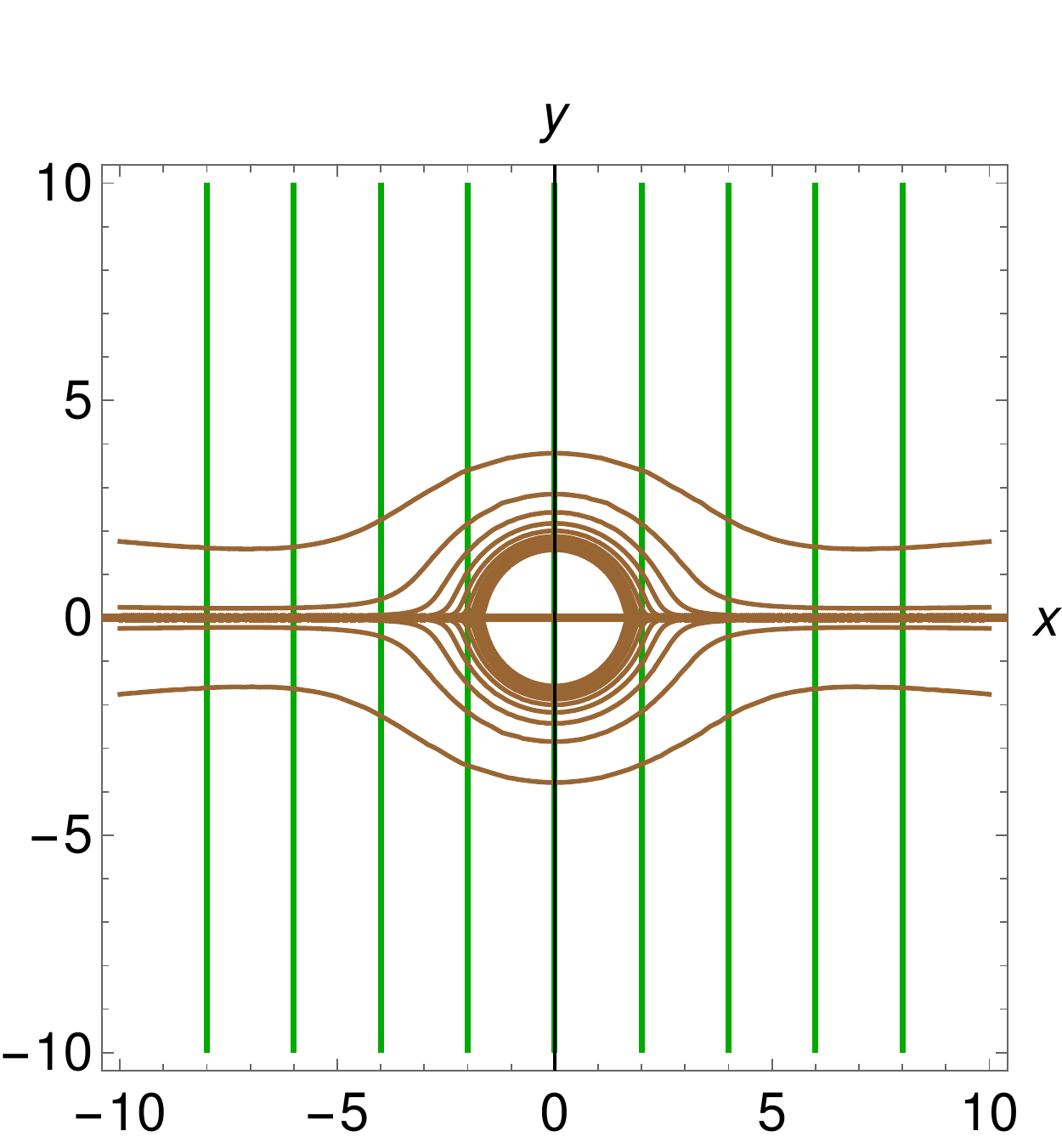}
\subcaption{Cusp end.}
\end{minipage}
\caption{Unoriented gradient flow orbits of $V$ (shown in brown) and
level sets of $V$ (shown in green) near noncritical plane, horn,
funnel and cusp ends in special Cartesian canonical coordinates
centered at the end. For the funnel end we took $\ell=1$. For flaring
ends, we show only orbits which have the end as a limit point (which
correspond to $A\geq 1/4$. The orbits flow from right to left since
$\mu_\e>0$.}
\label{fig:NonCritQP}
\end{figure}
The extended scalar potential $\hPhi$ of the canonical model can be
recovered from the extended classical effective potential as:
\ben
\label{Phic}
\hPhi=\frac{1}{2 M_0^2} \hV^2\approx \frac{1}{2} {\bar \mu}_\e^2\left[\hat{\bar{V}}(\e)+\omega \cos\theta\right]^2~~,
\een
where we defined:
\be
{\bar \mu}_\e\eqdef \frac{\mu_\e}{M_0}~~,~~{\hat {\bar V}}(\e)\eqdef \frac{{\hat V}(\e)}{\mu_\e}~~.
\ee
Figure \ref{fig:NoncritCosm} shows some infrared optimal cosmological
orbits of the uniformized model parameterized by $(M_0,\Sigma,G,\Phi)$
near noncritical ends $\e$; the initial point of each orbit is shown
as a black dot. In this figure, we took ${\bar \mu}_\e=1$, ${\hat
{\bar V}}(\e)=1$ and $M_0=1$. Notice that the accuracy of the first
order IR approximation depends on the value of ${\hat {\bar V}}(\e)$,
since the first IR parameter of \cite{ren} depends on this value.

\begin{figure}[H]
\centering
\begin{minipage}{.45\textwidth}
\centering ~~\includegraphics[width=.97\linewidth]{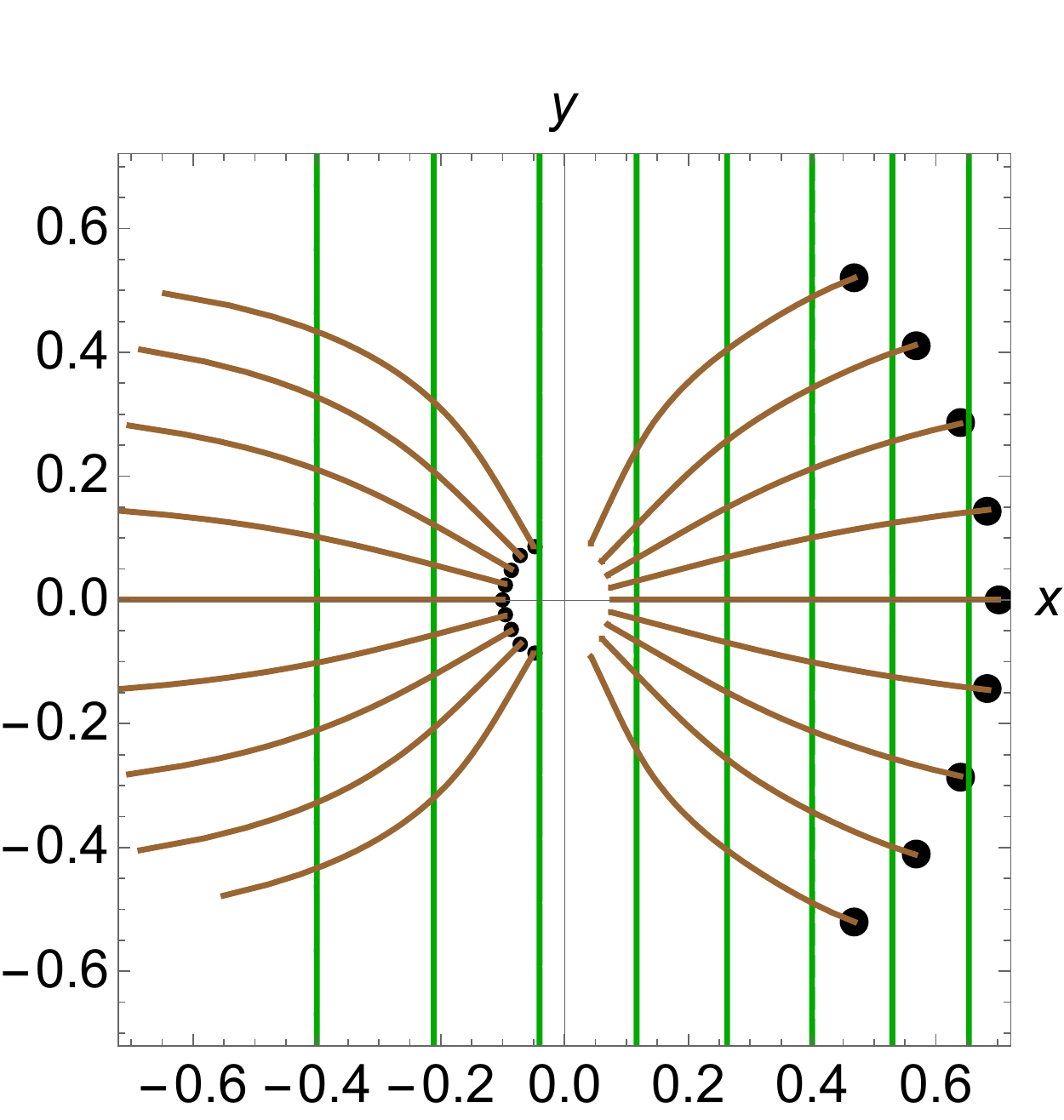}
\subcaption{Plane end.}
\end{minipage}\hfill 
\begin{minipage}{.45\textwidth}
\centering ~~\includegraphics[width=.95\linewidth]{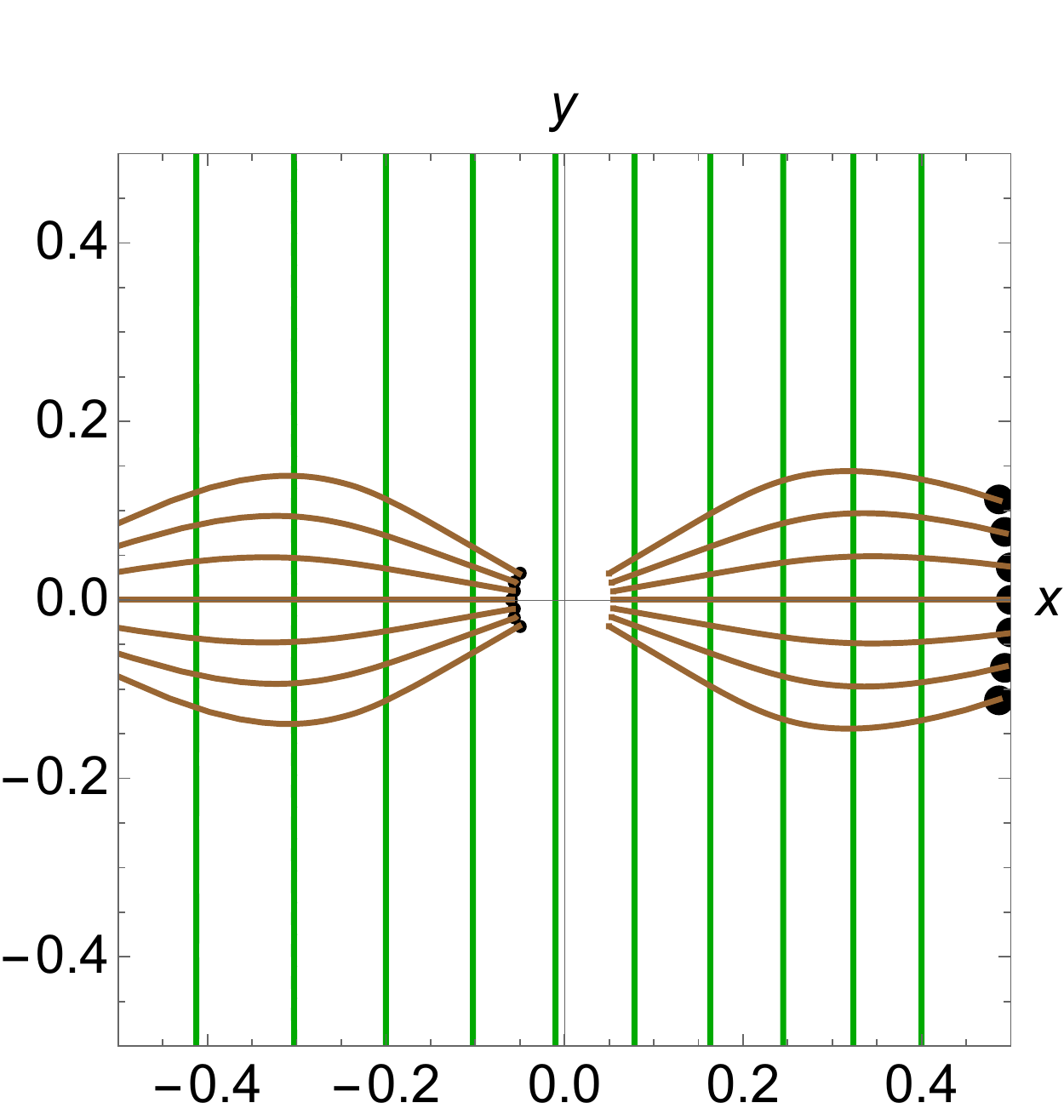}
\subcaption{Horn end.}
\end{minipage}\\
\begin{minipage}{.45\textwidth}
\centering ~~\includegraphics[width=.95\linewidth]{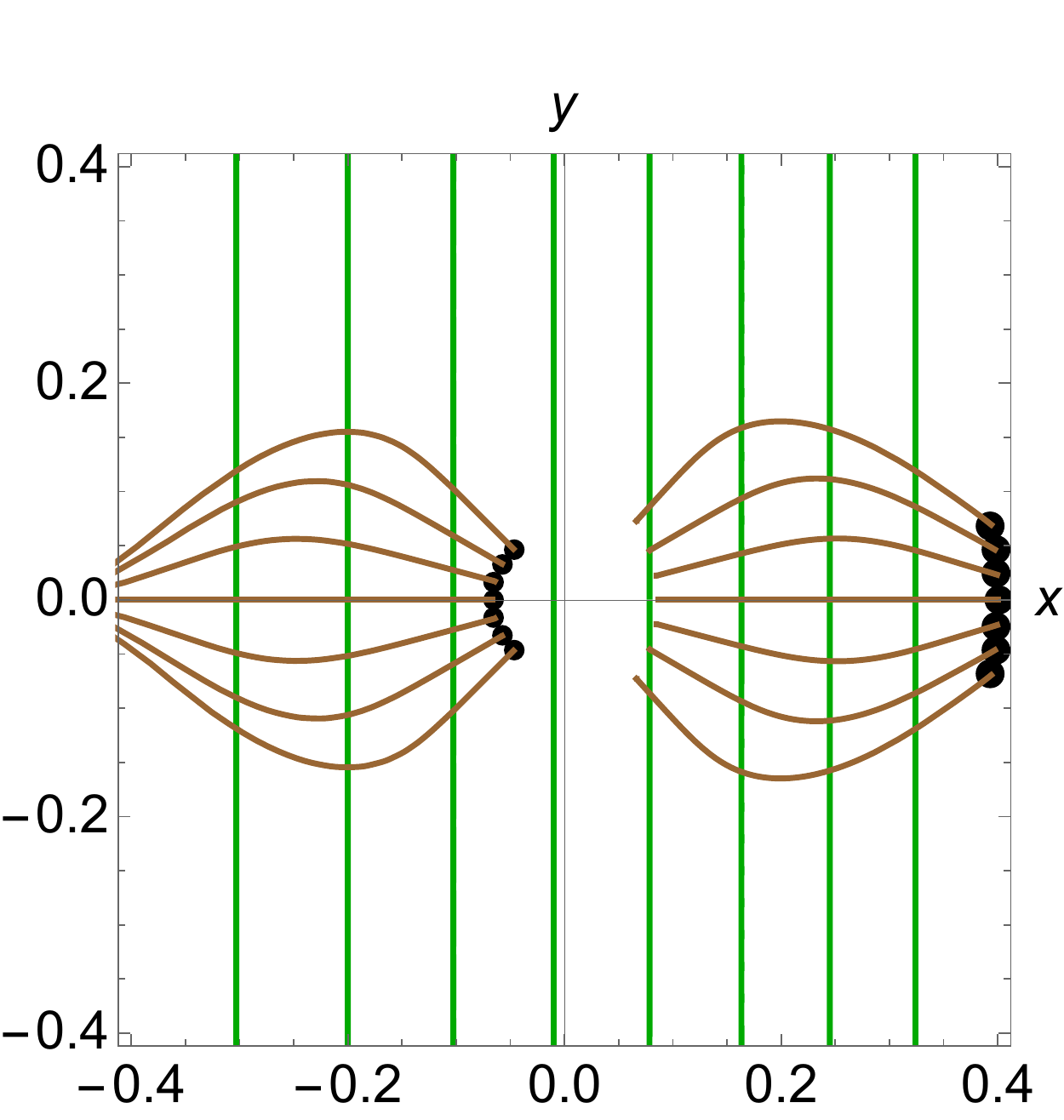}
\subcaption{Funnel end with $\ell=1$.}
\end{minipage}\hfill
\begin{minipage}{.45\textwidth}
\centering ~~ \includegraphics[width=.97\linewidth]{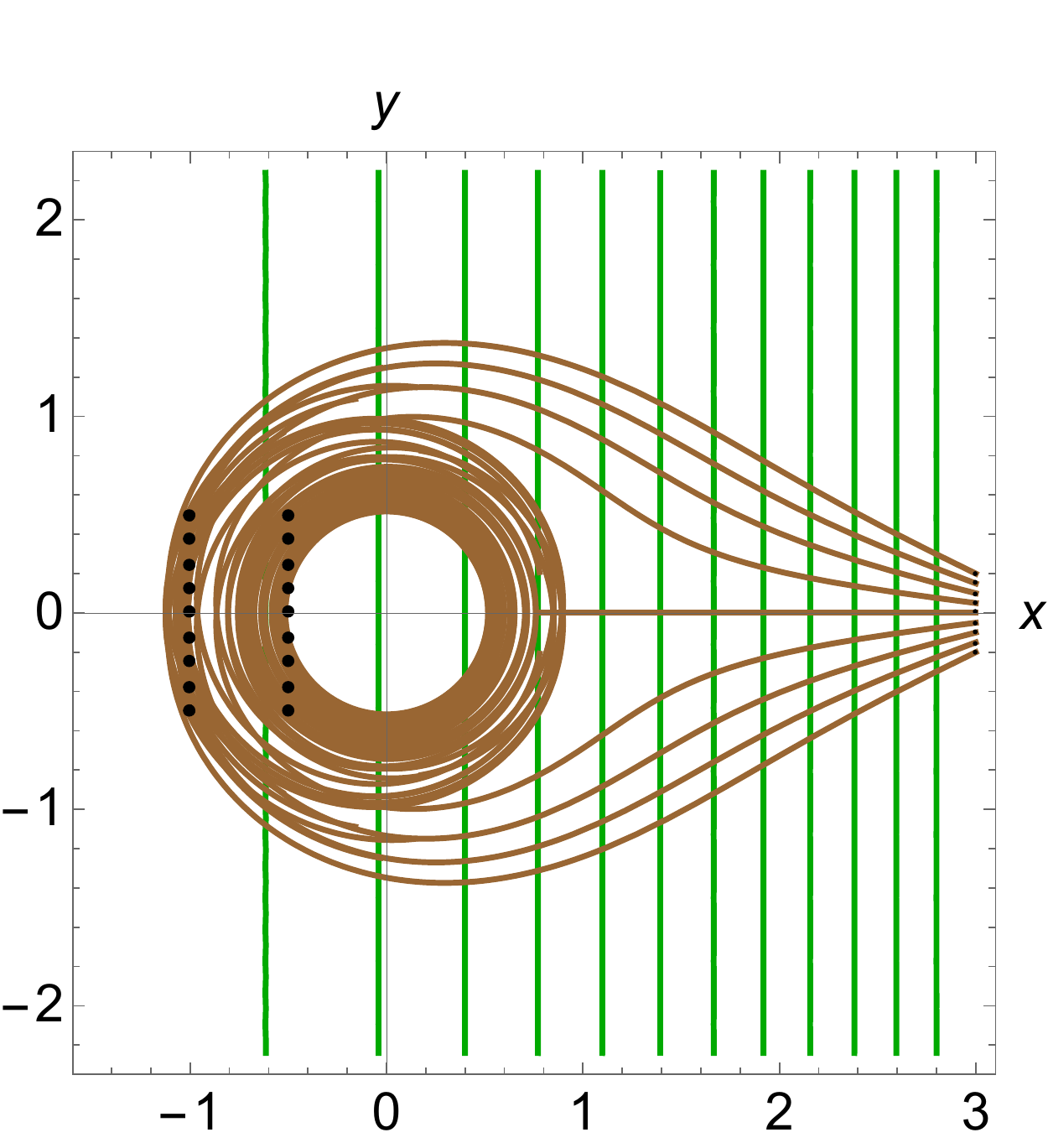}
\subcaption{Cusp end.}
\end{minipage}
\caption{Numerically computed infrared optimal cosmological orbits
(shown in brown) and level sets of $\Phi$ (shown in green) near
noncritical plane, horn, funnel and cusp ends in special Cartesian
canonical coordinates centered at the end. For the funnel end we took
$\ell=1$. The initial values of the corresponding curves are shown as
black dots, while the initial speeds lie in the gradient shell of
$(\Sigma,G,V)$.}
\label{fig:NoncritCosm}
\end{figure}

The case of noncritical cusp ends is particularly interesting. For
clarity, Figure \ref{fig:NoncritCosmCuspDetail} shows the evolution of
a few infrared optimal cosmological curves for four consecutive
cosmological times.

\begin{figure}[H]
\centering
\begin{minipage}{.45\textwidth}
\centering ~~\includegraphics[width=.97\linewidth]{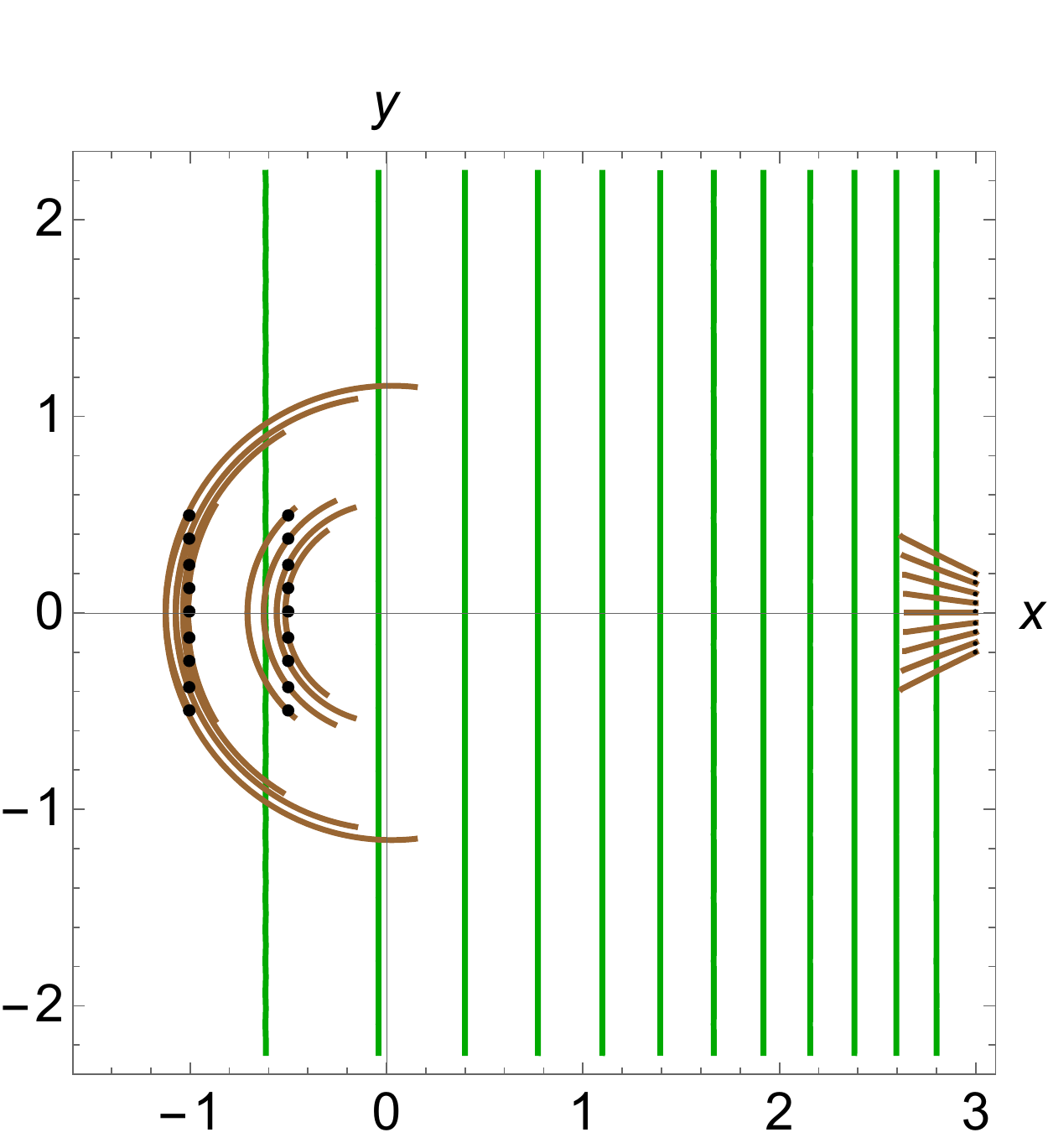}
\subcaption{For $t= 0.01$.}
\end{minipage}\hfill 
\begin{minipage}{.45\textwidth}
\centering ~~\includegraphics[width=.95\linewidth]{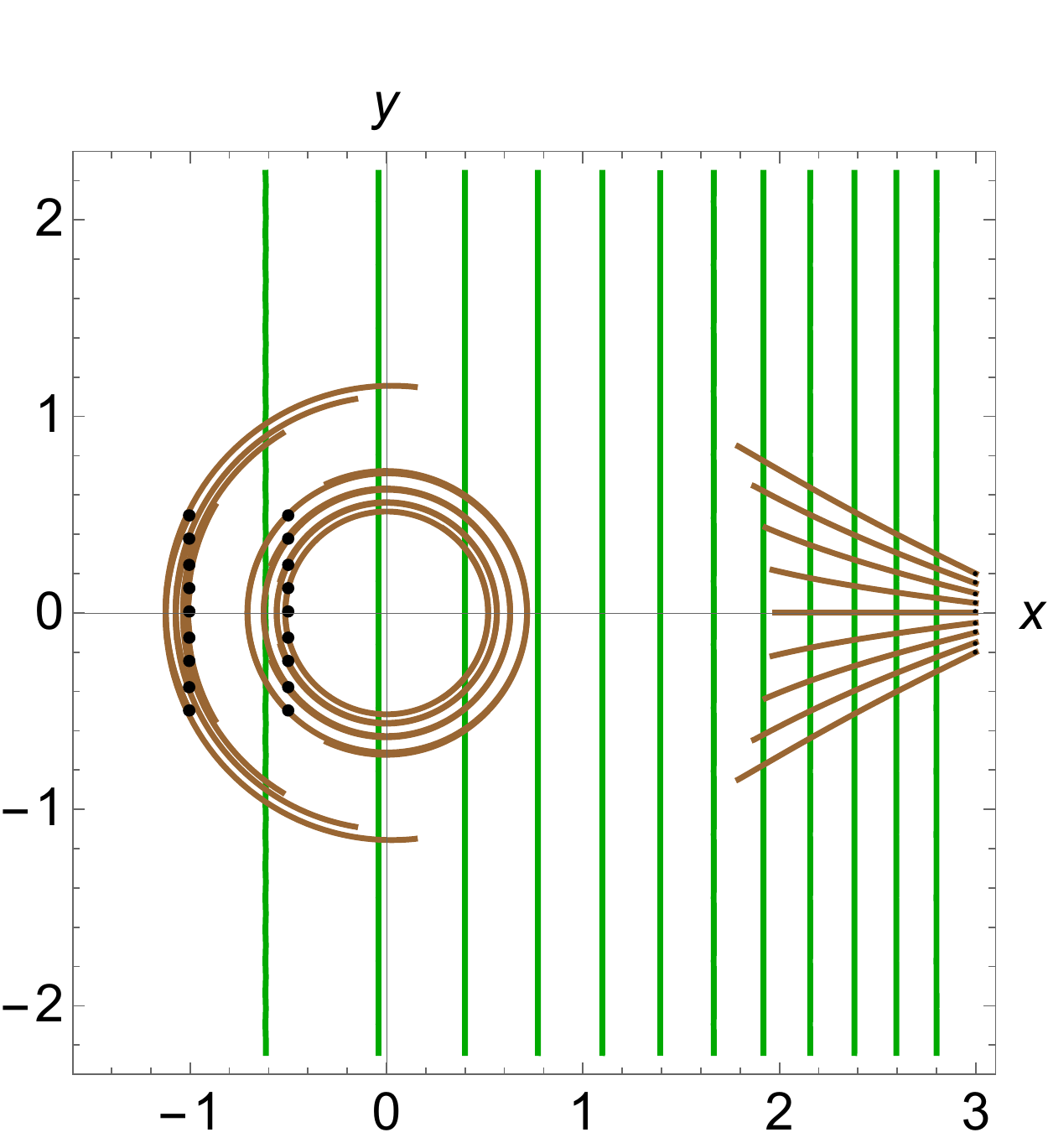}
\subcaption{For $t=0.04$.}
\end{minipage}\\
\begin{minipage}{.45\textwidth}
\centering ~~\includegraphics[width=.95\linewidth]{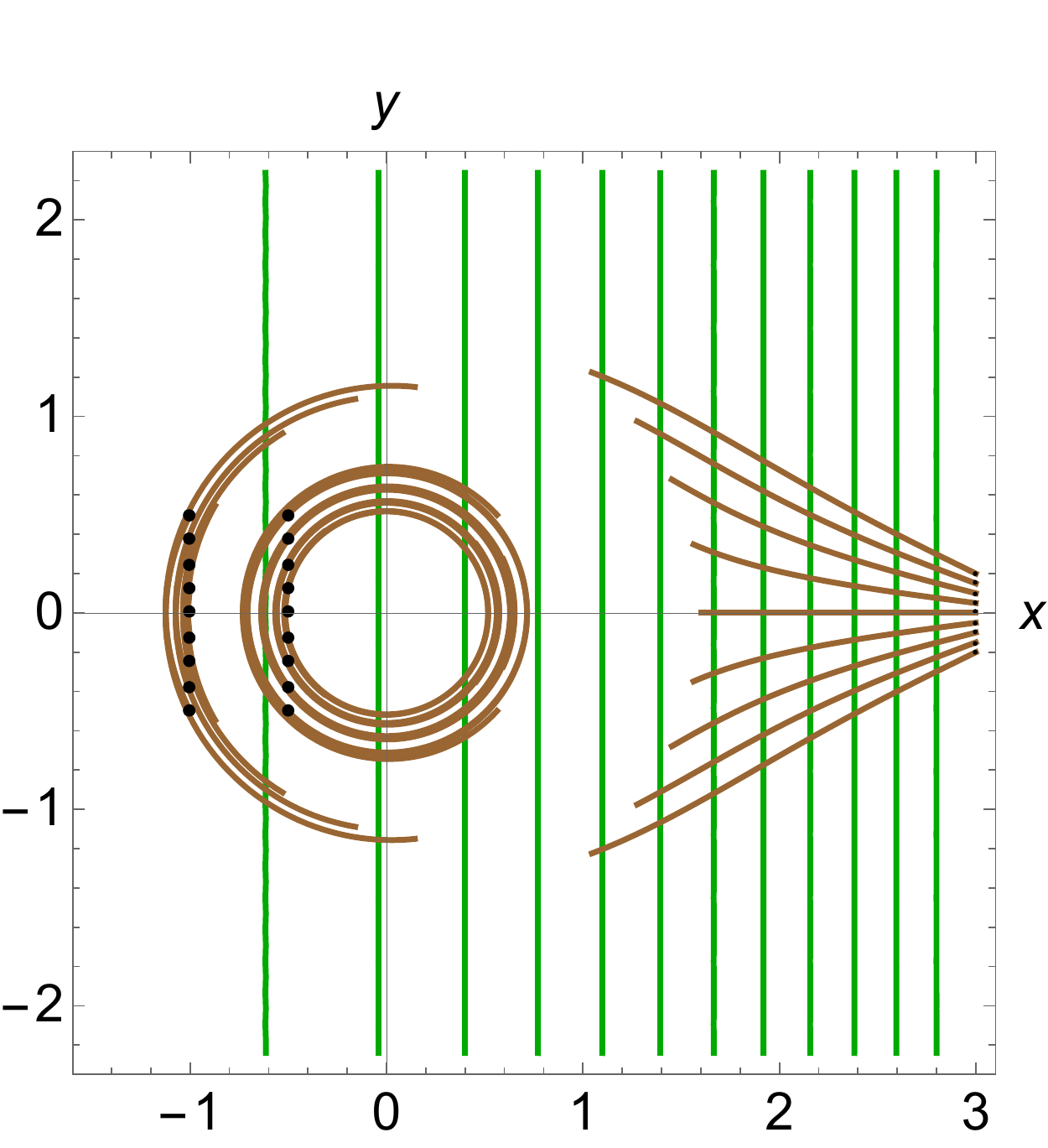}
\subcaption{For $t=0.07$.}
\end{minipage}\hfill
\begin{minipage}{.45\textwidth}
\centering ~~ \includegraphics[width=.97\linewidth]{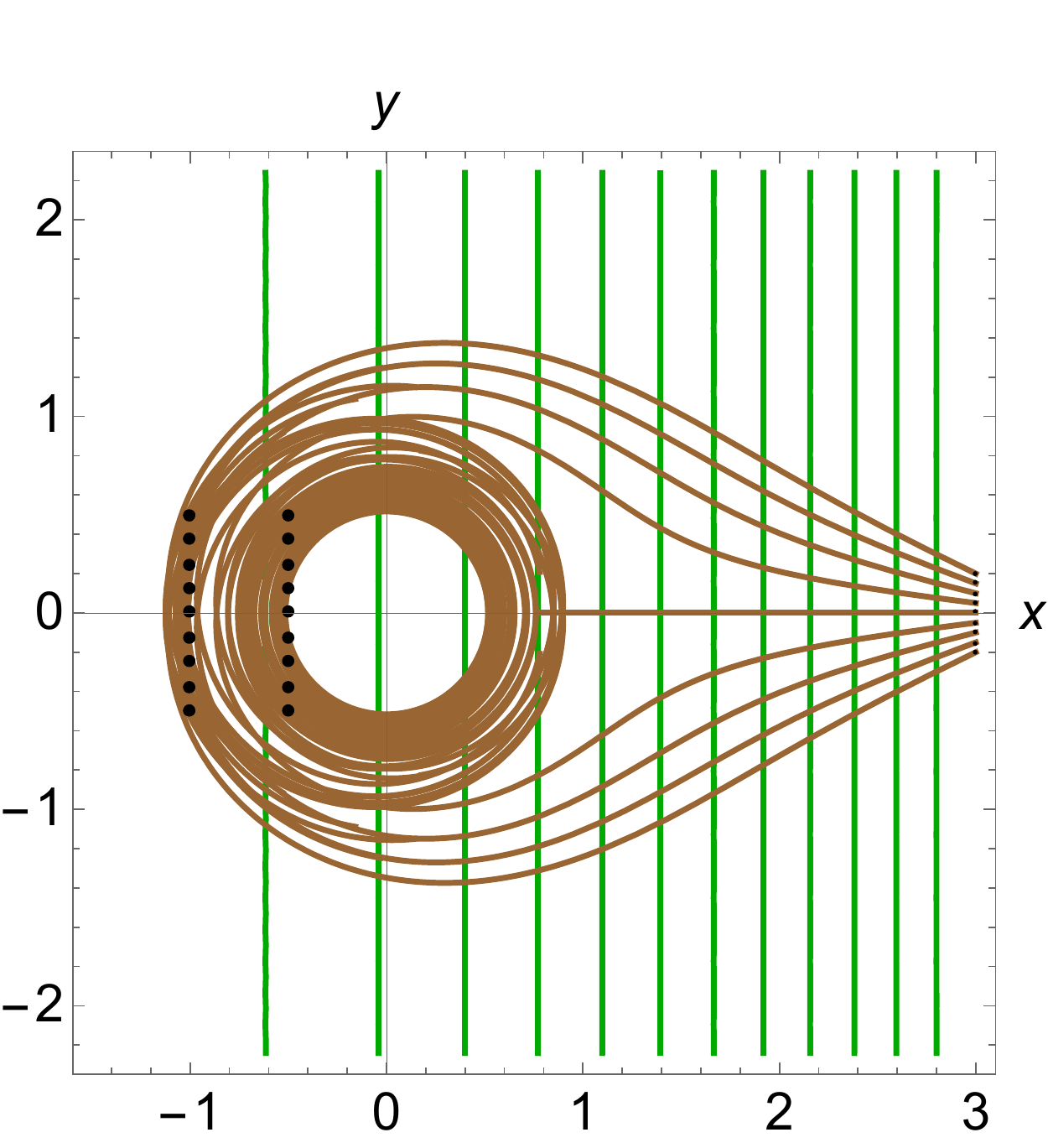}
\subcaption{For $t=0.3$.}
\end{minipage}
\caption{Numerically computed infrared optimal cosmological orbits of
the canonical model (shown in brown) and level sets of $\Phi$ (shown
in green) near a noncritical cusp end, with initial points shown as
black dots and initial speeds taken in the gradient flow shell of
$(\Sigma,G,V)$. The four figures show the orbits for cosmological times
between zero and $t= 0.01$, $0.04$, $0.07$ and $0.3$ respectively.}
\label{fig:NoncritCosmCuspDetail}
\end{figure}

\subsection{Stable and unstable manifolds of noncritical ends under the effective gradient flow}

\noindent The discussion above shows that noncritical flaring ends
behave like fictitious stationary points of the gradient flow of
$(\Sigma,G,V)$ even though Freudenthal ends are not points of $\Sigma$ and
even though $\hV$ does not have a critical point at a non-critical end. 
The stable and unstable manifolds of such an end (see \eqref{SUdefe})
are connected and of dimension two:
\be
\dim \cS(\e)=\dim \cU(\e)=2~~\mathrm{for}~~\e=\mathrm{noncritical~flaring~end}~~,
\ee
a dimension count which differs from that of ordinary hyperbolic fixed
points of dynamical systems. On the other hand, noncritical cusp ends
do not behave like stationary points of the effective gradient flow;
instead, they ``repel'' all gradient flow orbits with the exception of
the two special orbits which lie on the $x$ axis and hence correspond
to geodesic orbits having the cusp as a limit point. These special
orbits form the stable and unstable manifolds of a noncritical cusp
end, which are connected and one-dimensional:
\be
\dim \cS(\e)=\dim \cU(\e)=1~~\mathrm{for}~~\e=\mathrm{noncritical~cusp~end}~~.
\ee
With the exception of the two special orbits, every other effective
gradient flow orbit never reaches the noncritical cusp end. Notice
that the noncritical ends do not act as attractors of the effective
gradient flow.

\section{The IR phases of critical ends}
\label{sec:critends}

\noindent Consider principal Cartesian canonical coordinates $(x,y)$
centered at a critical end $\e$ (see Subsection
\ref{subsec:PrincCritEnds}). In such coordinates, we have:
\beqan
\label{egradas1}
&& (\grad V)^\omega\approx \omega^4 \partial_\omega V\approx \omega^5 \left[\lambda_1(\e) \cos^2 \theta+\lambda_2(\e)\sin^2 \theta\right]\\
&& (\grad V)^\theta\approx \frac{1}{f_\e(1/\omega)}\partial_\theta V\approx \frac{\lambda_2(\e)-\lambda_1(\e)}{{\tilde c}_\e}\omega^2 e^{-\frac{2\epsilon_\e}{\omega}} \sin\theta\cos\theta~~.\nn
\eeqan

\subsection{Special gradient flow orbits}

\noindent For $\theta\in \{0,\frac{\pi}{2},\pi,\frac{3\pi}{2}\}$, the
gradient flow equation of $(\Sigma,G,V)$ reduces in leading order near
$\e$ to $\theta=\const$ and:
\be
\frac{\dd\omega}{\dd q}=\twopartdef{\lambda_1(\e)\omega^5}{\theta\in \{0,\pi\}}{\lambda_2(\e)\omega^5}{\theta\in \{\frac{\pi}{2},\frac{3\pi}{2}\}}~~,
\ee
with general solution:
\be
\omega=\twopartdef{\frac{1}{[4\lambda_1(\e)(q_0-q)]^{1/4}}}{\theta\in \{0,\pi\}~(\mathrm{with}~q<q_0)}{\frac{1}{[4\lambda_2(\e)(q_0-q)]^{1/4}}}{\theta\in \{\frac{\pi}{2},\frac{3\pi}{2}\}~(\mathrm{with}~q<q_0)}~~.
\ee
Shifting $q$ by $q_0$ brings this to the form:
\be
\omega=\twopartdef{\frac{1}{[4\lambda_1(\e)|q|]^{1/4}}}{\theta\in \{0,\pi\}~(\mathrm{with}~q<0)}{\frac{1}{[4\lambda_2(\e)|q|]^{1/4}}}{\theta\in \{\frac{\pi}{2},\frac{3\pi}{2}\}~(\mathrm{with}~q<0)}~~.
\ee
This gives four gradient flow orbits which tend to $\e$ for
$q\rightarrow -\infty$ and asymptote to the principal
geodesic orbits near $\e$. 

\subsection{Non-special gradient flow orbits}

\noindent For $\theta\not \in \{0,\frac{\pi}{2},\pi,\frac{3\pi}{2}\}$, the
gradient flow equation reduces to:
\be
\frac{\dd \omega}{\dd \theta}=\frac{{\tilde c}_\e}{\lambda_2(\e)-\lambda_1(\e)} \omega^3 e^{\frac{2\epsilon_\e}{\omega}}\left[\lambda_1(\e)\cot\theta+\lambda_2(\e)\tan\theta\right]~~,
\ee
which can be written as:
\ben
\label{Pfaff}
[\lambda_2(\e)-\lambda_1(\e)]\omega^{-3}  e^{-\frac{2\epsilon_\e}{\omega}}\dd \omega={\tilde c}_\e \left[\lambda_1(\e)\cot\theta+\lambda_2(\e)\tan\theta\right]\dd\theta~~.
\een
Setting $v\eqdef \frac{2\epsilon_\e}{\omega}$ brings \eqref{Pfaff} to
the form:
\ben
\label{PfaffFlaring}
\frac{1}{4}[\lambda_1(\e)-\lambda_2(\e)]v  e^{-v}\dd v={\tilde c}_\e \left[\lambda_1(\e)\cot\theta+\lambda_2(\e)\tan \theta\right]\dd\theta~~.
\een
We have $v e^{-v}\dd v=\dd \bgamma_2(v)$, where $\bgamma_2(v)$ is the
lower incomplete Gamma function of order $2$ (see \eqref{gamma}).
Hence \eqref{PfaffFlaring} gives the following implicit equation for
the asymptotic gradient flow orbits near a critical end $\e$:
\ben
\label{gammaflow}
\frac{1}{4}[\lambda_1(\e)-\lambda_2(\e)]\bgamma_2\left(\frac{2\epsilon_\e}{\omega}\right)=A+{\tilde c}_\e\left[\,\lambda_1(\e)\log|\sin\theta|-\lambda_2(\e)\log|\cos\theta |\,\right]~~,
\een
where $A$ is an integration constant. Since:
\be
\bgamma_2(v)=1-e^{-v}-ve^{-v}~~,
\ee
this can be written explicitly as:
\ben
\label{eSol0}
\frac{1}{4}[\lambda_1(\e)-\lambda_2(\e)]\left[1-e^{-\frac{2\epsilon_\e}{\omega}}-\frac{2\epsilon_\e}{\omega}e^{-\frac{2\epsilon_\e}{\omega}}\right]=
A+{\tilde c}_\e \left[\,\lambda_1(\e)\log|\sin\theta|-\lambda_2(\e)\log|\cos\theta |\,\right]~~.
\een
Distinguish the cases:
\begin{enumerate}[1.]
\itemsep 0.0em
\item $V$ is circular at $\e$,
i.e. $\lambda_1(\e)=\lambda_2(\e):=\lambda(\e)$ (which amounts to
$\beta_\e=1$). Then \eqref{gammaflow} becomes:
\ben
\label{eqri}
\log|\sin\theta|-\log|\cos\theta|=B~~\mathrm{i.e.}~~\tan\theta=
\pm e^B\Longleftrightarrow \theta=\pm \arctan(e^B) \!\!\! \mod \pi
\een
with $B\eqdef -\frac{A}{\lambda(\e){\tilde c}_\e}$ and hence $\theta$
is constant for all asymptotic gradient flow curves. In this case, the
asymptotic gradient flow orbits near $\e$ are geodesic orbits of
$(\Sigma,G)$ having $\e$ as a limit point. For each value of $B$,
there are exactly four such orbits.
\item $V$ is not circular at $\e$, i.e. $\beta_\e\neq 1$. Then
\eqref{gammaflow} can be written as:
\ben
\label{eSolBeta}
\bgamma_2\left(\frac{2\epsilon_\e}{\omega}\right)
=1-e^{-\frac{2\epsilon_\e}{\omega}}-\frac{2\epsilon_\e}{\omega}e^{-\frac{2\epsilon_\e}{\omega}}=C-\frac{4{\tilde c}_\e}{1-\beta_\e}\left(\beta_\e \log|\sin\theta|-\log|\cos\theta |\right)~~,
\een
where $C\eqdef\frac{4A}{\lambda_1(\e)-\lambda_2(\e)}$.
\end{enumerate}

\noindent Thus:

\begin{prop}
The unoriented orbits of the asymptotic gradient flow of
$(\Sigma,G,V)$ near a critical end $\e$ are determined by the
hyperbolic type of the end (i.e. by $\epsilon_\e$ and ${\tilde c}_\e$)
and by the critical modulus $\beta_\e$, while the orientation of the
orbits is determined by the critical signs $\epsilon_i(\e)$, which
satisfy $\epsilon_1(\e)\epsilon_2(\e)=1$.
\end{prop}

\paragraph{Asymptotic sectors for non-special gradient flow curves when $\lambda_1(\e)\neq \lambda_2(\e)$.}

\noindent Let us assume that $\lambda_1(\e)\neq \lambda_2(\e)$
i.e. that $\beta_\e<1$. Then \eqref{eSolBeta} reads:
\ben
\label{eSB}
\bgamma_2\left(\frac{2\epsilon_\e}{\omega}\right)=C-\frac{4{\tilde c}_\e}{1-\beta_\e} H(\theta;\beta_\e)~~,
\een
where $H:\rS^1\times [-1,1) \rightarrow \overline{\R}$ is the function defined through:
\be
H(\theta;\beta)\eqdef \beta \log|\sin\theta|-\log|\cos\theta |~~.
\ee
This function satisfies:
\be
H(-\theta,\beta)=H(\pi-\theta;\beta)=H(\theta,\beta)~~.
\ee
Notice that:
\be
\bgamma_2(v) \in \twopartdef{(0,1)}{~v>0}{(0,+\infty)}{~v<0}
\ee
and:
\be
\bgamma_2(-\infty)=+\infty~~,~~\bgamma_2(0)=0~~,~~\bgamma_2(+\infty)=1~~.
\ee
Moreover:
\be
\bgamma_2'(v)=ve^{-v} \threepartdef {>0}{~v>0}{=0}{~v=0}{<0}{~v<0}~~,
\ee
hence $\bgamma_2(v)$ is strictly decreasing for $v<0$ and strictly increasing
for $v>0$. Since $\omega>0$, we have:
\be
v:=\frac{2\epsilon_\e}{\omega}\in \twopartdef{>0}{~\e=\mathrm{flaring~end}}{<0}{~\e=\mathrm{cusp~end}}~~,
\ee
which gives:
\ben
\label{gamma2im}
\bgamma_2\left(\frac{2\epsilon_\e}{\omega}\right)\in \twopartdef{(0,1)}{~\e=\mathrm{flaring~end}}{(0,+\infty)}{~\e=\mathrm{cusp~end}}~~. 
\een
Hence \eqref{eSB} requires:
\ben
\label{Crange}
C-\frac{4{\tilde c}_\e}{1-\beta_\e} H(\theta;\beta_\e)\in \twopartdef{(0,1)}{~\e=\mathrm{flaring~end}}{(0,+\infty)}{~\e=\mathrm{cusp~end}}~~. 
\een
Since $H(\theta;\beta)$ tends to $\pm \infty$ for $\theta
\,\mathrm{mod}\, 2\pi\in \{0,\frac{\pi}{2},\pi,\frac{3\pi}{2}\}$,
these values of $\theta$ cannot be attained along any asymptotic
gradient flow orbit. Hence all such orbits are contained in the
complement of the principal coordinate axes in the $(x,y)$-plane,
which means that they cannot meet the principal geodesic orbits close
to the end $\e$. For fixed $\beta$, the function $H(\theta,\beta)$ is
invariant under the action of the Klein four-group $\Z_2\times \Z_2$
generated by the reflections $\theta\rightarrow -\theta$ and
$\theta\rightarrow\pi -\theta$ with respect to the $x$ and $y$ axes;
in particular, this function is periodic of period $\pi$ and its
restriction to the interval $(0,\pi)$ is symmetric with respect to
$\frac{\pi}{2}$, hence it suffices to study the restriction of $H$ to
the interval $\theta\in (0,\frac{\pi}{2})$. Noticing that:
\be
\lim_{\theta\rightarrow 0}H(\theta,\beta)=-\sign(\beta)\infty~~,~~\lim_{\theta\rightarrow \frac{\pi}{2}}H(\theta,\beta)=+\infty
\ee
we distinguish the cases: 
\begin{enumerate}[A.]
\itemsep 0.0em
\item $\beta\in [-1,0)$. For $\theta\in (0,\pi/2)$, we have:
\be
\frac{\dd H(\theta;\beta)}{\dd \theta}=\beta\cot\theta+\tan\theta \threepartdef{<0}{~~0\leq \theta<\theta_0(\beta)}{=0}{~~\theta=\theta_0(\beta)}{>0}{~~\theta_0(\beta)\leq \theta\leq \frac{\pi}{2}}~~,
\ee
where $\theta_0(\beta)\eqdef \arctan\sqrt{|\beta|}$. In this case,
$H(\theta;\beta)$ tends to $+\infty$ at the endpoints of the interval
$(0,\frac{\pi}{2})$ and attains its minimum within this interval at
the point $\theta_0(\beta)$, the minimum value being given by:
\be
\mu(\beta)\eqdef H(\theta_0(\beta),\beta)= \frac{1}{2}\Big[(1+|\beta|)\log (1+|\beta|)-|\beta|\log|\beta|\Big]~~.
\ee
Notice that $\mu$ is an increasing function of $|\beta|$ (hence a decreasing function of $\beta$) and that we have:
\be
\mu(\beta)\in (0,\log2]~~\mathrm{with}~~\mu(-1)=\log2~~,~~\mu(0^-)=0~~.
\ee
It follows that $H(\theta,\beta)$ tends to $+\infty$ for $\theta \!\in\!
\{0,\frac{\pi}{2},\pi,\frac{3\pi}{2}\}$ and attains its minimum
$\mu(\beta)$ on $\rS^1$ for $\theta \in
\{\theta_0(\beta),\pi-\theta_0(\beta),\theta_0(\beta)+\pi,2\pi-\theta_0(\beta)\}$
(see Figure \ref{Hbm1}).

\begin{figure}[H]
\centering
\begin{minipage}{.48\textwidth}
\centering ~~~~ \includegraphics[width=.98\linewidth]{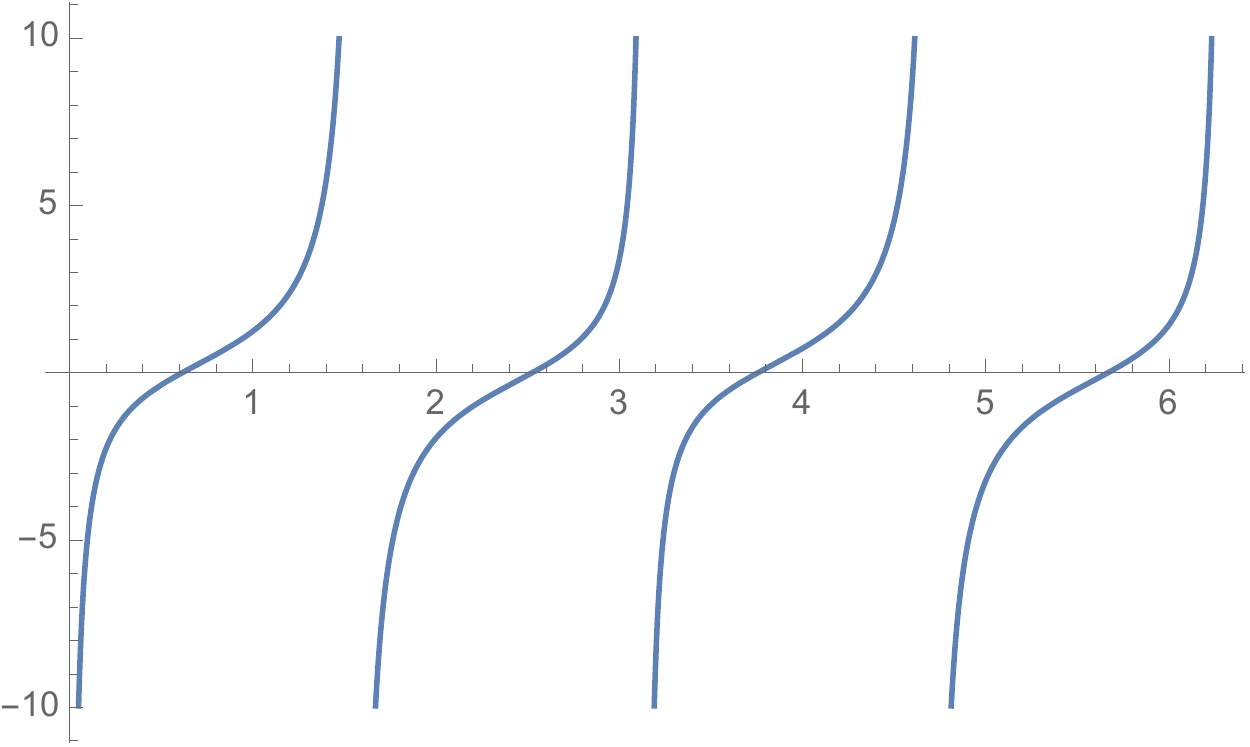}
\vskip 1.2em  
\subcaption{$\frac{\dd H(\theta;\beta)}{\dd \theta}$ for $\beta=-1/2$ and $\theta\in(0,2\pi)$.}
\end{minipage}\hfill 
\begin{minipage}{.48\textwidth}
\centering \includegraphics[width=.98\linewidth]{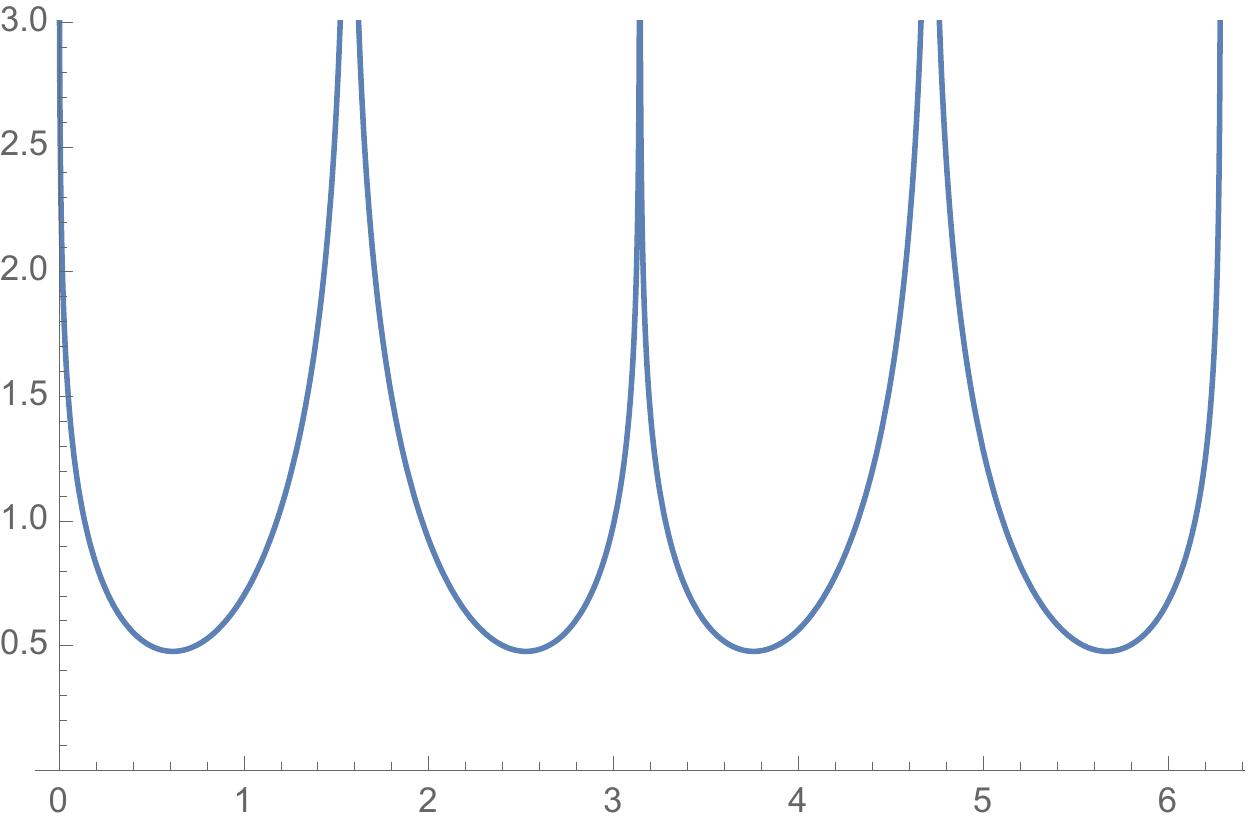}
\subcaption{$H(\theta;\beta)$ for $\beta=-1/2$ and $\theta\in (0,2\pi)$.}
\end{minipage}
\caption{Plots of $\frac{\dd H(\theta;\beta)}{\dd \theta}$ and $H(\theta;\beta)$ for 
$\beta=-1/2$ and $\theta\in (0,2\pi)$.}
\label{Hbm1}
\end{figure}

\item $\beta\in (0,1)$. Then the derivative:
\be
\frac{\dd H(\theta;\beta)}{\dd \theta}=\left(\beta+\tan^2\theta\right)\cot\theta
\ee
is strictly positive for $\theta \in
(0,\frac{\pi}{2})\cup(\pi,\frac{3\pi}{2})$ and strictly negative for
$\theta \in
(\frac{\pi}{2},\pi)\cup(\frac{3\pi}{2},2\pi)$. Hence $H(\theta,\beta)$
increases strictly from $-\infty$ to $+\infty$ along the
$\theta$-intervals $(0,\frac{\pi}{2})$ and $(\pi,\frac{3\pi}{2})$ and
decreases strictly from $+\infty$ to $-\infty$ along the
$\theta$-intervals $(\frac{\pi}{2},\pi)$ and
$(\frac{3\pi}{2},2\pi)$. Thus $H$ tends to $-\infty$ for $\theta\in
\{0,\pi\}$ and to $+\infty$ for $\theta \in
\{\frac{\pi}{2},\frac{3\pi}{2}\}$ and is strictly monotonous on the
circle intervals separating these four special points of $\rS^1$ (see
Figure \ref{Hbp1}).

\begin{figure}[H]
\centering
\begin{minipage}{.48\textwidth}
\centering ~~~~ \includegraphics[width=.97\linewidth]{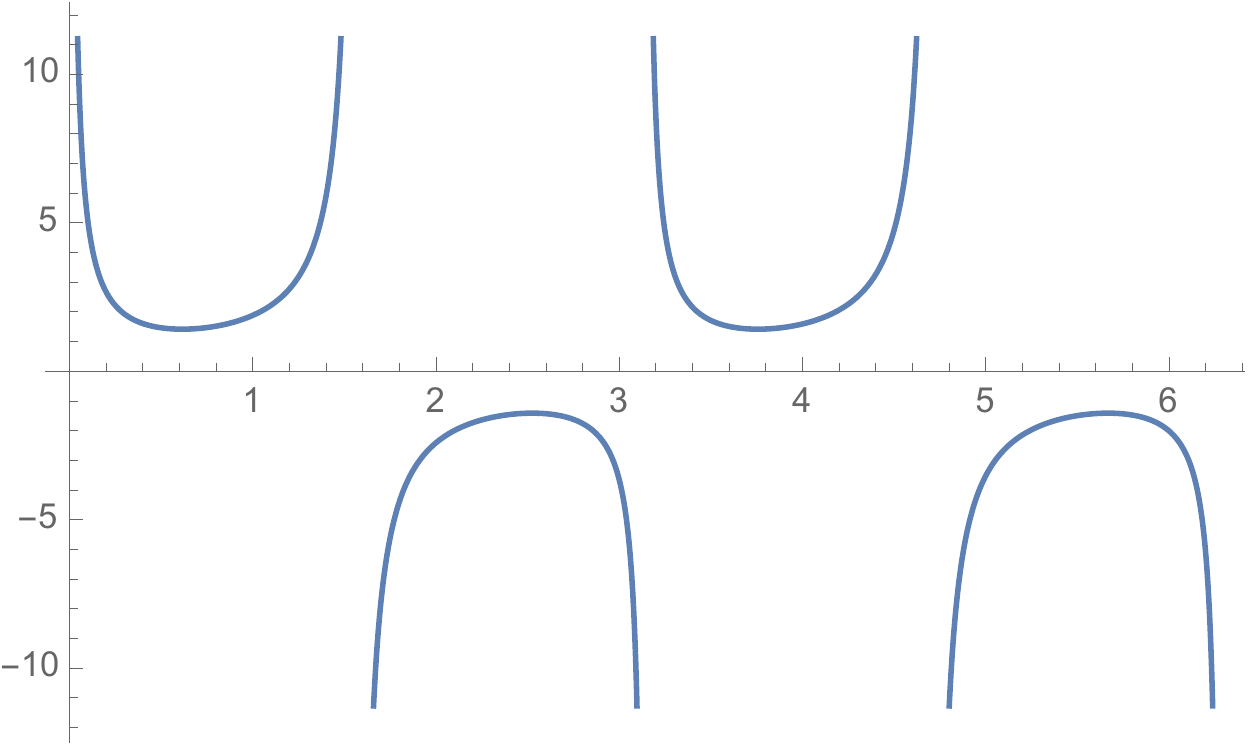}
\subcaption{$\frac{\dd H(\theta;\beta)}{\dd \theta}$ for $\beta=1/2$ and $\theta \in (0,2\pi)$.}
\end{minipage}\hfill 
\begin{minipage}{.48\textwidth}
\centering \includegraphics[width=.98\linewidth]{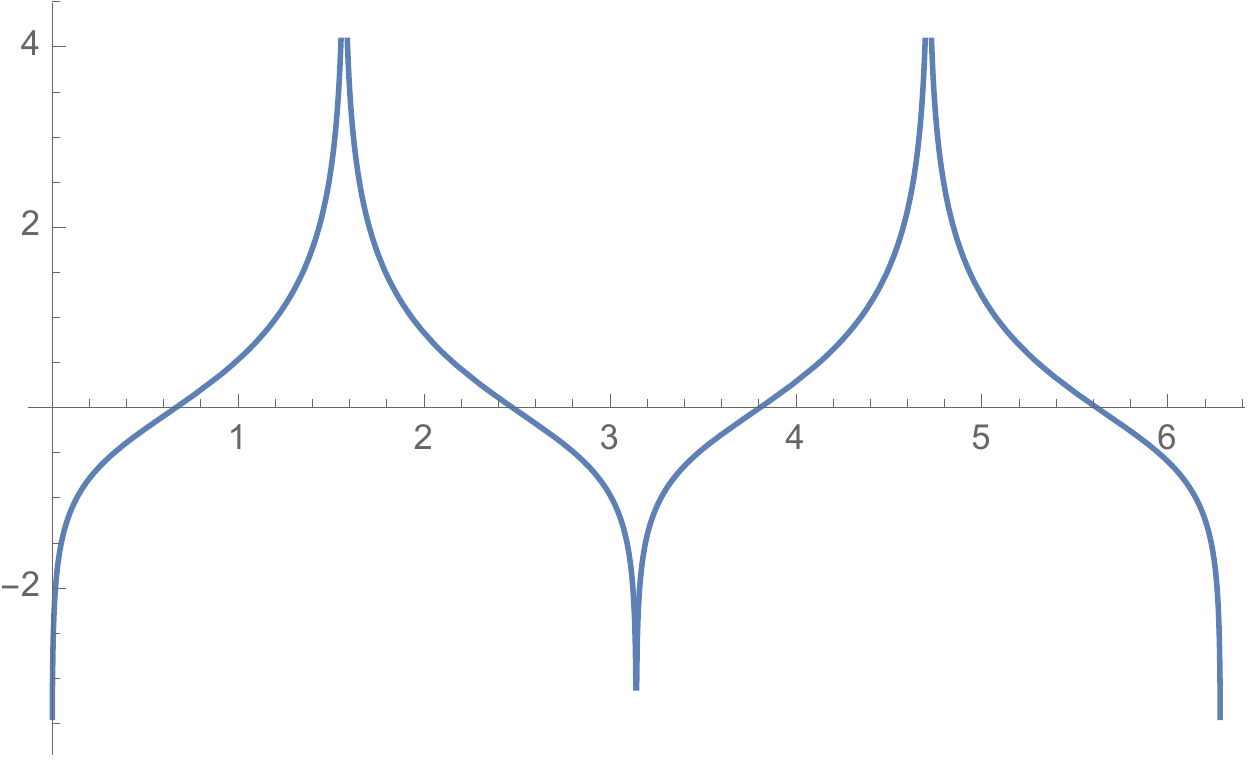}
\subcaption{$H(\theta;\beta)$ for $\beta=1/2$ and $\theta \in (0,2\pi)$.}
\end{minipage}
\caption{Plots of $\frac{\dd H(\theta;\beta)}{\dd \theta}$ and $H(\theta;\beta)$ for $\beta=1/2$ and $\theta \in (0,2\pi)$.}
\label{Hbp1}
\end{figure}
\end{enumerate}

\noindent Returning to condition \eqref{Crange}, we distinguish the
cases:

\begin{enumerate}[1.]
\itemsep 0.0em
\item $\e$ is a flaring end. Then \eqref{Crange} requires
$C-\frac{4{\tilde c}_\e}{1-\beta_\e} H(\theta;\beta_\e)\in (0,1)$.  We
have two sub-cases:
\begin{enumerate}[(a)]
\itemsep 0.0em
\item $\beta_\e\in [-1,0)$. Then $H(\theta,\beta_\e)\geq \mu(\beta_\e)$ 
and $C$ is constrained by the condition:
\ben
\label{CuspCNegBeta}
C>\frac{4{\tilde c}_\e}{1-\beta_\e}\mu (\beta_\e)~~.
\een
Each gradient flow orbit corresponds to a fixed value of $C$. On any such
orbit, condition \eqref{Crange} gives:
\ben
\label{Hcond}
\frac{4{\tilde c}_\e}{1-\beta_\e} H(\theta;\beta_\e)\in (C-1,C)~~.
\een
We have two possibilities:
\begin{itemize}
\item When $C-1\geq\frac{4{\tilde c}_\e}{1-\beta_\e}\mu(\beta_\e)$,
condition \eqref{Hcond} constrains $\theta$ to lie in a union of eight
disjoint open intervals on $\rS^1$. These eight intervals divide into
four successive pairs, where both intervals of each pair are contained
in one of the four quadrants and the four pairs are related by the
action of $\Z_2\times \Z_2$. Hence for each fixed value of $C$ we have
eight orbits in the $(x,y)$ plane, which arrange into four pairs
lying in the four quadrants; the pairs are related to each other by
the action of $\Z_2\times \Z_2$. We will see below that these eight
orbits have $\e$ as a limit point. 
\item When $C-1<\frac{4{\tilde c}_\e}{1-\beta_\e}\mu(\beta_\e)$, the
condition above constrains $\theta$ to lie in a union of four disjoint
open intervals on $\rS^1$ (each lying in a different quadrant) which
is invariant under the action of $\Z_2\times \Z_2$. In this case, we
have four orbits (one in each quadrant) which are related by this
action. These orbits do not have $\e$ as a limit point (see below). 
\end{itemize}

\item $\beta_\e\in (0,1)$. Then $C$ is unconstrained since
$H(\theta,\beta)$ is not bounded from below or from above. On any
gradient flow orbit corresponding to $C$, condition \eqref{Hcond} must
be satisfied. This constrains $\theta$ to lie in a union of four
disjoint open circular intervals which is invariant under the action
of the Klein four-group on $\rS^1$. Hence for each $C$ we have four
gradient flow orbits (one in each quadrant) which are related by the
action of $\Z_2\times \Z_2$. These four orbits have $\e$ as a limit
point (see below).

\end{enumerate}
\item $\e$ is a cusp end. In this case, condition \eqref{Crange}
  requires $\frac{4{\tilde c}_\e}{1-\beta_\e} H(\theta;\beta_\e)<C$.
  We have two sub-cases:
\begin{enumerate} \itemsep 0.0em
\item $\beta_\e\in [-1,0)$. In this case, we must have
$C>\frac{4{\tilde c}_\e}{1-\beta_\e}\mu(\beta_\e)$ and $\theta$ lies
in a union of four disjoint open intervals on $\rS^1$ which is
invariant under the action of $\Z_2\times \Z_2$. Hence for each value
of $C$ we have four gradient flow orbits (each in one of the four
quadrants) which are related by the action of $\Z_2\times \Z_2$. We will
see below that these orbits do not have $\e$ as a limit point. 
\item $\beta_\e\in (0,1)$. In this case, $C$ is unconstrained and
(since the values $\theta=0,\pi$ are forbidden) $\theta$ lies in a
union of four disjoint open intervals on $\rS^1$ of the form
$(0,\theta_0)\cup (\pi-\theta_0,\pi)\cup(\pi+\theta_0,2\pi-\theta_0)$.
Hence for each value of $C$ we have four gradient flow orbits (one in
each quadrant), which are related to each other by reflections in the
coordinate axes. These orbits have $\e$ as a limit point. 
\end{enumerate}
\end{enumerate}

\paragraph{Non-special gradient flow orbits having $\e$ as a limit point.}

\noindent Suppose that $\beta_\e\neq 1$ and distinguish the cases:
\begin{enumerate}
\itemsep 0.0em
\item $\e$ is a flaring end. In this case, we have
$\epsilon_\e=+1$ and the left hand side of \eqref{eSolBeta} tends to
$1$ as $\omega\rightarrow 0$. In this limit, the equation reduces to:
\ben
\label{Has1}
H(\theta,\beta_\e)=\frac{C-1}{{4{\tilde c}_\e}}(1-\beta_\e)~~.
\een
Distinguish the cases:
\begin{itemize}
\item When $\beta_\e\in [-1,0)$, condition \eqref{Has1} requires $C\geq
1+\frac{4\tc_\e}{1-\beta_\e}\mu(\beta_\e)$. When this inequality
holds strictly, the condition has eight solutions of the form
$\theta=\pm\theta_1\mod \pi$ and $\theta=\pm \theta_2\mod \pi$, where
$\theta_1$ and $\theta_2$ are the two solutions lying in the interval
$[0,\frac{\pi}{2})$ (see Figure \ref{Hbm1}(b)). The corresponding
eight gradient flow orbits asymptote to the end making one of these
angles with the $x$ axis.  The angles $\theta_1$ and $\theta_2$
coincide when $C=1+\frac{4\tc_\e}{1-\beta_\e}\mu(\beta_\e)$; in this
case, the geodesics lying in the same quadrant reach the end at the
same angle.
\item When $\beta_\e\in (0,1)$, condition \eqref{Has1} does not
constrain $C$. This condition has four solutions of the form
$\theta=\pm\theta_0\mod \pi$, where $\theta_0$ is the solution lying
in the interval $[0,\frac{\pi}{2})$ (see Figure \ref{Hbp1}(b)). The
corresponding four gradient flow orbits reach the origin making these
angles with the $x$ axis.
\end{itemize}
\item $\e$ is a cusp end. In this case, we have $\epsilon_\e=-1$ and
the left hand side of \eqref{eSolBeta} tends to plus infinity when
$\omega\rightarrow 0$. This requires that the right hand side also
tends to $+\infty$, which (for fixed $C$) happens when
$H(\theta,\beta_\e)$ tends to $-\infty$. Hence we must have
$\beta_\e>0$ and $\theta\rightarrow 0$ or $\theta\rightarrow \pi$,
which means that each of the four gradient flow orbits which asymptote
near the end to one of the two principal geodesic orbits which have
$\e$ as a limit point and correspond to the semi-axes determined by
the $x$ axis.
\end{enumerate}

\noindent A few unoriented gradient flow orbits of the effective
scalar triple $(\Sigma,G,V)$ near critical ends are plotted in Figures
\ref{fig:CritPlanePQ}-\ref{fig:CritCuspPQ} in principal canonical
coordinates centered at the end.

\vspace{-1cm}

\begin{figure}[H]
\centering
\begin{minipage}{.5\textwidth}
\centering \includegraphics[width=.99\linewidth]{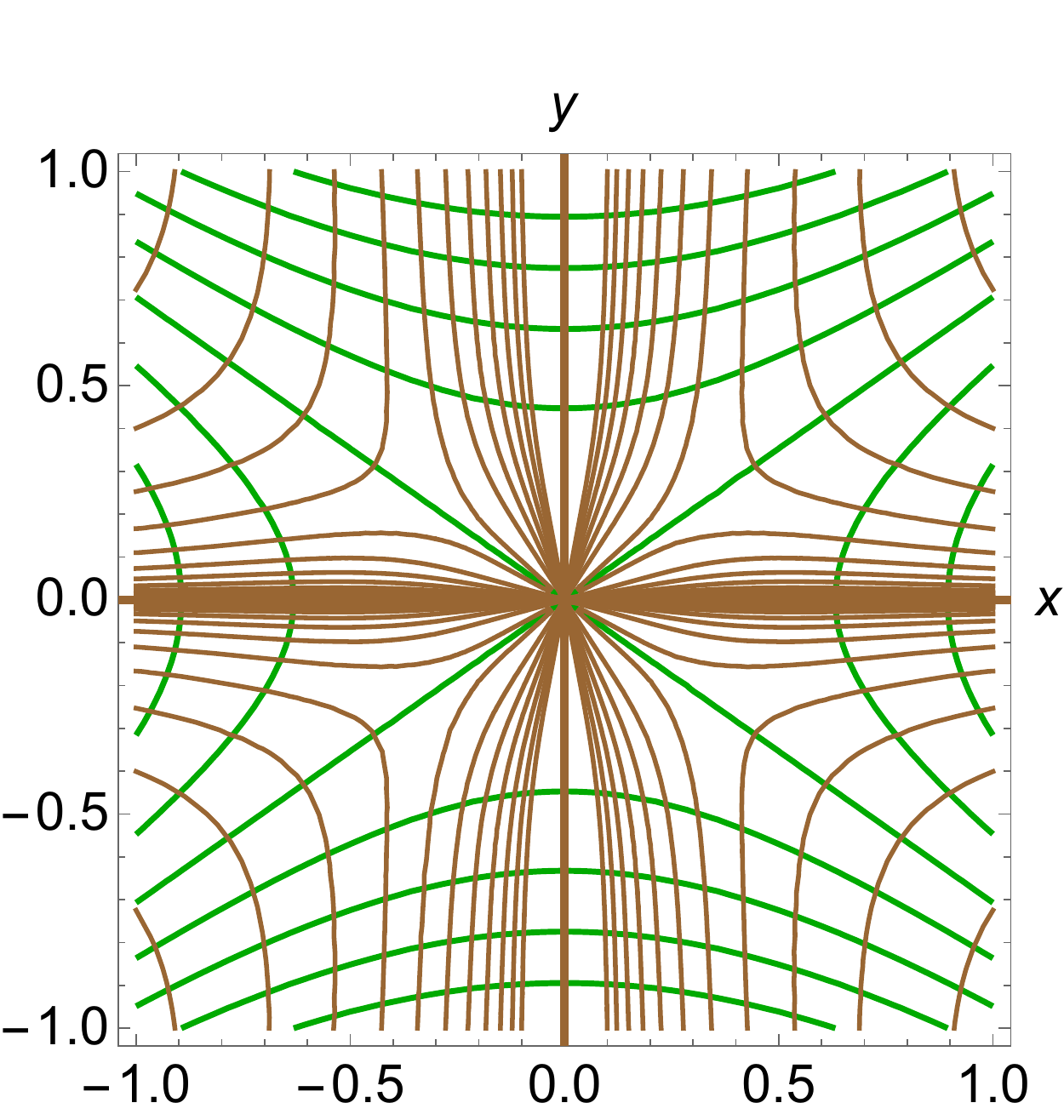}
\subcaption{For $\beta_\e=-1/2$}
\end{minipage}\hfill 
\begin{minipage}{.5\textwidth}
\vskip 0.2em
\centering \includegraphics[width=.99\linewidth]{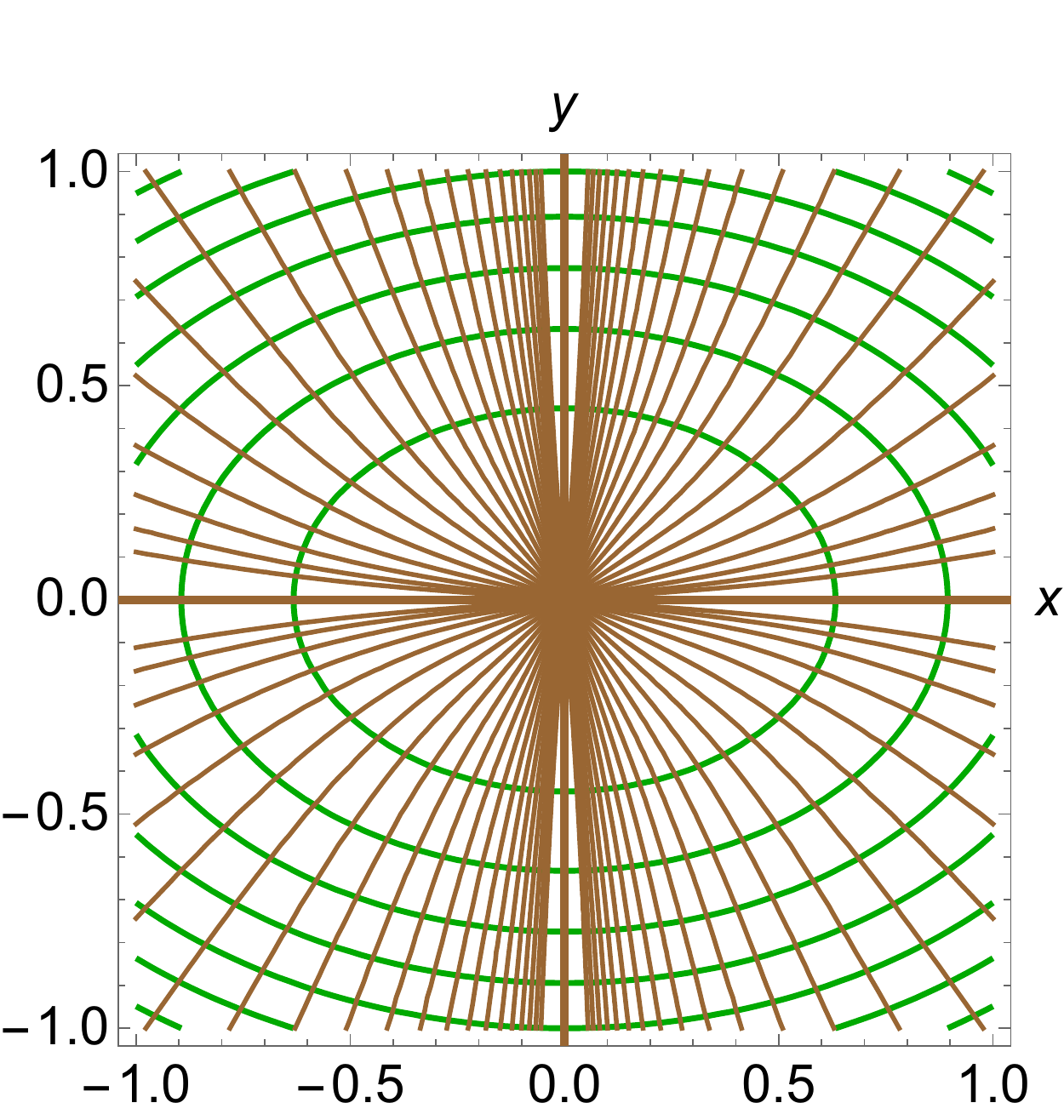}
\subcaption{For $\beta_\e=1/2$}
\end{minipage}
\caption{Gradient flow orbits of $V$ (shown in brown) and level sets
of $V$ (shown in green) near a critical plane end $\e$, drawn in
principal Cartesian canonical coordinates centered at $\e$ for two
values of $\beta_\e$.}
\label{fig:CritPlanePQ}
\end{figure}

\vspace{-0.5cm}

\begin{figure}[H]
\centering
\begin{minipage}{.5\textwidth}
\centering  \includegraphics[width=.99\linewidth]{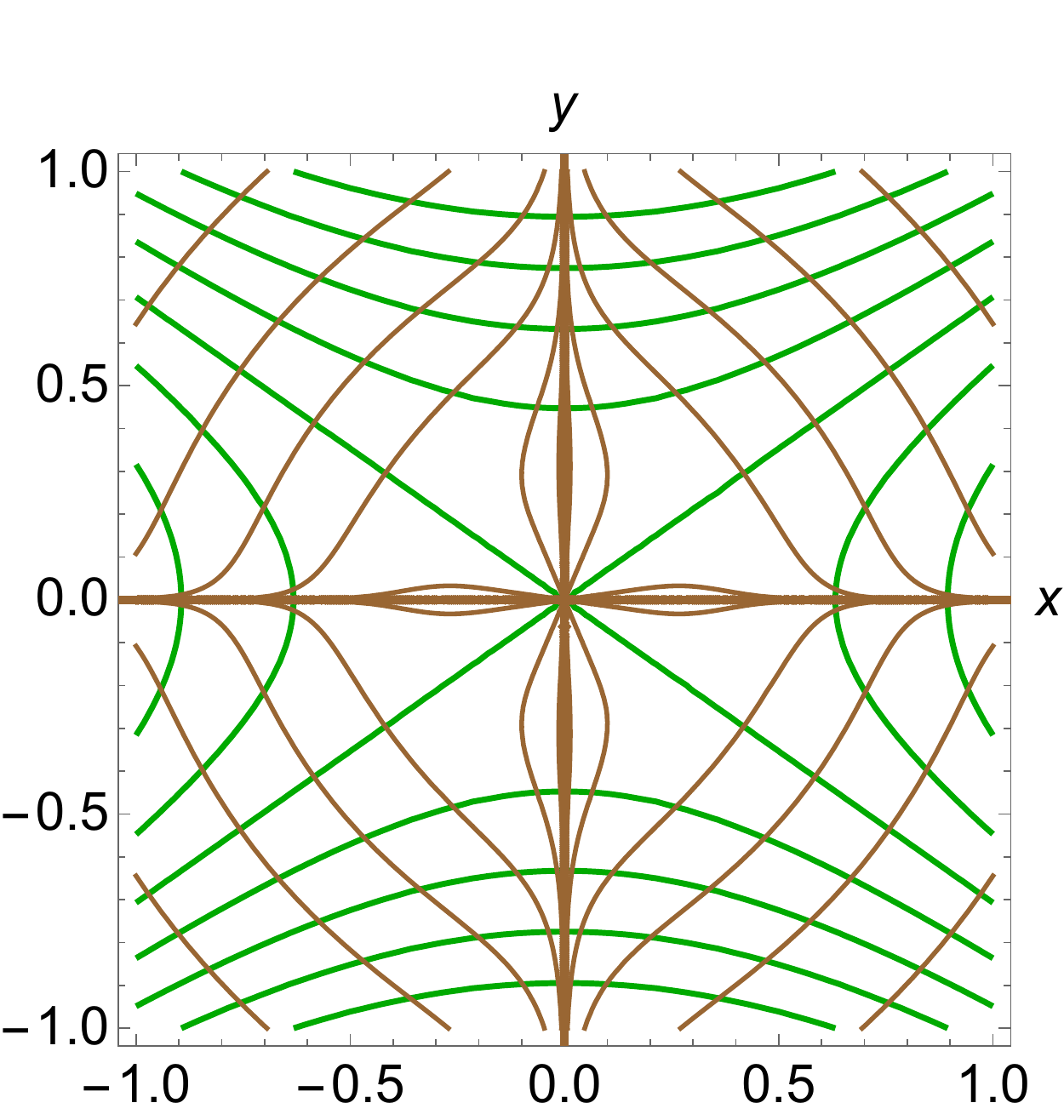}
\subcaption{For $\beta_\e=-1/2$.}
\end{minipage}\hfill
\begin{minipage}{.5\textwidth}
\centering \includegraphics[width=.99\linewidth]{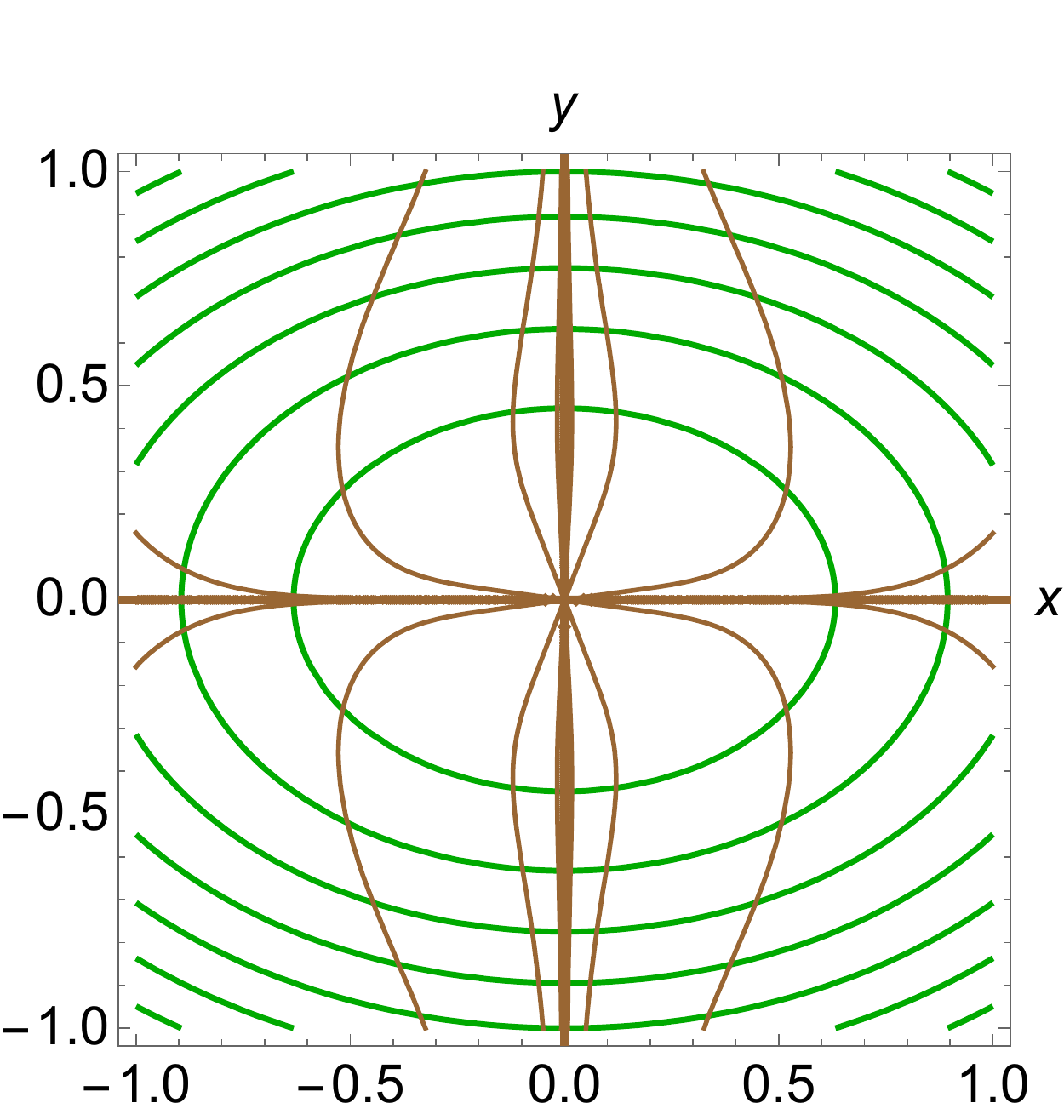}
\subcaption{For $\beta_\e=1/2$.}
\end{minipage}
\caption{Gradient flow orbits of $V$ (shown in brown) and level sets
of $V$ (shown in green) near a critical horn end $\e$, drawn in
principal Cartesian canonical coordinates centered at $\e$ for two
values of $\beta_\e$.}
\label{fig:CritHornPQ}
\end{figure}

\vspace{-0.5cm}

\begin{figure}[H]
\centering
\begin{minipage}{.5\textwidth}
\centering  \includegraphics[width=.99\linewidth]{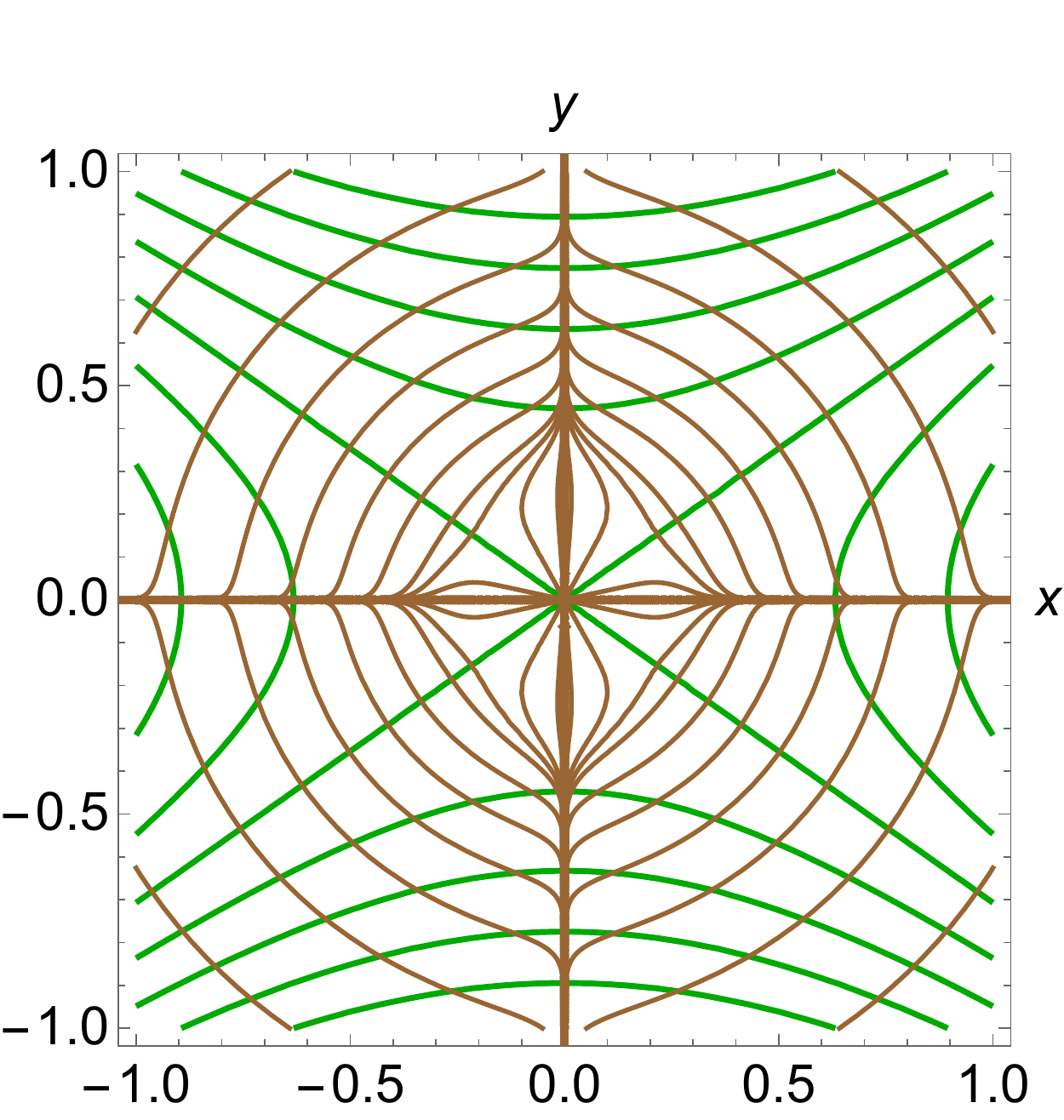}
\subcaption{For $\beta_\e=-1/2$.}
\end{minipage}\hfill
\begin{minipage}{.5\textwidth}
\centering \includegraphics[width=.99\linewidth]{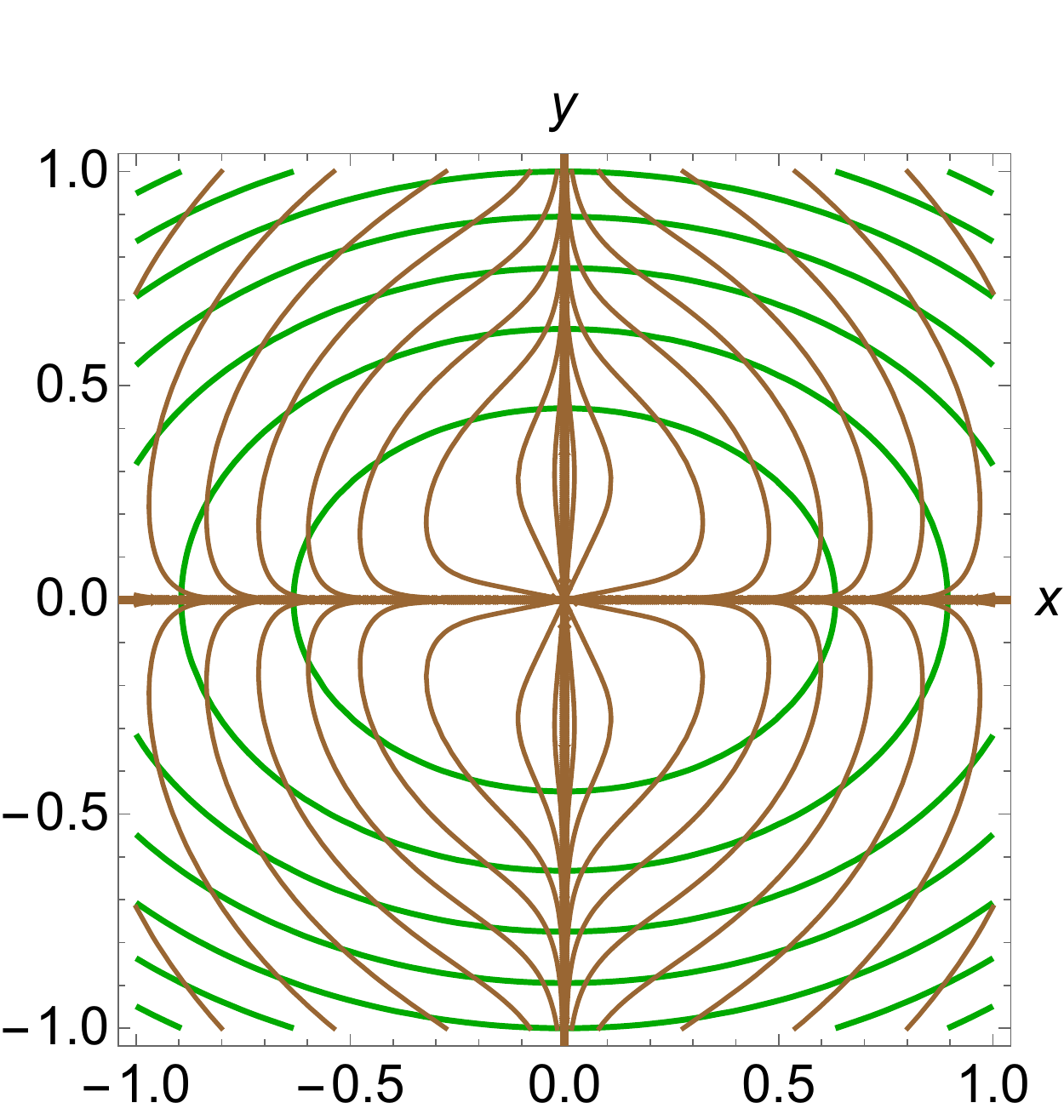}
\subcaption{For $\beta_\e=1/2$.}
\end{minipage}
\caption{Gradient flow orbits of $V$ (shown in brown) and level sets
of $V$ (shown in green) near a critical funnel end $\e$ of
circumference $\ell=1$, drawn in principal Cartesian canonical
coordinates centered at $\e$ for two values of $\beta_\e$.}
\label{fig:CritFunnelPQ}
\end{figure}

\begin{figure}[H]
\centering
\begin{minipage}{.5\textwidth}
\centering  \includegraphics[width=.99\linewidth]{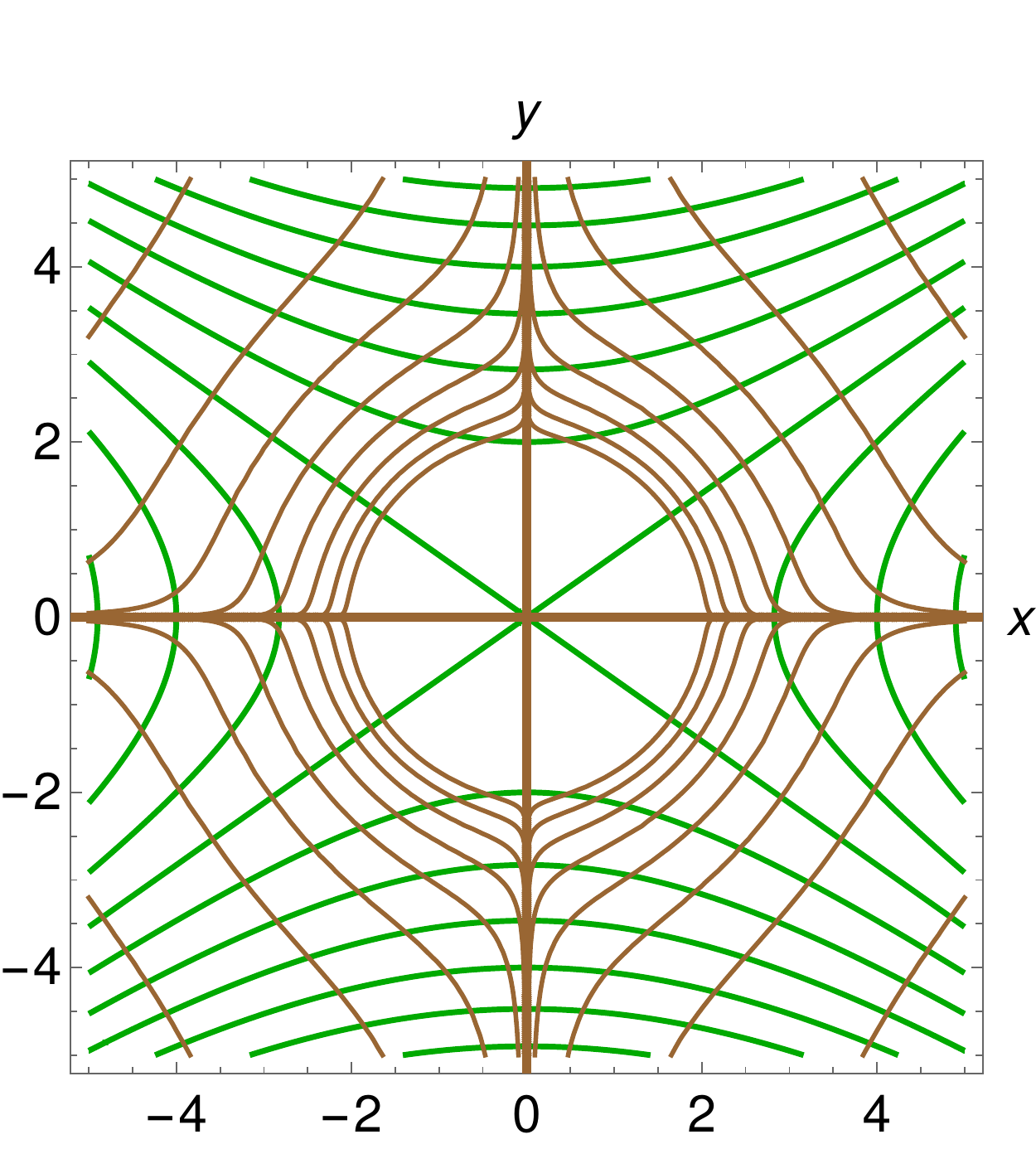}
\subcaption{For $\beta_\e=-1/2$.}
\end{minipage}\hfill
\begin{minipage}{.5\textwidth}
\centering \includegraphics[width=.99\linewidth]{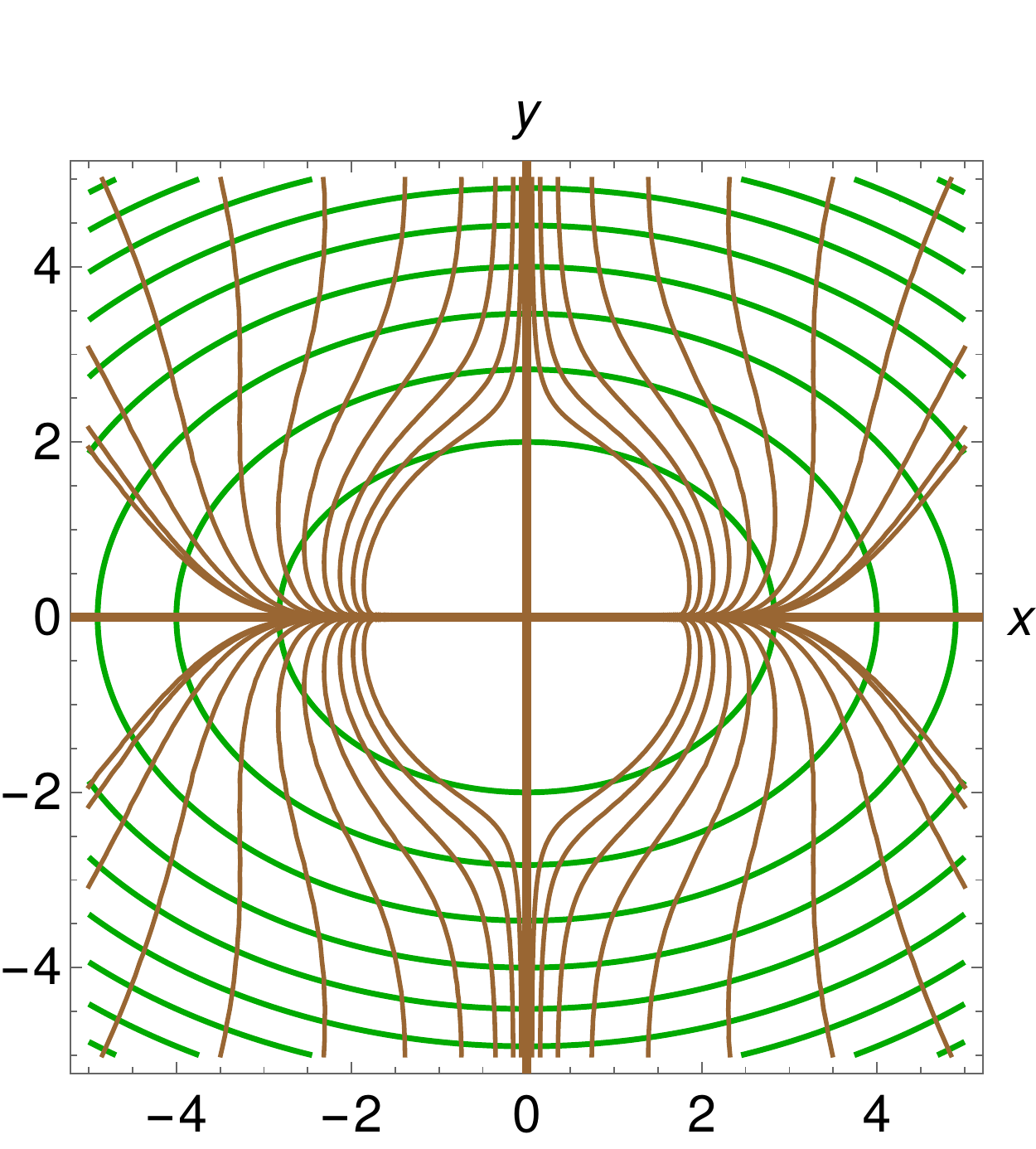}
\subcaption{For $\beta_\e=1/2$.}
\end{minipage}
\caption{Gradient flow orbits of $V$ (shown in brown) and level sets
of $V$ (shown in green) near a critical cusp end $\e$, drawn in
principal Cartesian canonical coordinates centered at $\e$ for two
values of $\beta_\e$.}
\label{fig:CritCuspPQ}
\end{figure}

The extended scalar potential $\hPhi$ of the canonical model can be
recovered from the extended classical effective potential as:
\ben
\label{Phic}
\hPhi=\frac{1}{2 M_0^2} \hV^2\approx \frac{{\bar \lambda}_2(\e)^2}{2}\left[\hat{\bar{V}}(\e)+\frac{1}{2}\omega^2 (\beta_\e \cos^2\theta+\sin^2\theta)\right]^2~~,
\een
where we defined:
\be
{\bar \lambda}_2(\e)\eqdef \frac{\lambda_2(\e)}{M_0}~~,~~{\hat {\bar V}}(\e)\eqdef \frac{{\hat V}(\e)}{\lambda_2(\e)}~~.
\ee
Figures \ref{fig:CritCosmPlane}-\ref{fig:CritCosmCusp} show some
numerically computed infrared optimal cosmological orbits of the
canonical model parameterized by $(M_0,\Sigma,G,\Phi)$ near critical ends
$\e$. In these figures, we took ${\bar \lambda}_2(\e)=1$, ${\bar {\hat
V}}(\e)=1$ and $M_0=1$. Notice that the accuracy of the first order
IR approximation depends on the value of ${\hat {\bar V}}(\e)$, since
the first IR parameter of \cite{ren} depends on this value. The initial
point of each orbit is shown as a black dot. 

\begin{figure}[H]
\centering
\begin{minipage}{.5\textwidth}
\vskip 0.5em
\centering  \includegraphics[width=.99\linewidth]{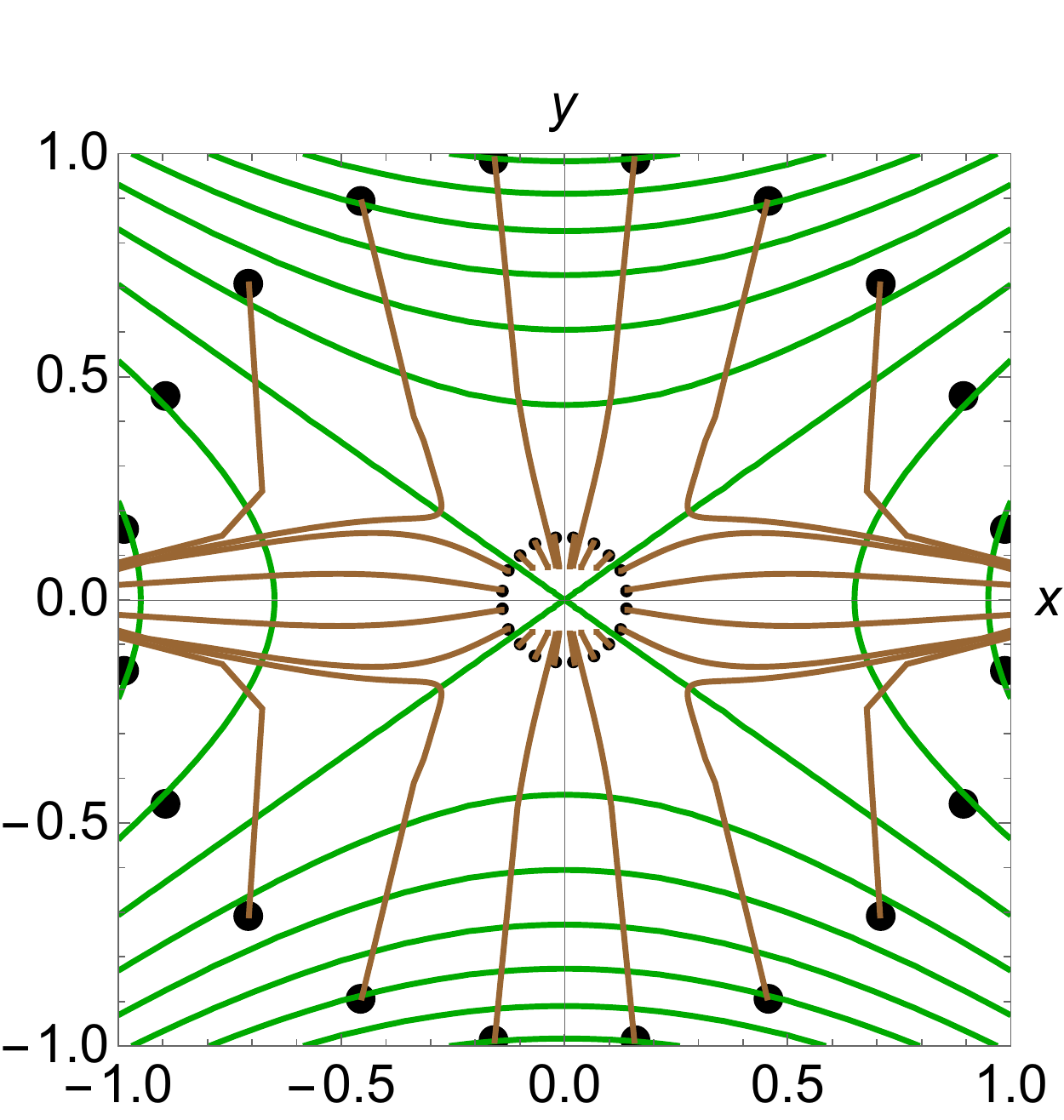}
\subcaption{For $\beta_\e=-1/2$}
\end{minipage}\hfill 
\begin{minipage}{.5\textwidth}
\vskip 0.2em
\centering \includegraphics[width=.99\linewidth]{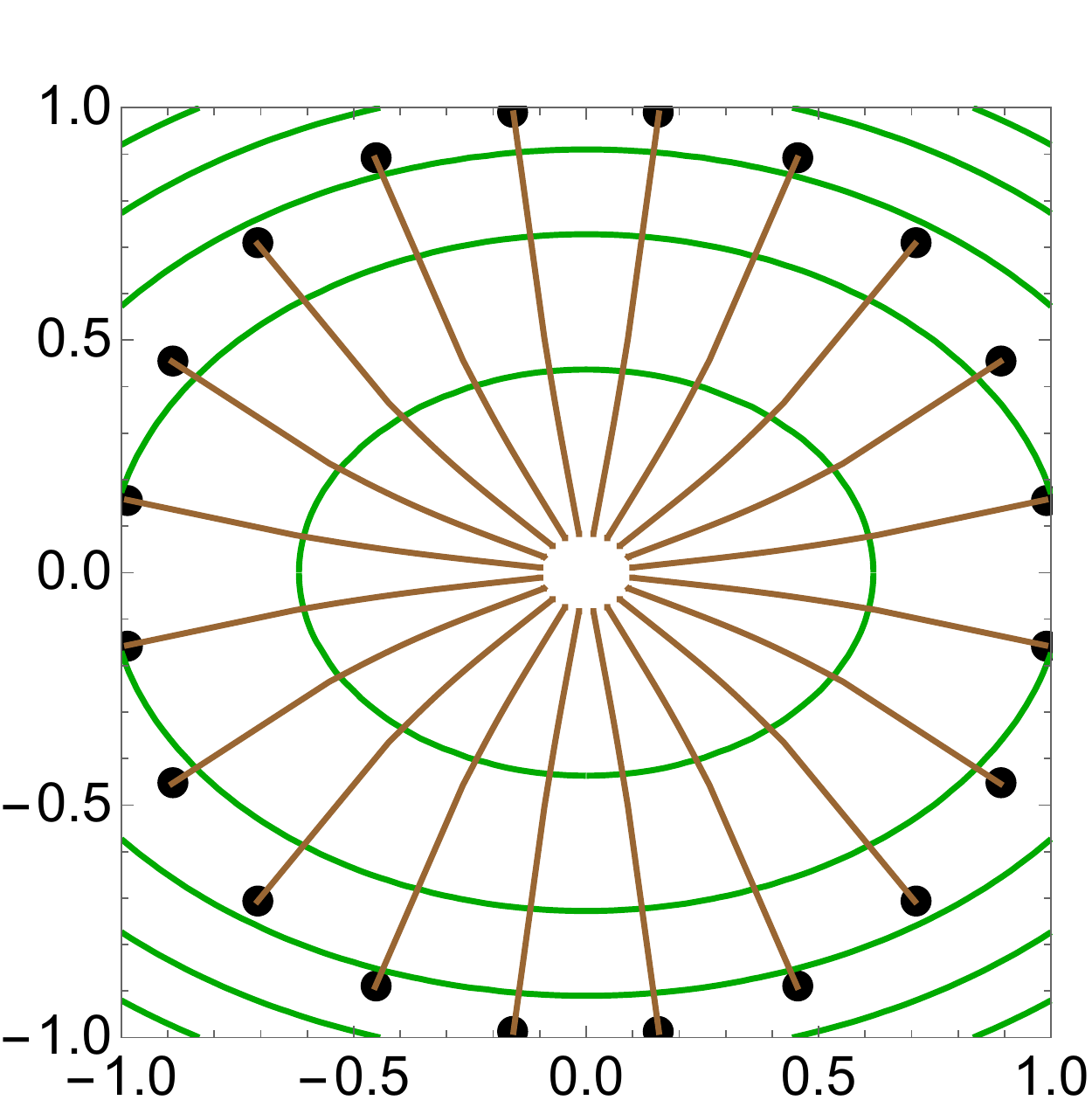}
\subcaption{For $\beta_\e=1/2$}
\end{minipage}
\caption{Numerically computed infrared optimal cosmological orbits of
the canonical model (shown in brown) and level sets of $\hPhi$ (shown
in green) near a critical plane end $\e$, drawn in principal canonical
Cartesian coordinates centered at $\e$ for two values of $\beta_\e$.}
\label{fig:CritCosmPlane}
\end{figure}

\vspace{-0.5cm}

\begin{figure}[H]
\centering
\begin{minipage}{.5\textwidth}
\centering  \includegraphics[width=.99\linewidth]{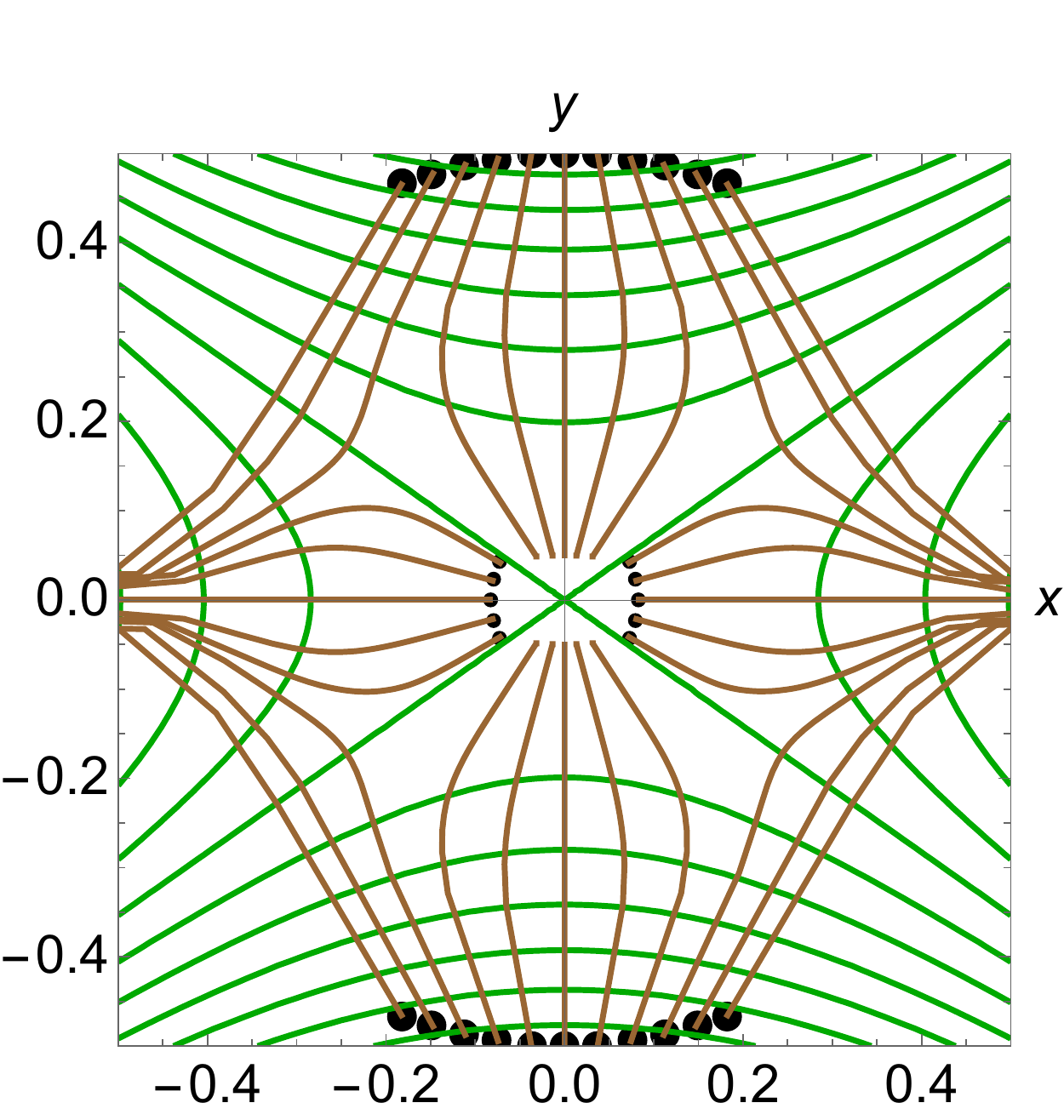}
\subcaption{For $\beta_\e=-1/2$.}
\end{minipage}\hfill
\begin{minipage}{.5\textwidth}
\centering \includegraphics[width=.99\linewidth]{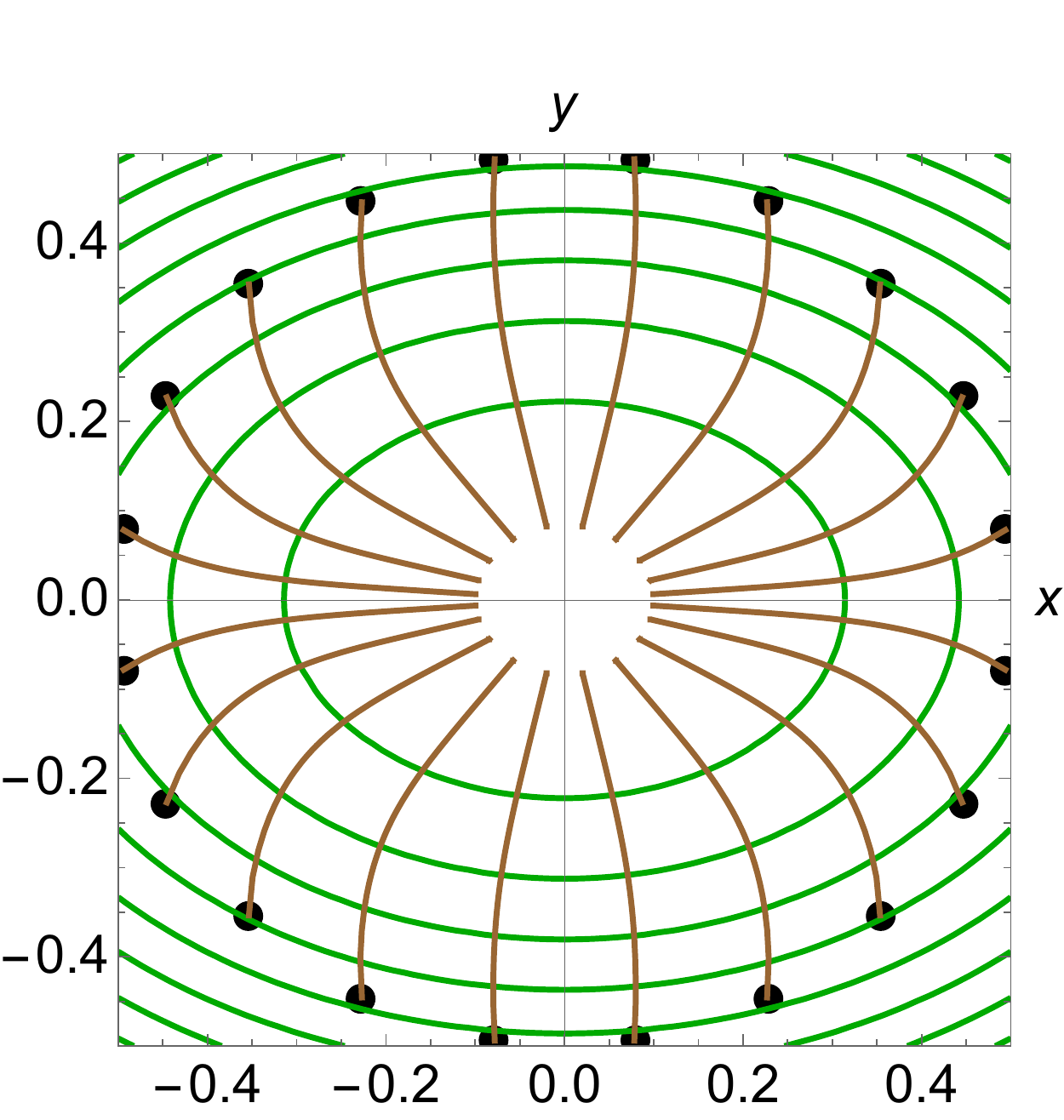}
\subcaption{For $\beta_\e=1/2$.}
\end{minipage}
\caption{Numerically computed infrared optimal cosmological orbits of
the canonical model (shown in brown) and level sets of $\hPhi$ (shown
in green) near a critical horn end $\e$, drawn in principal canonical
Cartesian coordinates centered at $\e$ for two values of $\beta_\e$.}
\label{fig:CritCosmHorn}
\end{figure}

\vspace{-0.5cm}

\begin{figure}[H]
\centering
\begin{minipage}{.5\textwidth}
\centering  \includegraphics[width=.99\linewidth]{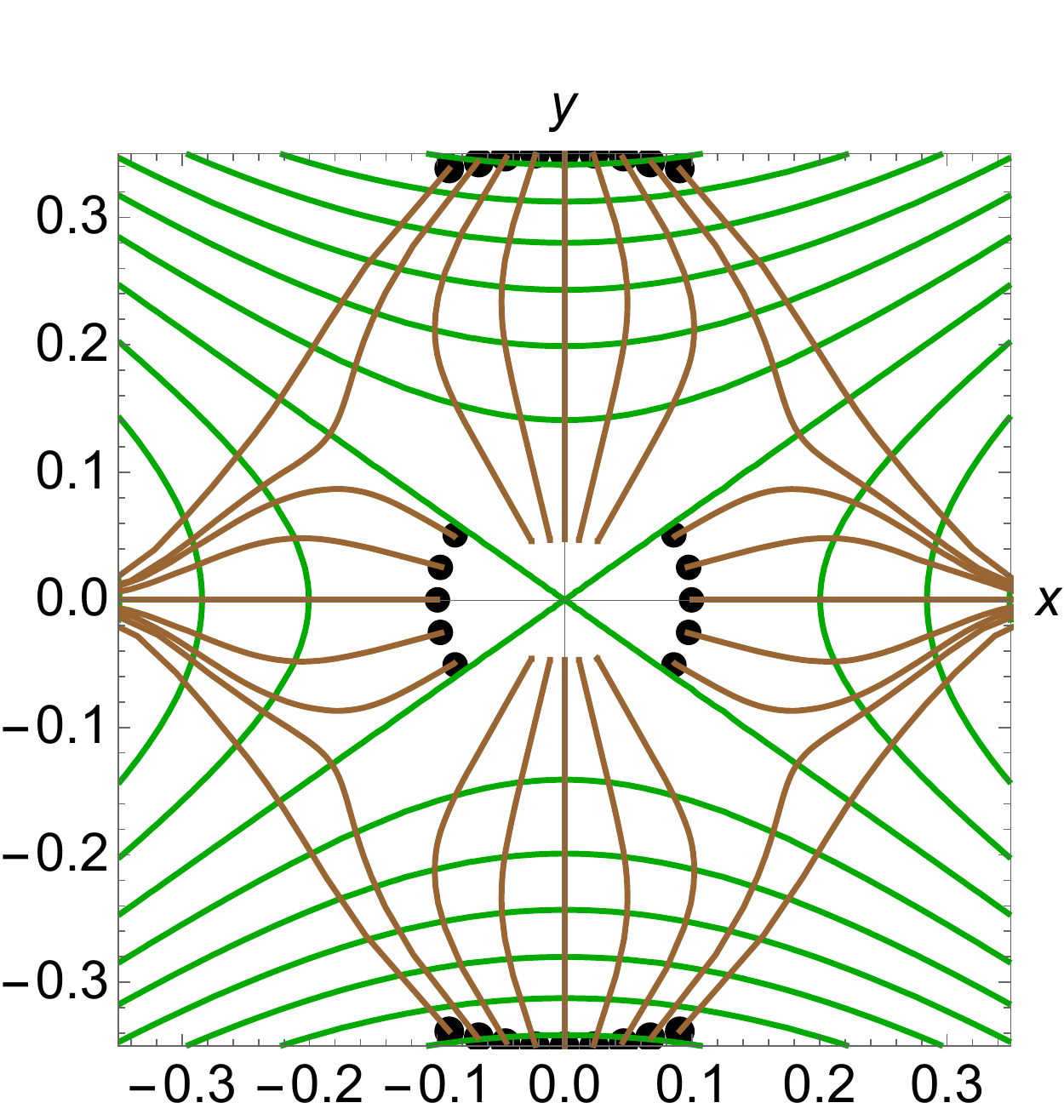}
\subcaption{For $\beta_\e=-0.5$.}
\end{minipage}\hfill
\begin{minipage}{.5\textwidth}
\centering \includegraphics[width=.99\linewidth]{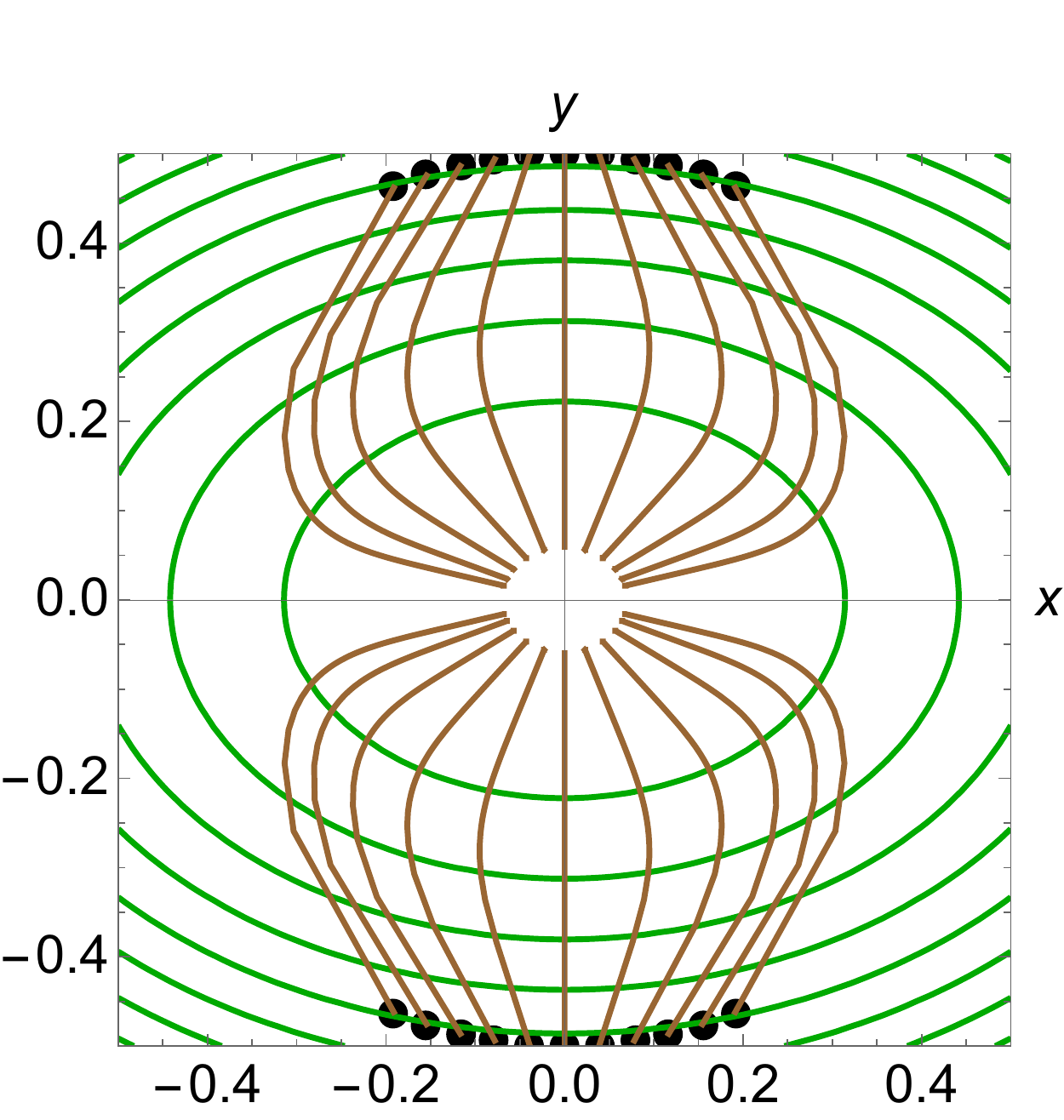}
\subcaption{For $\beta_\e=0.5$.}
\end{minipage}
\caption{Numerically computed infrared optimal cosmological orbits of the
canonical model (shown in brown) and level sets
of $\hPhi$ (shown in green) near a critical funnel end $\e$ of
circumference $\ell=1$, drawn in principal canonical coordinates
centered at $\e$ for two values of $\beta_\e$.}
\label{fig:CritCosmFunnel}
\end{figure}

\begin{figure}[H]
\centering  
\begin{minipage}{.5\textwidth}
\centering  \includegraphics[width=.99\linewidth]{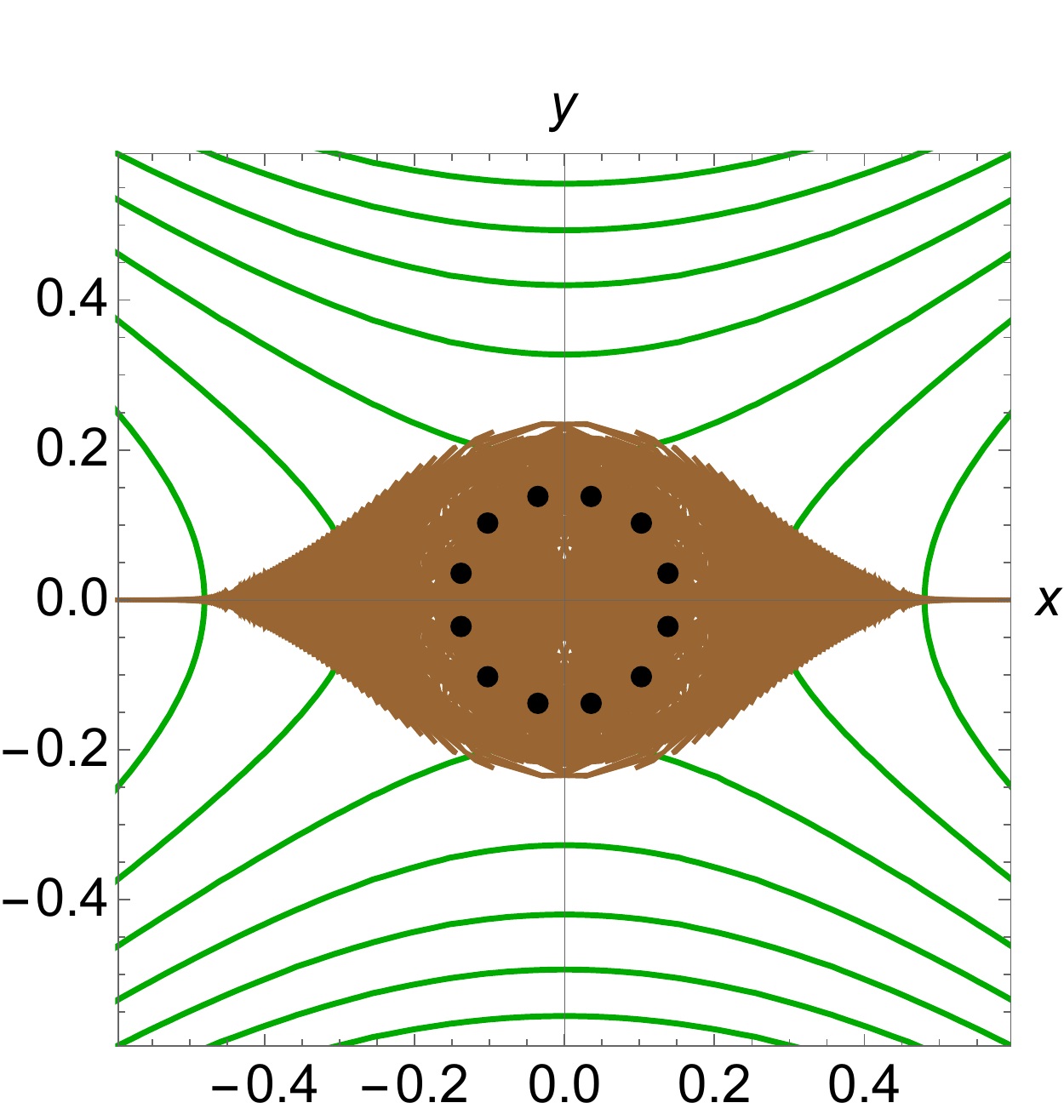}
\subcaption{For $\beta_\e=-0.5$.}
\end{minipage}\hfill
\begin{minipage}{.5\textwidth}
\centering \includegraphics[width=.99\linewidth]{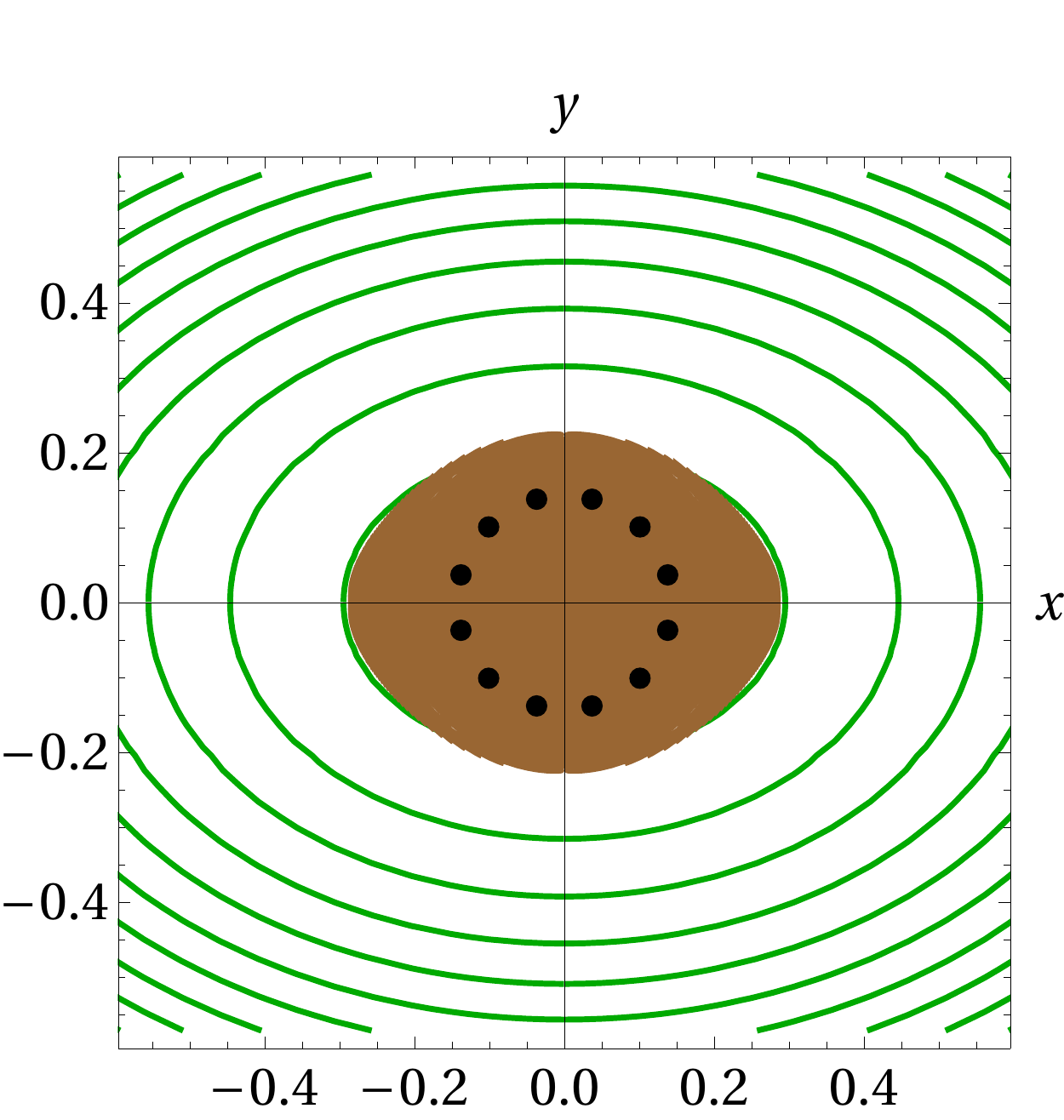}
\subcaption{For $\beta_\e=0.5$.}
\end{minipage}
\caption{Numerically computed infrared optimal cosmological orbits of the
canonical model (shown in brown) and level sets
of $\hPhi$ (shown in green) near a critical cusp end $\e$, drawn in
principal canonical coordinates centered at $\e$ for two values of
$\beta_\e$.}
\label{fig:CritCosmCusp}
\end{figure}

The case of critical cusp ends is particularly interesting. For
clarity, Figures \ref{fig:CritCosmCuspMinusSingle} and
\ref{fig:CritCosmCuspPlusSingle} display the evolution of a single
infrared optimal cosmological curve of the uniformized model for four
consecutive cosmological times when $\beta_\e<0$ and $\beta_\e>0$
respectively, where in the second case we assume that $\e$ is a local
minimum of $\hV$; this orbit was chosen such that its initial point
does not lie on any of the four principal geodesic orbits.  When
$\beta_\e<0$, the orbit spirals numerous times around the cusp while
approaching it, after which it spirals away from the cusp until it is
finally repelled by it along one of two principal geodesic
orbits. When $\beta_\e>0$ (and assuming as in Figure
\ref{fig:CritCosmCuspPlusSingle} that $\e$ is a local minimum of
$\hV$), the cosmological orbit first spirals around the cusp end
approaching it, after which it distances from it oscillating with
gradually decreasing amplitude around one of the principal geodesic
orbits until stopped by the attractive force generated by the scalar
potential. At that point the cosmological curve starts evolving back
towards the cusp and ``falls back into'' the cusp end along the
principal geodesic. On the other hand, infrared optimal cosmological
curves which start from a point lying on one of the four principal
geodesic orbits flow away from the cusp or into it along that geodesic
depending on whether the cusp is a source, sink or saddle for $\hV$
(and, in the saddle case, on the choice of that principal geodesic).

\begin{figure}[H]
\centering
\begin{minipage}{.45\textwidth}
\centering ~~\includegraphics[width=.97\linewidth]{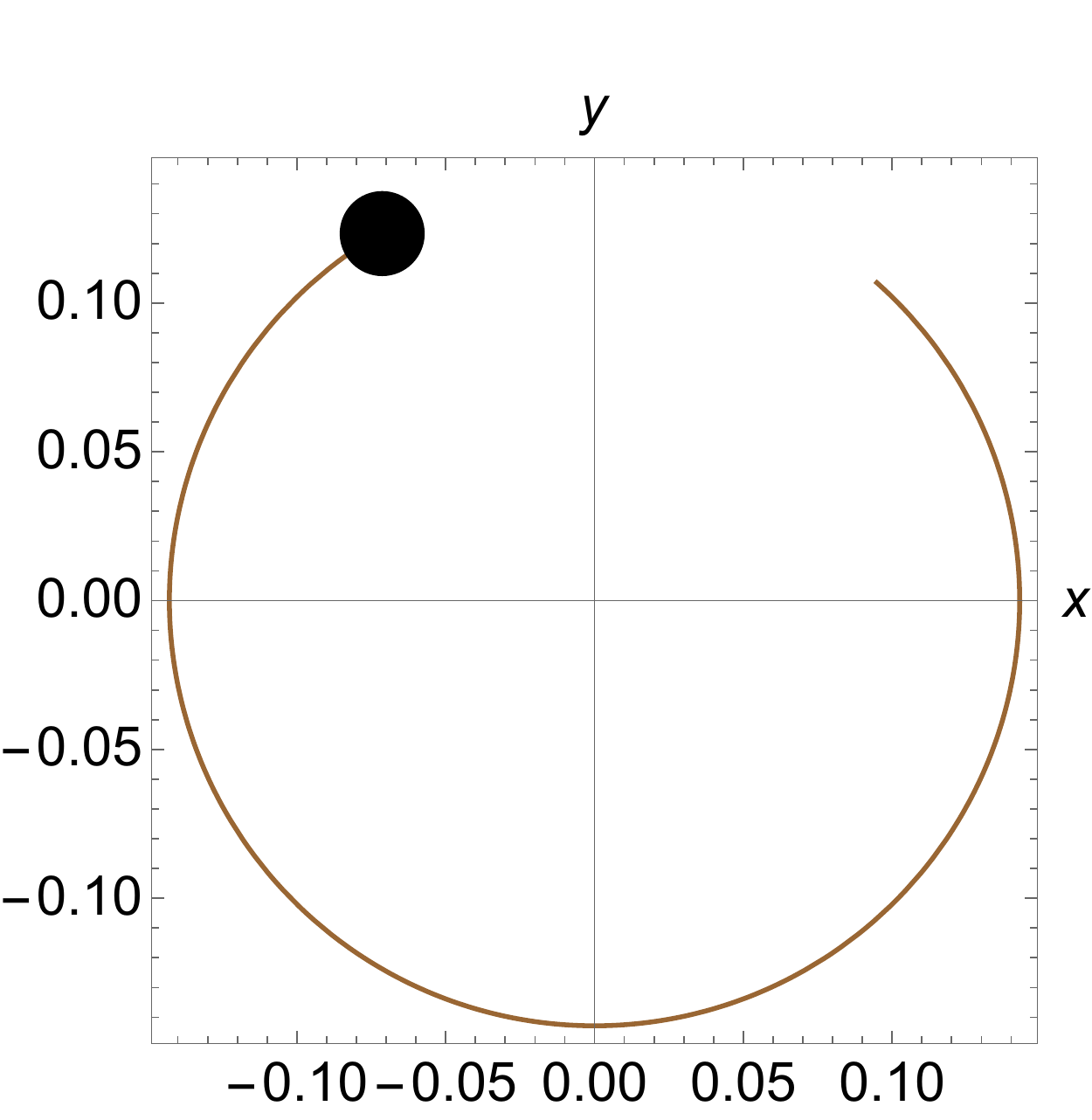}
\subcaption{For $t= 0.000008$.}
\end{minipage}\hfill 
\begin{minipage}{.45\textwidth}
\centering ~~\includegraphics[width=.95\linewidth]{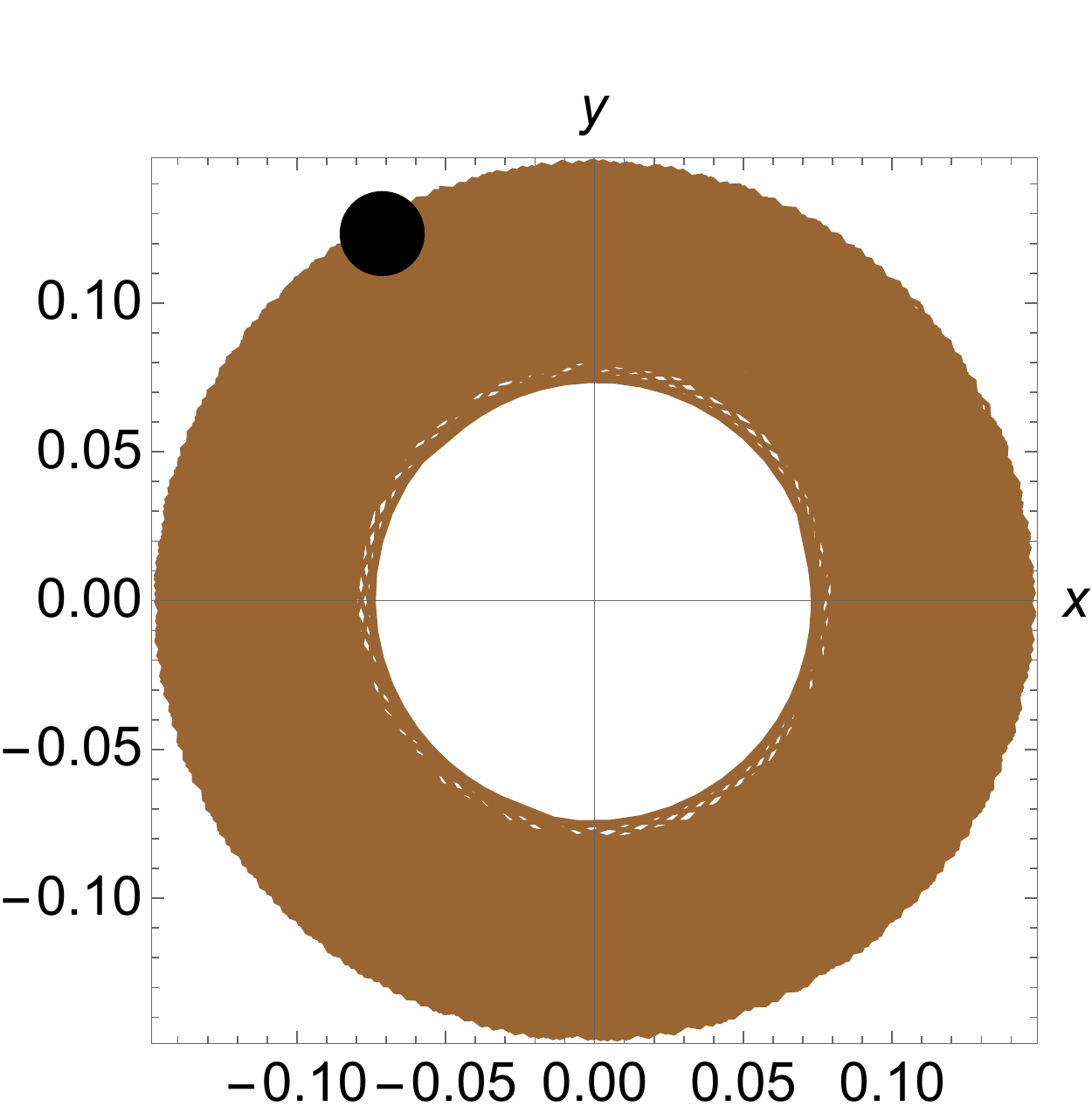}
\subcaption{For $t=0.01$.}
\end{minipage}\\
\begin{minipage}{.45\textwidth}
\centering ~~\includegraphics[width=.95\linewidth]{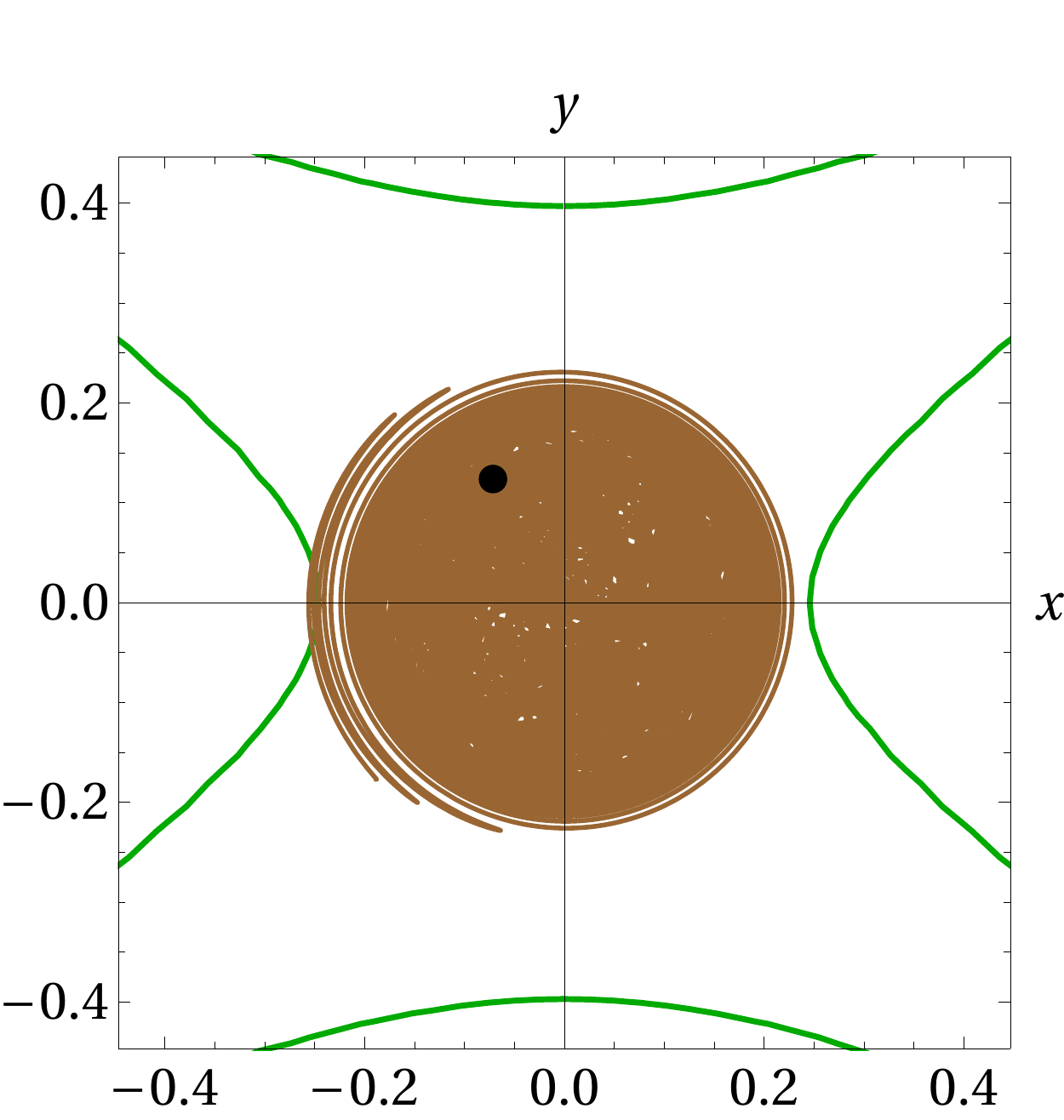}
\subcaption{For $t=0.5$.}
\end{minipage}\hfill
\begin{minipage}{.45\textwidth}
\centering ~~ \includegraphics[width=.97\linewidth]{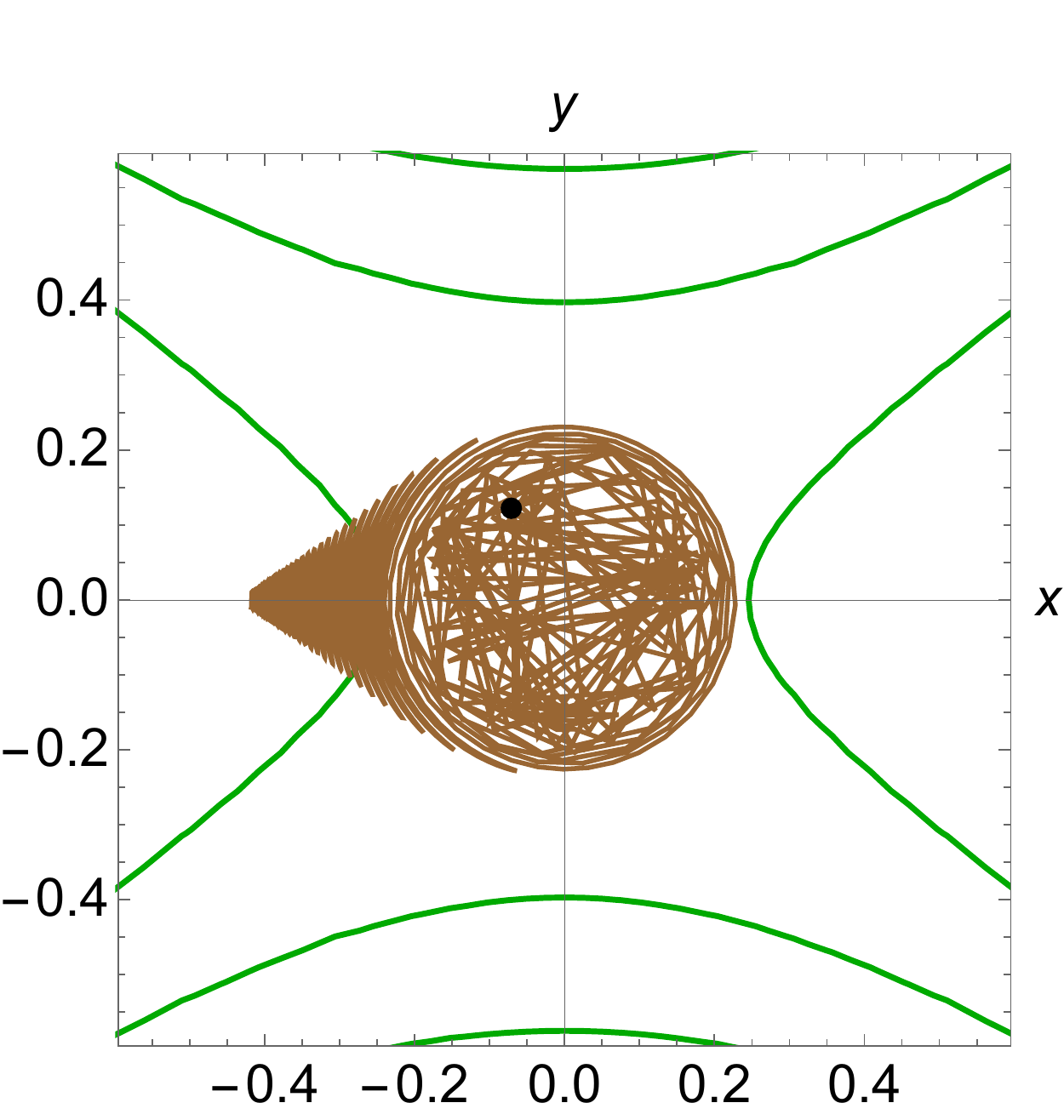}
\subcaption{For $t=5$.}
\end{minipage}
\caption{A numerically computed infrared optimal cosmological orbit of
the canonical model (shown in brown) and level sets of $\hPhi$ (shown
in green) near a critical cusp end $\e$ for $\beta_\e=-1/2$. The
trajectory starts at $\omega(0)=\frac{1}{7}$ and
$\theta(0)=\frac{2\pi}{3}$ with initial speed in the gradient flow
shell of $(\Sigma,G,V)$. The four figures show the cosmological orbit
for cosmological times between $t=0$ and $t= 0.000008$, $0.01$, $0.5$
and $5$ respectively. The white interior areas in the last figure
are numerical artifacts due to limited float precision.}
\label{fig:CritCosmCuspMinusSingle}
\end{figure}

\begin{figure}[H]
\centering
\begin{minipage}{.45\textwidth}
\centering ~~\includegraphics[width=.8\linewidth]{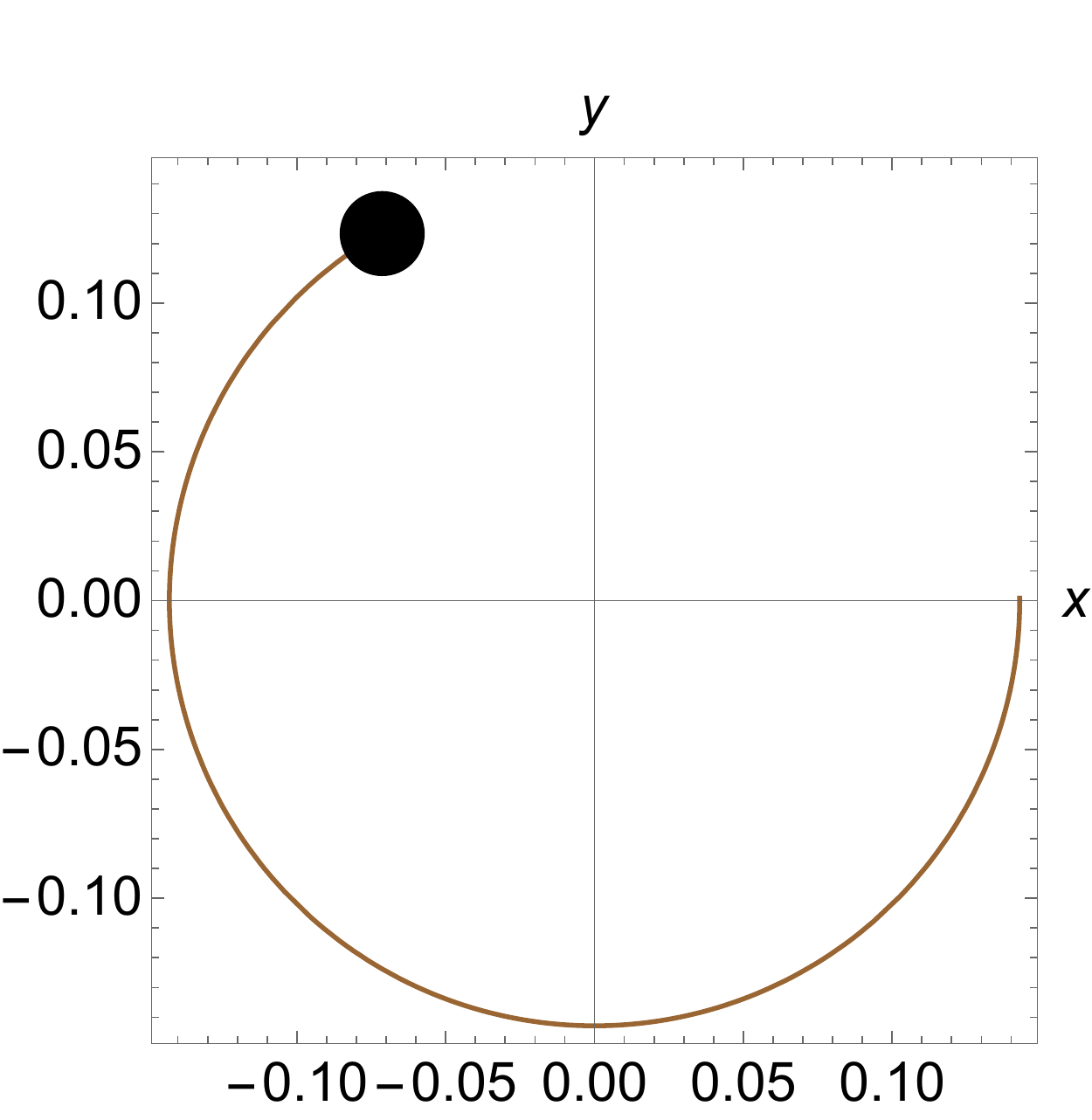}
\subcaption{For $t= 0.00002$.}
\end{minipage}\hfill 
\begin{minipage}{.45\textwidth}
\centering ~~\includegraphics[width=.8\linewidth]{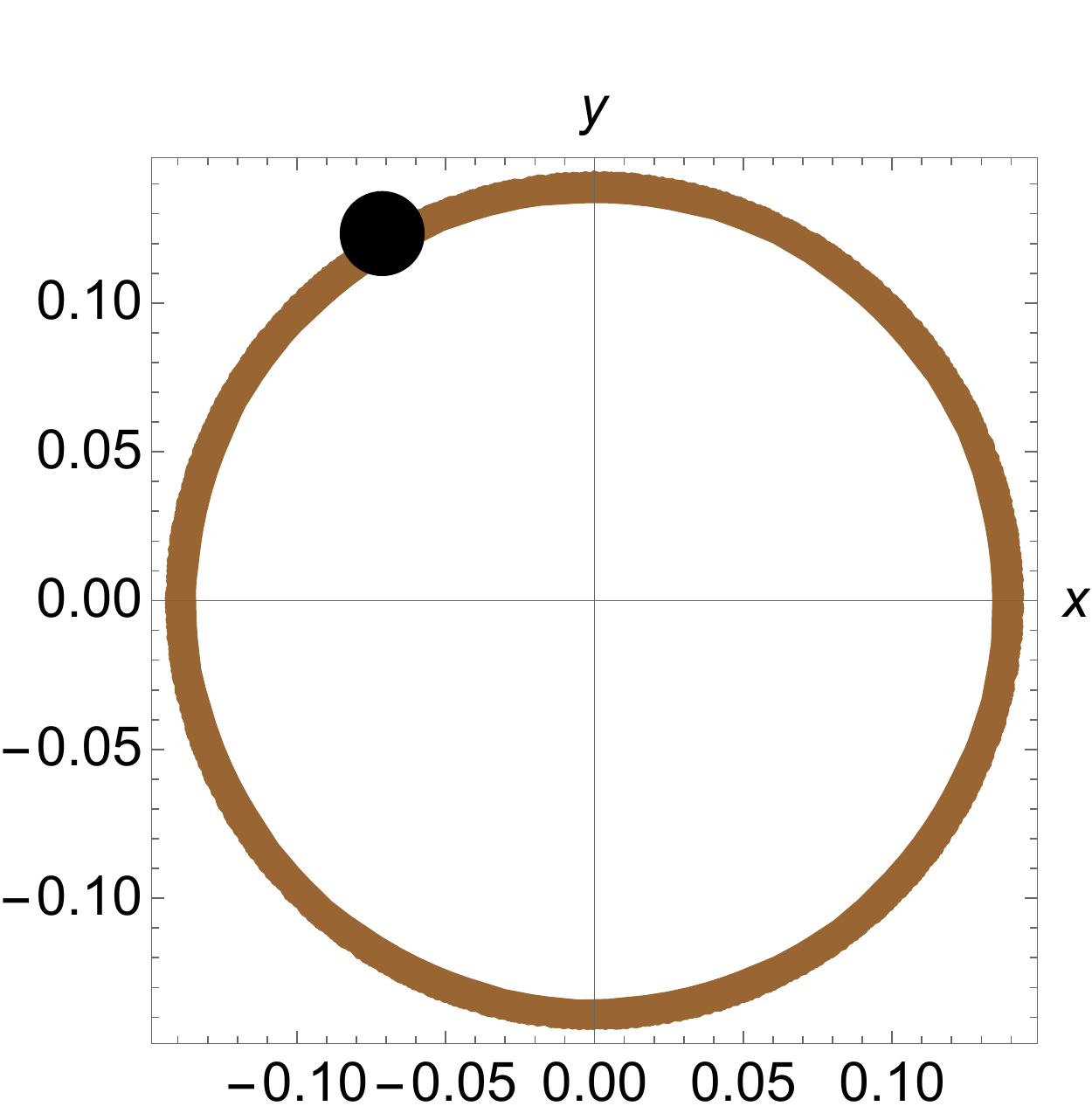}
\subcaption{For $t=0.01$.}
\end{minipage}\\
\begin{minipage}{.45\textwidth}
\centering ~~\includegraphics[width=.81\linewidth]{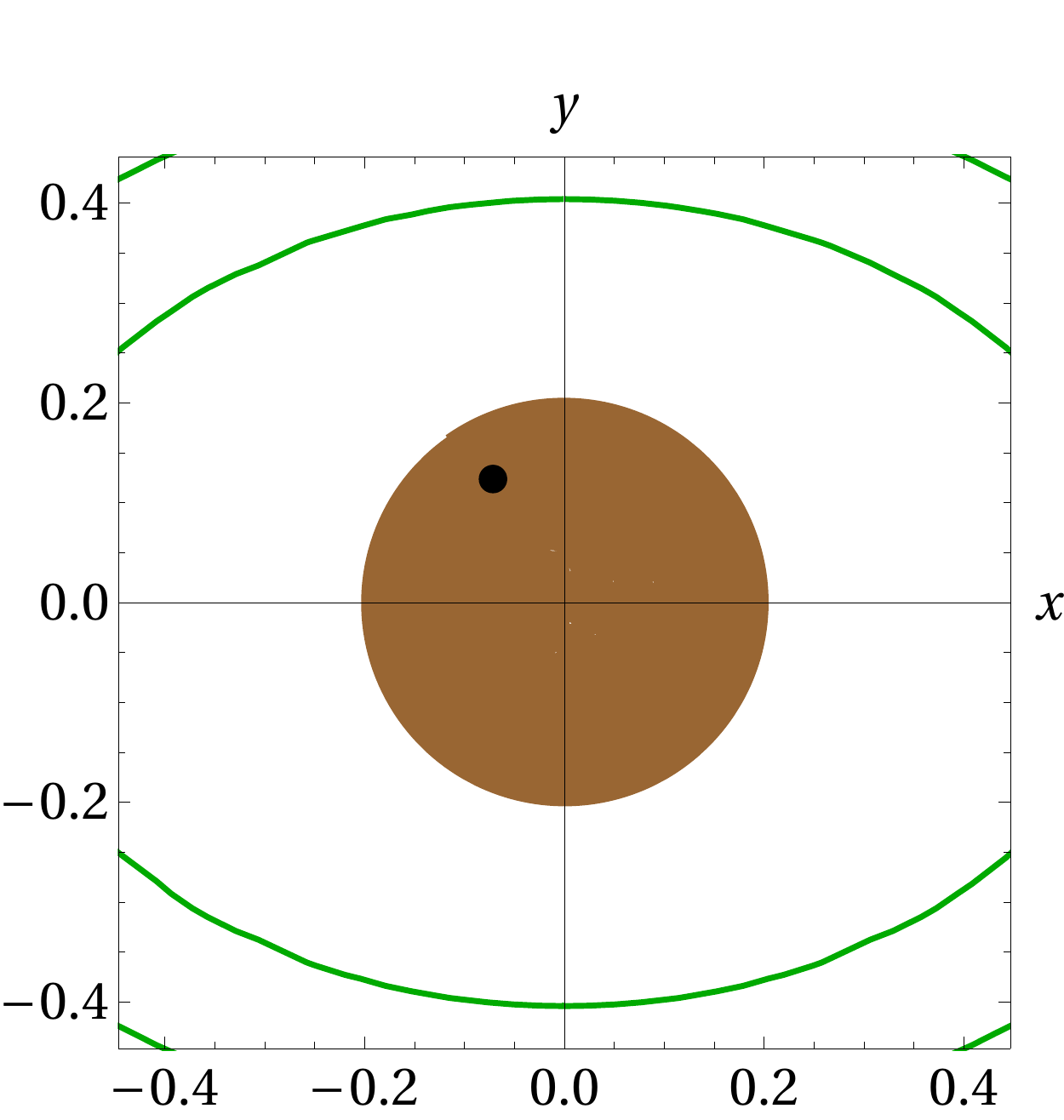}
\subcaption{For $t=0.5$.}
\end{minipage}\hfill
\begin{minipage}{.45\textwidth}
\centering ~~ \includegraphics[width=.81\linewidth]{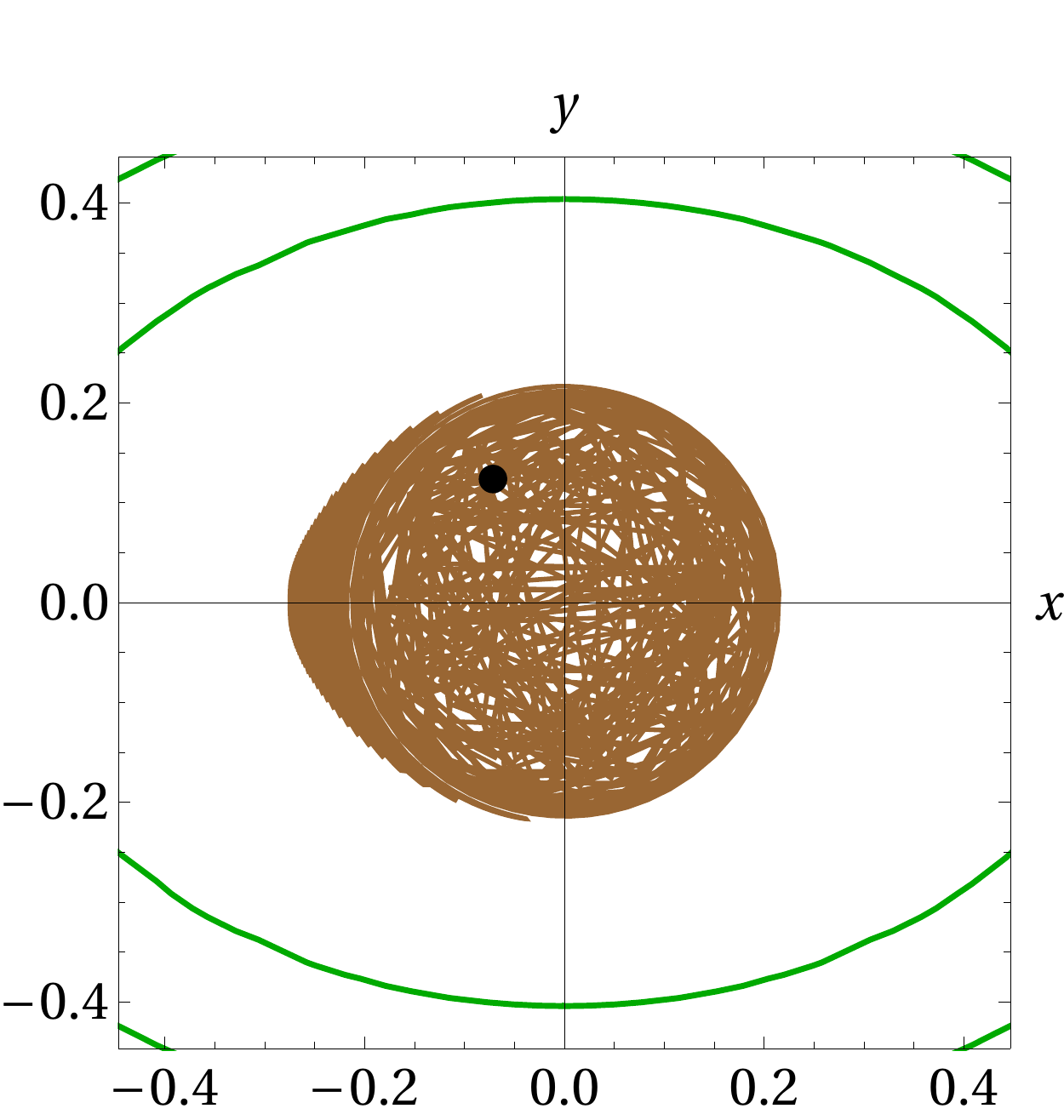}
\subcaption{For $t=5$.}
\end{minipage}
\caption{A numerically computed infrared optimal cosmological orbit of
the canonical model (shown in brown) and level sets of $\hPhi$ (shown
in green) near a critical cusp end $\e$ for $\beta_\e=+1/2$ when the
cusp end is a local minimum of $\hV$. The trajectory starts at
$\omega(0)=\frac{1}{7}$ and $\theta(0)=\frac{2\pi}{3}$ with initial
speed in the gradient flow shell of $(\Sigma,G,V)$. The four figures
show the orbit for cosmological times between $t=0$ and $t= 0.00002$,
$0.01$, $0.5$ and $5$ respectively. The white interior areas in the last figure
are numerical artifacts due to limited float precision.}
\label{fig:CritCosmCuspPlusSingle}
\end{figure}

\subsection{Stable and unstable manifolds of critical ends under the effective gradient flow}

\noindent The analysis above shows that critical ends of $\Sigma$
behave like exotic fixed points of the gradient flow of
$(\Sigma,G,V)$. The dimensions and number of connected components of
the stable and unstable manifolds (in $\Sigma$) are listed below,
where we use the notations:
\beqa
\cA_+(\e)\eqdef \cS(\e)~~&,&~~\cA_-(\e)\eqdef \cU(\e)~~\\
d_+(\e)\eqdef \dim\cS(\e)~~&,&~~d_-(\e)\eqdef \dim\cU(\e)~~\\
n_+(\e)\eqdef \Card [\pi_0(\cS(\e))]~~&,&~~n_-(\e)\eqdef \Card[\pi_0(\cU(\e))]~~.
\eeqa

\begin{enumerate}
\item If $\e$ is a flaring end:
\begin{itemize}
\item $\beta_\e\in [-1,0)$ (i.e. $\e$ is a saddle point of $\hV$):
$d_+(\e)=d_-(\e)=2$, $n_+(\e)=n_-(\e)=2$
\item $\beta_\e\in (0,1]$ (i.e. $\e$ is an extremum of $\hV$): Then
$\epsilon_1(\e)=\epsilon_2(\e):=\epsilon\in \{-1,1\}$ and
$d_{\epsilon}(\e)=2$, $n_{\epsilon}(\e)=1$ with
$\cA_{-\epsilon}=\emptyset$.
\end{itemize}
\item If $\e$ is a cusp end:
\begin{itemize}
\item $\beta_\e\in [-1,0)$ (i.e. $\e$ is a saddle point of $\hV$):
$d_+(\e)=d_-(\e)=1$, $n_+(\e)=n_-(\e)=2$
\item $\beta_\e\in (0,1]$ (i.e. $\e$ is an extremum of $\hV$): Then
$\epsilon_1(\e)=\epsilon_2(\e)=\epsilon\in \{-1,1\}$ and
$d_{\epsilon}(\e)=2$, $n_{\epsilon}(\e)=1$ with
$\cA_{-\epsilon}=\emptyset$.
\end{itemize}
\end{enumerate}

\noindent Notice that the stable and unstable manifolds of an end are
subsets of $\Sigma$ and that the number of their connected components
depends on the fact that the Freudenthal ends are not points of
$\Sigma$.

\section{Conclusions and further directions}
\label{sec:Conclusions}

We studied the first order IR approximants of hyperbolizable tame
two-field models, which are defined by the conditions that the target
surface $\Sigma$ is oriented and has finitely-generated fundamental
group, that the scalar field metric $\cG$ is hyperbolizable and that
the scalar potential $\Phi$ admits a strictly positive and smooth Morse
extension $\hPhi$ to the end compactification of $\Sigma$. In this
situation, the asymptotic form of the gradient flow orbits of the
uniformized effective scalar triple $(\Sigma,G,V)$ (which describe the
asymptotic behavior of the first IR approximant of the two-field
model parameterized by $(\Sigma,\cG,\Phi)$) can be determined
explicitly near each critical point of the classical effective
potential $V$ as well as near each end of $\Sigma$. We found that the
gradient flow of $(\Sigma,G,V)$ has exotic behavior near the ends,
which act like fictitious stationary points of this flow. Using
results from the theory of geometrically finite hyperbolic surfaces,
we showed that the IR behavior of the model near a critical point $c$
of the extended effective potential $\hV$ (which can be an interior
critical point or a critical end) is characterized by the critical
modulus $\beta_c\in [-1,1]\setminus\{0\}$ and by two sign
factors $\epsilon_1(c), \epsilon_2(c)\in \{-1,1\}$ which satisfy the
relation $\epsilon_1(c)\epsilon_2(c)=1$. When $c$ is a critical end,
this behavior also depends on the hyperbolic type of that end. For
critical ends, the definition of these quantities relies on the fact
that the hyperbolic metric $G$ admits an $\O(2)$ symmetry in a
vicinity of each end. For noncritical ends, we found that the
asymptotic behavior of effective gradient flow orbits depends only on
the hyperbolic type of the end. Non-critical flaring ends act like
fictitious but exotic stationary points of the effective gradient
flow even though they are not critical points of the extended
effective potential.

These results characterize the infrared behavior in each IR phase of
all tame two-field cosmological models up to first order in the
infrared expansion of \cite{ren} and hence open the way for systematic
studies of such models. We note that tame two-field models form an
extremely large class of cosmological models which was inaccessible
until now to systematic or conceptual analysis. With the exception of
the very special class of models discussed in \cite{Noether1,
Noether2} (which admit a `hidden' Noether symmetry and subsume all
previously considered integrable two-field models with canonical
kinetic term and canonical coupling to gravity), such models were
approached before only with numerical methods. Moreover, the vast
majority of work in this direction (with the exception of
\cite{genalpha,elem,modular}) was concerned exclusively with the
topologically trivial case of models whose target manifold is the
Poincar\'e disk \cite{KLR}. In view of \cite{ren} and of the results
of the present paper, models based on the Poincar\'e disk are
extremely far from capturing the infrared universality classes of tame
two-field models.

The results of this paper suggest various directions for further
research. As an immediate extension, one can study in more generality
the UV and IR behavior of two-field models whose target surface
$\Sigma$ is a disk, a punctured disk or an annulus. In this case the
uniformized scalar manifold $(\Sigma,G)$ is either an elementary
Euclidean surface or an elementary hyperbolic surface. In the second
situation, the universality classes are described by the UV or IR behavior
of the elementary two-field $\alpha$-attractor models considered in
\cite{elem}. The geodesic flow on elementary hyperbolic surfaces
(which describes the UV limit) is well-understood, while the
effective gradient flow can be studied for potentials $V$ which admit
a smooth extension $\hV$ to the Freudenthal compactification
$\hSigma\simeq \rS^2$ assuming that $\hV$ satisfies Morse-Bott
conditions \cite{Bott} on $\hSigma$. Similar questions can be asked
for $n$-field models whose target manifold is an elementary hyperbolic
space form \cite{Ratcliffe}.

One could also study the IR approximation of models for which
$(\Sigma,G)$ corresponds to a modular curve (such as the curve $Y(2)$
considered in \cite{modular}) with Morse-Bott conditions on the
extended potential. Using the uniformization theorem, such problems
can be reduced to the Poincar\'e disk, i.e. to studying the IR limit
of a modular cosmological model in the sense of \cite{Sch1, Sch2,
Sch3, Sch4}, though -- as pointed out in \cite{genalpha, modular} --
the quotient by the uniformization group can be highly nontrivial.

Finally, we mention that a characterization of the cosmological and
effective gradient flow near the ends of $\Sigma$ up to topological
equivalence can be extracted using the conformal compactification of
$(\Sigma,G)$ and the Vishik normal form \cite{Vishik} of vector fields
near the conformal boundary of $\Sigma$ and near its lift to
$T\Sigma$; we hope to report on this in a future publication.

\acknowledgments{\noindent This work was supported by grant PN
19060101/2019-2022. The authors thank the Simons Center for 
Geometry and Physics for hospitality. }

\appendix

\section{Details of computations for each case}

This appendix gives some details of the computation of cosmological
curves near critical points $c\in \hSigma$ of the extended
potential. We take $M_0=1$ as explained in the main text. In principal
Cartesian canonical coordinates near $c$, the first order
approximation of the effective potential is:
\be
V=V(c)+\frac{1}{2}\big(\lambda_1(c)x^2 + \lambda_2(c)y^2\big)~~,
\ee
with $V(c)$ a positive constant. We
take $V(c)=1$ and $\lambda_2(c)=1$, which gives
$\beta(c)\eqdef\frac{\lambda_1(c)}{\lambda_2(c)}=\lambda_1(c)$. In
polar principal canonical coordinates and with the assumptions
considered, we have:
\be
V=1+\frac{\omega^2}{2}\big(\beta(c)\cos^2 \theta + \sin^2 \theta\big)
\ee
In local coordinates on $\hSigma$, we have:
\be
\nabla_t \dot{\varphi}^i(t)\!=\!\ddot\varphi^i(t)+\Gamma^i_{jk}(\varphi(t))\dot\varphi^j(t)\dot\varphi^k(t)~,~~
||\dot\varphi(t)||^2_G\!=\!G_{ij}(\varphi(t))\dot\varphi^i(t)\dot\varphi^j(t)~,~~
\grad_G\Phi\!=\!G^{ij}\partial_i\Phi\partial_j~.
\ee
An infrared optimal cosmological curve is a solution $\varphi(t)$ of the
cosmological equation \eqref{eomsingle} which satisfies:
\be
{\dot \varphi}(0)=-(\grad_G V)(\varphi(0))~.
\ee

\subsection{Interior critical points}

\noindent In this case, $c$ is denoted by $\c$ and the hyperbolic
metric $G$ has the following form in semigeodesic coordinates
$(r,\theta)$ on the disk $D_{\omega_\rmax(\c)}$ (which is contained in
the Poincar\'e disk $\mD$):
\be
\dd s_G^2=\dd r^2+\sinh^2(r)\dd\theta^2~~,
\ee
where $r$ is related to $\omega\eqdef \sqrt{x^2+y^2}$ by
$\omega=\tanh(r/2)$ and we have $r\leq r(\c)$. Notice that
$(r,\theta)$ can be identified with the polar coordinates on the
tangent space $T_0\mD$ through the exponential map $\exp_0^\mD$ of the
Poincar\'e disk used in Subsection \ref{subsec:CanIntCrit}. The only
nontrivial Christoffel symbols are:
\be
\Gamma^r_{\theta\theta}=-\sinh(r)\cosh(r)~~,~~
\Gamma^\theta_{r\theta}=\Gamma^\theta_{\theta r}=\coth(r)~~.
\ee
The cosmological equations \eqref{eomsingle} become:
\beqa
&&\ddot r-\frac{1}{2}\sinh(2r)\dot\theta^2+\dot r\sqrt{\dot r^2+\sinh^2(r)\dot\theta^2+2\Phi}+\partial_r\Phi=0~~\\
&&\ddot\theta+2\coth(r)\dot r\dot\theta+\dot\theta\sqrt{\dot r^2+\sinh^2(r)\dot\theta^2+2\Phi}+\frac{1}{\sinh^2(r)}\partial_\theta\Phi=0~~,
\eeqa
which we solved numerically to obtain Figure \ref{fig:CosmIntCritical}.
The scalar potential is $\Phi(r,\theta)=\frac{1}{2}V^2(r,\theta)$,
where the classical effective potential takes the approximate form:
\be
V(r,\theta)=1+\frac{1}{2}\tanh^2(r/2)\big(\beta(\c)\cos^2(\theta)+\sin^2(\theta)\big)~.
\ee 
In principal polar canonical coordinates $(\omega,\theta)$ centered at
$\c$, we have:
\be
\dd s^2_G=\frac{4}{(1-\omega^2)^2}[\dd \omega^2+\omega^2\dd \theta^2]~~
\ee
and:
\be
V(\omega,\theta)= 1+\frac{1}{2}\omega^2\left[\lambda_1(\c) \cos^2\theta +\lambda_2(\c) \sin^2\theta\right]~~.
\ee
Thus:
\beqa
&& H(\omega,\theta,\dot{\omega},\dot{\theta})=\sqrt{\frac{4}{(1-\omega^2)^2}\big(\dot\omega^2+\omega^2\dot\theta^2 \big)+2\Phi(\omega,\theta)} ~~\\
&&  \Gamma^\omega_{\omega\omega}=\frac{2\omega}{1-\omega^2}~~,~~\Gamma^\omega_{\theta\theta}=-\omega^2\big( \frac{1}{\omega} + \frac{2\omega}{1-\omega^2}\big)~~,~~
\Gamma^\theta_{\omega\theta}=\Gamma^\theta_{\theta\omega}= \frac{1}{\omega} + \frac{2\omega}{1-\omega^2}~~\\
&& (\grad \Phi)^\omega\!\approx\!\frac{(1-\omega^2)^2}{4} \pd_\omega \Phi~~,~~
 (\grad \Phi)^\theta\!\approx\!\frac{(1-\omega^2)^2}{4\omega^2} \pd_\theta \Phi
\eeqa
The cosmological equations become:
\beqa
&&\ddot\omega+\Gamma^\omega_{\omega\omega}\dot\omega^2+\Gamma^\omega_{\theta\theta}\dot\theta^2 + H\dot\omega+(\grad\Phi)^\omega=0~,\\
&&\ddot\theta+2\Gamma^\theta_{\omega\theta}\dot\omega\dot\theta+H\dot\theta+
(\grad\Phi)^\theta=0~~.
\eeqa

\subsection{Critical and noncritical ends}

\noindent Recall the asymptotic form \eqref{eomegametric} near the end
$\e$ in principal polar canonical coordinates centered at $\e$:
\be
\dd s_G^2|_{\dot{U}_\e}=\frac{\dd \omega^2}{\omega^4}+f_\e(1/\omega)\dd \theta^2~~,
\ee
with:
\be
f_\e(1/\omega)= {\tilde c}_\e e^{\frac{2\epsilon_\e}{\omega}}\left[1+\O\left(e^{-\frac{2}{\omega}}\right)\right]~~\mathrm{for}~~\omega\ll 1~~,
\ee
where:
\be
{\tilde c}_\e=\fourpartdef{\frac{1}{4}}{~\e=\mathrm{plane~end}}
{\frac{1}{(2\pi)^2}}{~\e=\mathrm{horn~end}}
{\frac{\ell^2}{(4\pi)^2}}{~\e=\mathrm{funnel~end~of~circumference}~\ell>0}{\frac{1}{(2\pi)^2}}{~\e=\mathrm{cusp~end}}
\ee
and:
\be
\epsilon_\e=\twopartdef{+1}{~\e=\mathrm{flaring~(i.e.~plane,~horn~or~funnel)~end}}{-1}{~\e=\mathrm{cusp~end}}~~.
\ee
The term $\O\left(e^{-\frac{2}{\omega}}\right)$
vanishes identically when $\e$ is a cusp or horn end, but we will approximate it to zero for all ends. 
The only nontrivial Christoffel symbols are:
\be
\Gamma^\omega_{\omega\omega}=-\frac{2}{\omega}~~,~~
\Gamma^\omega_{\theta\theta}=\tilde c_\e\epsilon_\e\omega^2\e^\frac{2\epsilon_\e}{\omega}~~,~~
\Gamma^\theta_{\theta\omega}=\Gamma^\theta_{\omega\theta}=-\frac{\epsilon_\e}{\omega^2}~~.
\ee
The cosmological equations \eqref{eomsingle} become:
\beqan
\label{cosmeqends}
&&\ddot\omega-\frac{2}{\omega}\dot\omega^2+\tilde c_\e\epsilon_\e\omega^2 e^\frac{2\epsilon_\e}{\omega}\dot\theta^2+H\dot\omega+\omega^4\partial_\omega\Phi=0~~,\nn\\
&&\ddot\theta-\frac{2\epsilon_\e}{\omega^2}\dot\omega\dot\theta +H\dot\theta +
\frac{1}{\tilde c_\e}e^{-\frac{2\epsilon_\e}{\omega}}\partial_\theta\Phi=0~~,
\eeqan
where 
\be
H=\sqrt{ \frac{\dot\omega^2}{\omega^4}+
\tilde c_\e e^\frac{2\epsilon_\e}{\omega} \dot\theta^2+2\Phi}~~.
\ee
The difference between the critical and noncritical ends manifests in
the form of the second order approximations for the potential:
\begin{description}
\item $\bullet$ for {\bf critical ends}:
\be
V(\omega,\theta)=1 +\frac{1}{2}\omega^2\big(\beta(\e)\cos^2 \theta + \sin^2 \theta \big)~,
\ee
\item $\bullet$ for the {\bf noncritical ends}:
\be
V(\omega,\theta)=1+\mu\omega\cos\theta~,
 \ee
where $\mu_\e$ is a positive constant which we chose to be $1/2$ in
our graphs.
\end{description}

\noindent Figures \ref{fig:NoncritCosm},
\ref{fig:NoncritCosmCuspDetail}, and
\ref{fig:CritCosmPlane}-\ref{fig:CritCosmCuspPlusSingle} were obtained
by solving \eqref{cosmeqends} numerically in each case.

\end{document}